\documentclass[aps,pra,superscriptaddress,floatfix,notitlepage,nofootinbib,longbibliography,11pt]{revtex4-1} 
\pdfoutput=1
\usepackage{amsmath,amsthm,amsfonts,amssymb,braket}
\usepackage{graphicx}
\usepackage{color}

\numberwithin{figure}{section}

\makeatletter
\def\l@subsubsection#1#2{}
\makeatother

\usepackage[colorlinks=true,citecolor=blue,linkcolor=blue]{hyperref}
\usepackage[capitalize,nameinlink]{cleveref}

\usepackage[all]{xy}

\usepackage{cancel}

\newtheorem{lemma}{Lemma}[section]

\newtheorem{corollary}[lemma]{Corollary}
\newtheorem{proposition}[lemma]{Proposition}

\theoremstyle{definition}
\newtheorem{definition}[lemma]{Definition}
\newtheorem{remark}[lemma]{Remark}

\newcommand{\CC}{{\mathbb{C}}}
\newcommand{\RR}{{\mathbb{R}}}
\newcommand{\ZZ}{{\mathbb{Z}}}
\newcommand{\Z}{{\mathbb{Z}}}

\newcommand{\diag}{{\mathrm{diag}}}

\DeclareMathOperator{\Supp}{\mathrm{Supp}}

\DeclareMathOperator{\Lnk}{Lnk}
\DeclareMathOperator{\Int}{Int}
\DeclareMathOperator{\Frm}{Frm}
\DeclareMathOperator{\BiFrm}{BiFrm}

\newcommand{\stab}{\mathcal{S}}

\newcommand{\cnf}{{\mathfrak{c}}}

\newcommand{\bM}{{\mathbf M}}

\newcommand{\tc}{{\tilde{\cnf}}}

\newcommand{\bd}{{\partial}} 
\newcommand{\db}{{\delta}} 

\begin{document}

\title{Gravitational anomaly of $3+1$ dimensional $\ZZ_2$ toric code with fermionic charges and fermionic loop self-statistics}

\author{Lukasz Fidkowski}
\affiliation{Department of Physics, University of Washington, Seattle, Washington, USA}

\author{Jeongwan Haah}
\affiliation{Microsoft Quantum, Redmond, Washington, USA}

\author{Matthew B.~Hastings}
\affiliation{Station Q, Microsoft Research, Santa Barbara, California, USA}
\affiliation{Microsoft Quantum, Redmond, Washington, USA}

\begin{abstract}
Quasiparticle excitations in $3+1$ dimensions can be either bosons or fermions.
In this work, we introduce the notion of fermionic {\emph{loop}} excitations in $3+1$ dimensional topological phases.
Specifically, we construct a new many-body lattice invariant of gapped Hamiltonians, 
the loop self-statistics $\mu=\pm1$, 
that distinguishes two bosonic topological orders that both superficially resemble $3+1$d $\mathbb Z_2$ gauge theory
coupled to fermionic charged matter.  
The first has fermionic charges and bosonic $\mathbb Z_2$ gauge flux loops~(FcBl) 
and is just the ordinary fermionic toric code.  
The second has fermionic charges and fermionic loops~(FcFl) 
and, as we argue, can only exist at the boundary of a non-trivial $4+1$d invertible phase, 
stable without any symmetries 
{\em i.e.}, it possesses a gravitational anomaly.
We substantiate these claims by constructing an explicit exactly solvable $4+1$d model using a method that bootstraps a boundary theory into a bulk Hamiltonian, analogous to that of Walker and Wang,
and computing the loop self-statistics in the fermionic $\mathbb Z_2$ gauge theory hosted at its boundary.  
We also show that the FcFl phase has the same gravitational anomaly as all-fermion quantum electrodynamics.  
Our results are in agreement with the recent classification of 
nondegenerate braided fusion $2$-categories by Johnson-Freyd, 
and with the cobordism prediction of a non-trivial $\mathbb Z_2$-classified $4+1$d invertible phase 
with action $S=\tfrac{1}{2}\int w_2 w_3$.
\end{abstract}

\maketitle
\tableofcontents

\section{Introduction}

In the past few decades it has been realized that sophisticated mathematical tools and structures can be applied to understand the classification of gapped many-body quantum phases.  This includes, for example, unitary modular tensor categories, which, having already had close connections to quantum field theory \cite{Witten1, Witten2, MS, Segal, RT, Drinfeld}, found an application in the classification of $2+1$-dimensional topological orders~\cite{FRS89, frohlich1990braid, Kitaev_2005}, as well as cobordism groups, which turn out to be useful in classifying invertible and symmetry protected topological (SPT) phases in arbitrary dimensions~\cite{Kapustin, Freed, Freed_Hopkins}.  More recently, braided fusion $2$-categories have been introduced in order to classify topological orders in $3+1$-dimensions \cite{KW14, LKW18, KTZ20}.  Validating any such mathematical classification scheme requires constructing a physical observable 
--- {\em i.e.}, a quantized invariant of many-body lattice Hamiltonians --- 
that distinguishes among the proposed phases.
While in some cases this is (in principle) straightforward 
--- {\it e.g.,} using interferometry to measure the mutual braiding statistics of anyons ---
in other cases it is more complicated.  
One example of the latter situation is the recent prediction~\cite{Freyd_2020}, 
based on a classification of nondegenerate braided fusion $2$-categories, 
of two distinct variants of $3+1$d $\ZZ_2$ gauge theory coupled to fermions:
the ordinary fermionic toric code, 
and an anomalous variant that can only exist at the boundary of a non-trivial $4+1$d invertible phase.
Assuming these phases do indeed exist, one can ask: what physical observable, defined in the context of gapped lattice spin Hamiltonians, distinguishes between them?  Furthermore, are there exactly solved models that realize these two phases?

We answer these questions by defining a new many-body lattice invariant of bosonic $3+1$d gapped Hamiltonians,
the loop self statistics $\mu=\pm 1$, and constructing models that realize both of these values.  
Importantly, the loop self statistics is well defined for $\ZZ_2$ gauge theory topological orders 
if and only if the gauge charge is a fermion.
In this case there are two possibilities: 
(i)~$\mu=1$, the usual fermionic toric code with fermionic charges and bosonic loops~(FcBl) 
and 
(ii)~$\mu=-1$, the anomalous variant with fermionic charges and fermionic loops~(FcFl).
The definition of~$\mu$ is reminiscent of the $T$-junction process 
(reviewed in \cref{subsec:review}) 
used to measure exchange statistics of identical anyons by applying a product of string operators 
that exchanges the anyons in such a way as to carefully cancel all non-universal phases (see \cite{LevinWen2003Fermions} and Sec $8.3$ and figure $10$ in \cite{Kitaev_2005}).
In the present case of loops in $3+1$d, 
the $T$-junction is replaced by a more complicated geometry, 
illustrated in \cref{fig:all_configurations_figure}, 
and the string operators are replaced with $2$d membrane operators.  
The process ends up effectively rotating the initial loop configuration 
in such a way as to reverse the orientation along the loop, 
as illustrated in \cref{fig:invariant_figure}.

As already alluded to, 
the loop self statistics are also an anomaly indicator: 
the FcFl phase cannot exist in a standalone $3+1$d lattice model, 
but only at the boundary of a non-trivial $4+1$d invertible phase, 
and hence possesses a gravitational anomaly.  
At an intuitive level this is because the fermionic nature of both the charges and the loops 
in the FcFl phase prevents either one from being condensed, 
making it impossible to drive a phase transition to a trivial phase, 
and hence leading to it being anomalous.
More formally, if the FcFl phase could be realized strictly in $3+1$d, 
then, by ``un-gauging'' fermion parity, we would obtain a non-trivial $3+1$d fermionic invertible phase, 
which is absent in current classification schemes and believed not to exist.

The anomalous nature of the FcFl phase is also related to its connection to 
all-fermion quantum electrodynamics (QED), {\em i.e.}, 
$3+1$d QED with fermionic charges, monopoles, and dyons \cite{Wang_2014, Swingle_2015}.
To elucidate this connection, 
let us imagine condensing pairs of fermionic charges in all-fermion QED.
The resulting Meissner effect confines the magnetic field in $\pm \pi$ flux tubes, 
with the domain wall between $\pi$ and $-\pi$ flux trapping a neutral fermionic monopole.
This intuitive picture motivates the construction of an exactly solved $4+1$d model using a method similar to that of Walker and Wang \cite{Walker_2011}, where a boundary theory is bootstrapped into a bulk Hamiltonian in one dimension higher.  Specifically, at a continuum level we view a spacetime trajectory of loops and particles in the anomalous $3+1$d theory as a configuration in $4$-dimensional space, and assign to it an amplitude equal to its exponentiated action.  

The fermionic nature of the charge worldlines leads us to pick a ``blackboard'' framing, 
as is typical in Walker--Wang constructions~\cite{BCFV}.  
This same blackboard framing can be used to put a local orientation 
on the (possibly not globally orientable) 2d loop trajectory worldsheet.  
The key feature which then makes our exactly solved model non-trivial is that 
the 1d domain wall where this local orientation reverses 
(referred to as the ``$w_1$ line'' below) 
is decorated with with an additional {\emph{gauge neutral}} fermion 
--- the remnant of the fermionic monopole.%
\footnote{
	Note that binding an additional physical gauge charged fermion to this $w_1$ line results in effectively a gauge charged boson being bound to it.  This gives another, equivalent, Walker--Wang model, where the $w_1$ line decoration is by a gauge charged boson.  The key point is that it is impossible to get rid of both the gauge charge and the statistics of the $w_1$ line decoration 
when the gauge charge is a fermion.  When the gauge charge is a boson, on the other hand, the decoration can be screened out, showing, at a heuristic level, why the loop self-statistics is not well defined in that case.
} The choice of blackboard framing in our model should be viewed as a technical tool used to obtain a lattice model whose boundary excitations have the appropriate statistics, verified through computing the appropriate commutation relations of string or world-sheet operators.  In particular, the notion of framing never needs to be explicitly referred to subsequently, when studying the lattice models.

Our model is specifically designed so that,  when truncated, its boundary hosts the FcFl fermionic toric code topological order, as we explicitly verify.  Furthermore, we explicitly show that the $4+1$d bulk is invertible --- in this case, it means that the bulk can be disentangled by a shallow depth circuit for two stacked copies of this phase.  In particular, this implies that there is no topological order in the bulk.  In view of the generalized Walker--Wang prescription guiding the construction of our model, this is a reflection of the nondegenerate nature of the braided fusion $2$-categories describing our boundary.  In fact, our model can be interpreted as a gauge theory involving $2$-form and $3$-form gauge fields whose action, valued in a certain cohomology group, contains the data of the braided fusion $2$-category.  Further discussion of this connection and the work of \cite{Freyd_2020} in particular is given in \cref{sec:discussion} below.  The boundary of our model can also be driven into the all-fermion QED phase (after the addition of some ancilla boundary degrees of freedom), confirming that all-fermion QED and the FcFl phase indeed possess the same gravitational anomaly.

There have been several previous works related to ours.  
The anomalous nature of all-fermion QED has been studied in a continuum field theory context 
in~\cite{Swingle_2015} and~\cite{Wang_2019}, 
where the anomaly is diagnosed by putting the theory on a ${\mathbb{CP}}^2$ spacetime topology.
The anomalous nature of all-fermion QED was also studied in~\cite{Wang_2014}, 
where a proof by contradiction exploited the edge-ability of any standalone $3+1$d model,
and relied on an assumption of the existence of a gapped surface topological order 
for any invertible phase of fermions.
In~\cite{Freyd_2020} the existence of two distinct variants of fermionic $\ZZ_2$ gauge theory 
was posited based on the classification of nondegenerate braided fusion $2$-categories.  
The fact that our loop self-statistics $\mu$ are well defined only in the case of fermionic gauge charges 
turns out to be a reflection of the trivialness 
of a certain automorphism of the corresponding braided fusion $2$-category in the work of \cite{Freyd_2020}.  
The notion of a fermionic loop excitation was also introduced in a field theory context 
in~\cite{Thorngren_2015}, 
although the relation between this and our loop self-statistics is not completely clear.

The rest of this paper is structured as follows.
In~\cref{sec:indicator} we construct the loop exchange statistics $\mu$ in fermionic gauge theories,
and show that it is independent of the various arbitrary choices made in the construction, 
assuming that the gauge charge is a fermion.
In~\cref{sec:WW}, motivated by a continuum intuition coming from a decorated domain wall picture, 
we construct a $4+1$d lattice model and verify in~\cref{sec:bdtheory} 
that it hosts the FcFl phase on its $3+1$d boundary.
We conclude with some remarks about the connection of our work to the classification of braided fusion $2$-categories and the work of \cite{Freyd_2020}, as well as some future directions in \cref{sec:discussion}.
  
{\bf{Note added:}} Near the completion of this work, we learned of another paper \cite{YuAn2021} in preparation that also constructs an exactly solvable lattice model for the nontrivial invertible bosonic phase in $4+1$ dimensions.  Also, after the completion of the initial draft of this work we learned about mathematical work of Johnson-Freyd and Reutter \cite{Freyd_Reutter} where a so-called `Klein' invariant is defined.  We believe that this Klein invariant should correspond to our loop self-statistics, modulo the fact that one is defined in the continuum field theory and the other for lattice many-body quantum systems.



\section{Fermionic loop self-statistics}\label{sec:indicator}

\subsection{Review of exchange statistics of identical point particles} \label{subsec:review}

Our loop self-statistics will be defined in analogy with the process 
that measures exchange statistics of identical 
(and for simplicity abelian) quasiparticle excitations (see \cite{LevinWen2003Fermions} and Sec $8.3$ and figure $10$ in \cite{Kitaev_2005}), 
so let us first review this process in a way that will naturally generalize.
We have a T-junction geometry, as in \cref{fig:T_junction}, 
with all distances much longer than the correlation length.  
We choose $6$ different states $|\cnf_i\rangle$, $i=1,\ldots,6$, 
corresponding to the ${4 \choose 2}=6$ configurations of two identical quasiparticles 
illustrated in \cref{fig:T_junction}.  
These states have the property that if $|\cnf_i\rangle$ and $|\cnf_j\rangle$ 
both have a given location occupied by a quasiparticle,
or both have it unoccupied, then the reduced density matrices of $|\cnf_i\rangle$ and $|\cnf_j\rangle$ 
in the neighborhood of that location are identical 
(in the latter case, being just the ground state reduced density matrix).
We then choose string operators~$M_i$, $i=1,2,3$, 
which move a quasiparticle from the center out to one of the three outer endpoints.
We require that the $M_i$ be shallow circuits 
supported in thin neighborhoods of the intervals connecting the center to these outer endpoints, 
with Lieb--Robinson length much smaller than the lengths of the intervals 
(but possibly larger than the correlation length).
We then compute
\begin{align}
M_2 M_3^{-1} M_1 M_2^{-1} M_3 M_1^{-1} \ket{\cnf_1} = \theta \ket{\cnf_1} \label{eq:exchangestatistics}
\end{align}
where $\theta$ encodes the exchange statistics of the quasiparticles.

\begin{figure}[b]
\centering
\includegraphics[width=0.7\textwidth, trim={20mm 70mm 90mm 70mm}, clip]{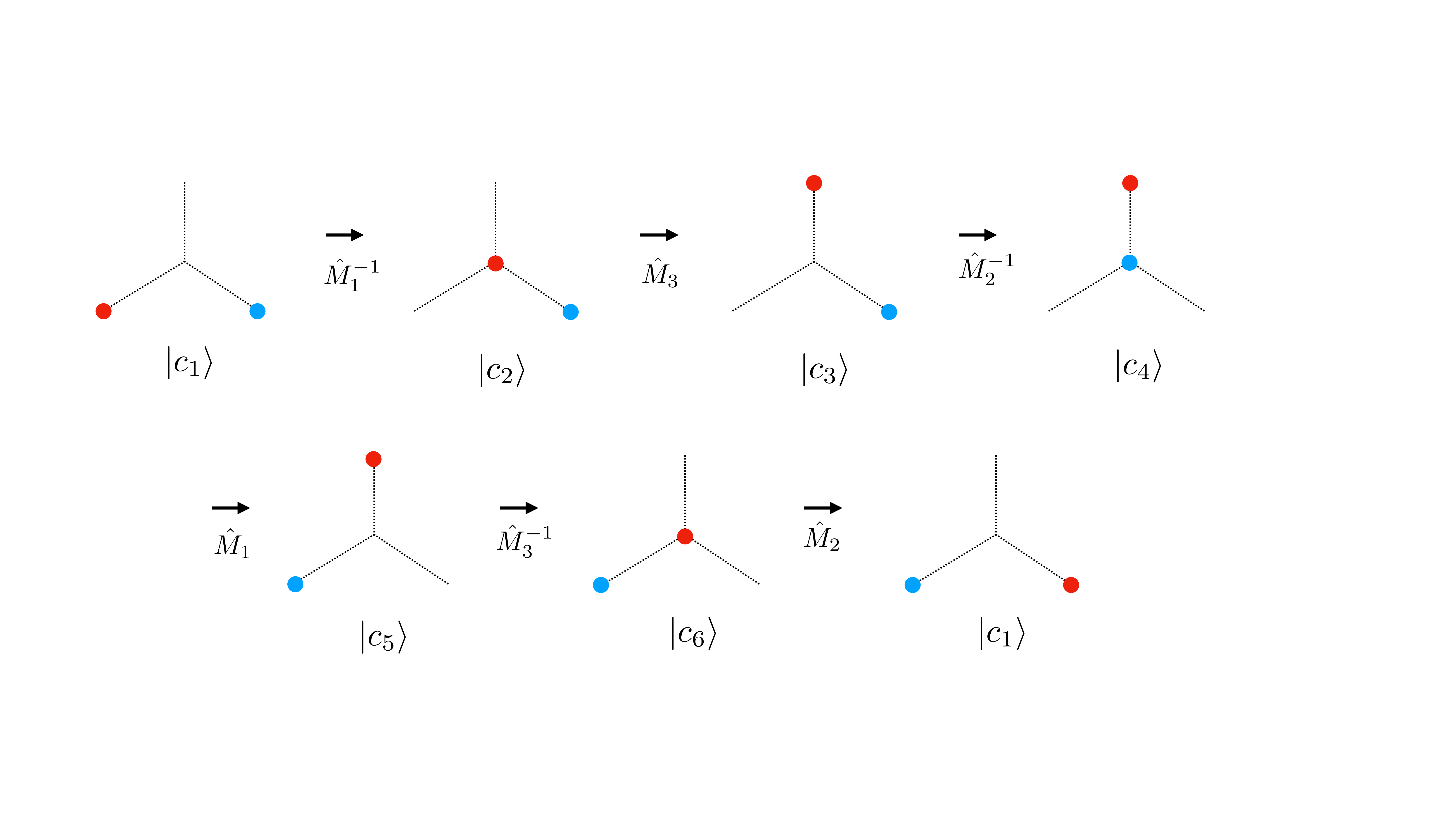}
\caption{T-junction process used to measure statistics of identical particles}
\label{fig:T_junction}
\end{figure}

To argue that $\theta$ is well defined,
one must show that 
(i) for a given choice of the states $\{ |\cnf_i\rangle \}$, 
the phase factor~$\theta$ is independent of the choice of the $M_i$, and 
(ii) $\theta$ is independent of the choice of $\{ |\cnf_i\rangle \}$.
By ``choice of $\{ |\cnf_i\rangle \}$'' we mean a potentially different set of states $\{ |\cnf_i'\rangle \}$ 
where all of the $|\cnf_i'\rangle$ have the same topological charges as the $|\cnf_i\rangle$ locally, 
but may differ by some topologically trivial excitations; 
for example, the locations of the quasiparticles may have moved slightly.

To prove the first statement, let us take, for a fixed set of $\{ |\cnf_i\rangle \}$, 
a different choice of string operators $M_i'$.  
Then we must have $M_i' = F_i M_i$, where $F_i$ is a shallow circuit.
$F_1$ has the property that $F_1|\cnf_1\rangle = \alpha |\cnf_1\rangle$, $F_1|\cnf_5\rangle = \alpha' |\cnf_5\rangle$.
Furthermore, $|\cnf_5\rangle = U|\cnf_1\rangle$ 
where $U$ is a shallow circuit supported away from the support of $F_1$.
Being supported on disjoint spatial regions, 
$U$ and $F_1$ commute, so that $\alpha = \alpha'$.  
A similar argument applies to $F_2$ and $F_3$, so that:
\begin{align*}
F_1|\cnf_1\rangle &= \alpha|\cnf_1\rangle, \quad
F_1|\cnf_5\rangle =\alpha|\cnf_5\rangle \\
F_2|\cnf_3\rangle &= \beta|\cnf_3\rangle,\quad
F_2|\cnf_1\rangle = \beta|\cnf_1\rangle \\
F_3|\cnf_5\rangle &= \gamma|\cnf_5\rangle,\quad
F_3|\cnf_3\rangle =\gamma|\cnf_3\rangle
\end{align*}
Thus,
\begin{align*}
M'_2 &\left({M'}_3\right)^{-1} M'_1 \left({M'}_2\right)^{-1} M'_3 \left({M'}_1\right)^{-1}|\cnf_1\rangle \\
&= F_2 M_2 M_3^{-1} F_3^{-1} F_1 M_1 M_2^{-1}F_2^{-1} F_3 M_3 M_1^{-1} F_1^{-1} |\cnf_1\rangle \\
&= \beta M_2 M_3^{-1} \gamma^{-1}\alpha M_1 M_2^{-1} \beta^{-1}\gamma M_3 M_1^{-1} \alpha^{-1} |\cnf_1\rangle \\
&= M_2 M_3^{-1} M_1 M_2^{-1} M_3 M_1^{-1}|\cnf_1\rangle
\end{align*}
so we get the same value of $\theta$.

To prove the second statement, 
suppose we have a different set of configuration states $\{ |\cnf_i'\rangle\}$. 
Then clearly there exists a shallow circuit~$V$ 
such that $|\cnf_i'\rangle = V |\cnf_i\rangle$ 
(just take one that moves the quasiparticles from their old positions to their new positions);
conjugating~$M_i$ by~$V$ gives a set of string operators for $\{ |\cnf_i'\rangle\}$,
and the composition of these new string operators used in computing~$\theta$
is simply the $V$ conjugate of the old composition (note that conjugation commutes with taking inverses).
Thus the exchange phase stays the same.

We will now use a similar procedure to define self-exchange statistics of loop excitations 
in $3$~spatial dimensions, 
when the gauge charge is a fermion.

\subsection{Data used to define the loop self statistics in $3+1$d} \label{subsec:data}

\begin{figure}
\centering
\includegraphics[width=6.5in]{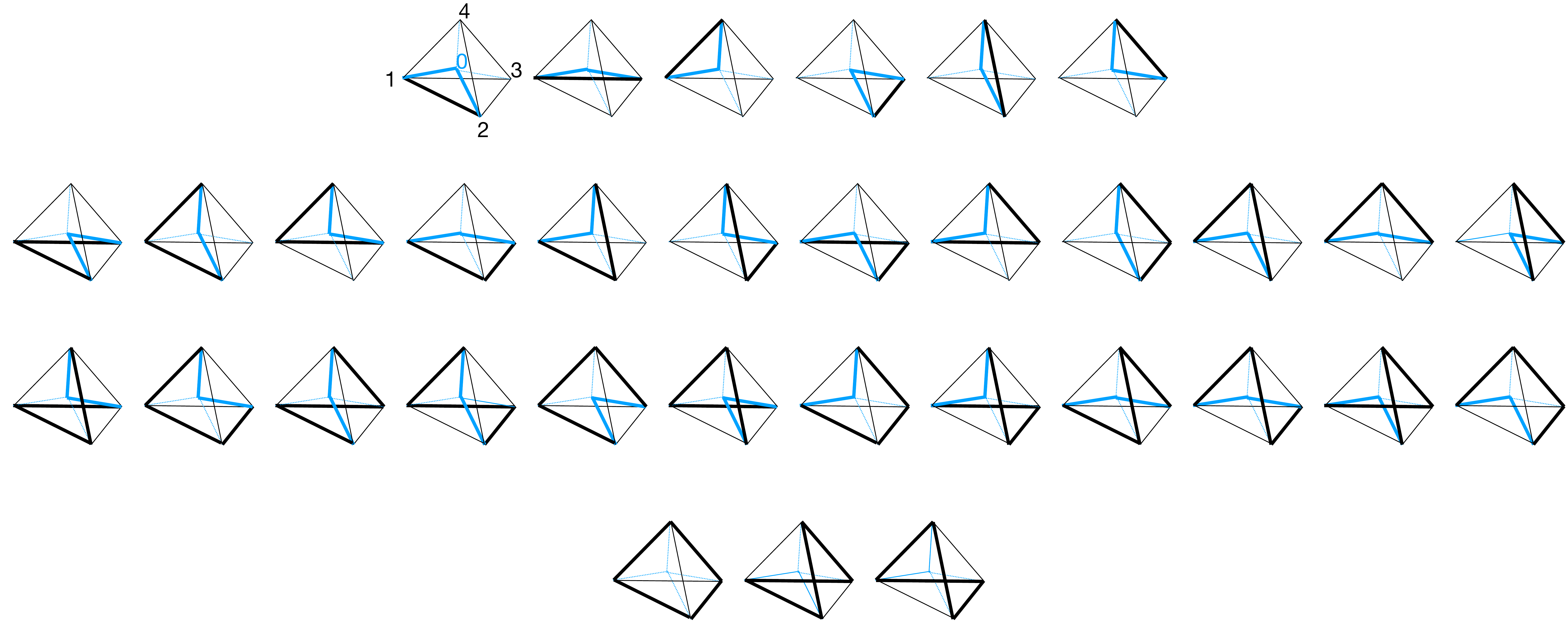}
\caption{
	All $33$ loop configurations.  
	The length scale is much longer than the correlation length.  
	Whenever two of these configurations look identical in some local region, 
	we require that their reduced density matrices in this local region be identical.  
	This in particular means that there are $6 =  {4 \choose 2}$ different reduced density matrices 
	in the neighborhood of each vertex. 
	We have colored the interior edges blue and the outside edges black for clarity.  
	The bottom three configurations appear twice in the sequence defining our invariant; 
	all other configurations appear exactly once.
}
\label{fig:all_configurations_figure}
\end{figure}

Before delving into the details of the process used to define the loop self statistics, let us make some general remarks.  The process will move a loop excitation in such a way that its final position is the same as its initial position, but the orientation along the loop is reversed.  This means that the spacetime history of this process, with periodic boundary conditions identifying the initial and final state of the loop, is a Klein bottle.  We believe that the resulting invariant is the same as the Klein invariant of Freyd and Reutter, defined in section 3.3 or \cite{Freyd_Reutter}.  We emphasize though that the key feature of the process defined below is that it defines a manifestly universal quantity in the lattice quantum many-body system, independent of arbitrary choices of the string movement operators.  This is the reason for the large number of seemingly un-motivated steps in the process.

Our geometry is now a tetrahedron with vertices $1,2,3,4$, together with a central vertex $0$.
Again, all length scales are much longer than the correlation length.
Consider $33$ loop configurations~$\cnf$ illustrated in \cref{fig:all_configurations_figure}.
The first piece of data we will need for defining our loop self-statistics 
is a corresponding set of $33$ states~$|\cnf\rangle$, 
where occupied edges of~$\cnf$ form a $\ZZ_2$ gauge flux loop.
We demand the following property of the $\{|\cnf\rangle \}$: 
if two configurations~$\cnf$ and~$\cnf'$ look the same locally, 
then the reduced density matrices of $|\cnf\rangle$ and $|{\cnf'}\rangle$ in that local region are identical.
More precisely, if $\cnf$ and $\cnf'$ both have the same edge occupied, or both have it unoccupied, 
then the reduced density matrices of $|\cnf\rangle$ and $|{\cnf'}\rangle$ 
in a neighborhood of the interior of that edge (not including its endpoints) are identical.
Also, if $\cnf$ and $\cnf'$ have the same two edges adjoining a given vertex occupied, 
or both have all edges adjoining that vertex unoccupied, 
then the reduced density matrices of $|\cnf\rangle$ and $|{\cnf'}\rangle$ 
in a neighborhood of that vertex are identical 
(in the latter case, being just the ground state density matrix).


To produce such states $|\cnf\rangle$, 
we can act with membrane operators on various plaquettes $(ij0)$ 
to produce the desired configuration of $\ZZ_2$ gauge flux loops, 
then act locally in the neighborhoods of the various edges $(i0)$ to remove local excitations, 
and then finally do the same in the neighborhoods of the various vertices.  
There is a possible obstruction that may potentially arise in this last step: 
there may be extra $\ZZ_2$ gauge charges stuck at the vertices, 
relative to the desired configuration.  
We will assume that these can always be removed by acting with gauge charge string operators along the various edges.
It is possible that such an obstruction never arises, 
but we do not prove this here.  
We note that this obstruction certainly does not arise in the two universality classes we are interested in in this paper,
namely that of the trivial $3+1$d fermionic toric code (FcBl), and the anomalous one (FcFl), 
as can be explicitly seen from the exactly solved model presented below.

The other piece of data we need for defining our loop self-statistics 
are the membrane operators~$M_{ij}$ ($i<j$), 
which nucleate a $\ZZ_2$ gauge flux loop around the plaquette $(ij0)$.  
We impose the following conditions on $M_{ij}$.  
First, we demand that $M_{ij}$ is a shallow circuit of local unitaries, 
with Lieb--Robinson length possibly longer than the correlation length 
but much shorter than the length scales associated with our tetrahedron. 
It could also be a circuit with tails, {\em i.e.}, a short time evolution of a quasi-local time dependent Hamiltonian.
Second, suppose that $\cnf$ is one of our $33$ configurations, 
with edge $(ij)$ unoccoupied, and that $\cnf'$ differs from $\cnf$ 
precisely in the occupation numbers of the three edges $(ij)$,$(i0)$,$(j0)$.  
Then, if $\cnf'$ is also one of our $33$ allowed configurations, we demand that
\begin{align} 
|\cnf'\rangle \langle \cnf'| = M_{ij} |\cnf\rangle \langle \cnf| M_{ij}^{-1}  \label{eq:Mcondition}
\end{align}
A construction of a particular set of $M_{ij}$ satisfying these conditions 
is illustrated in \cref{fig:F_move_figure2}.

\begin{figure}
\centering
\includegraphics[width=4.0in]{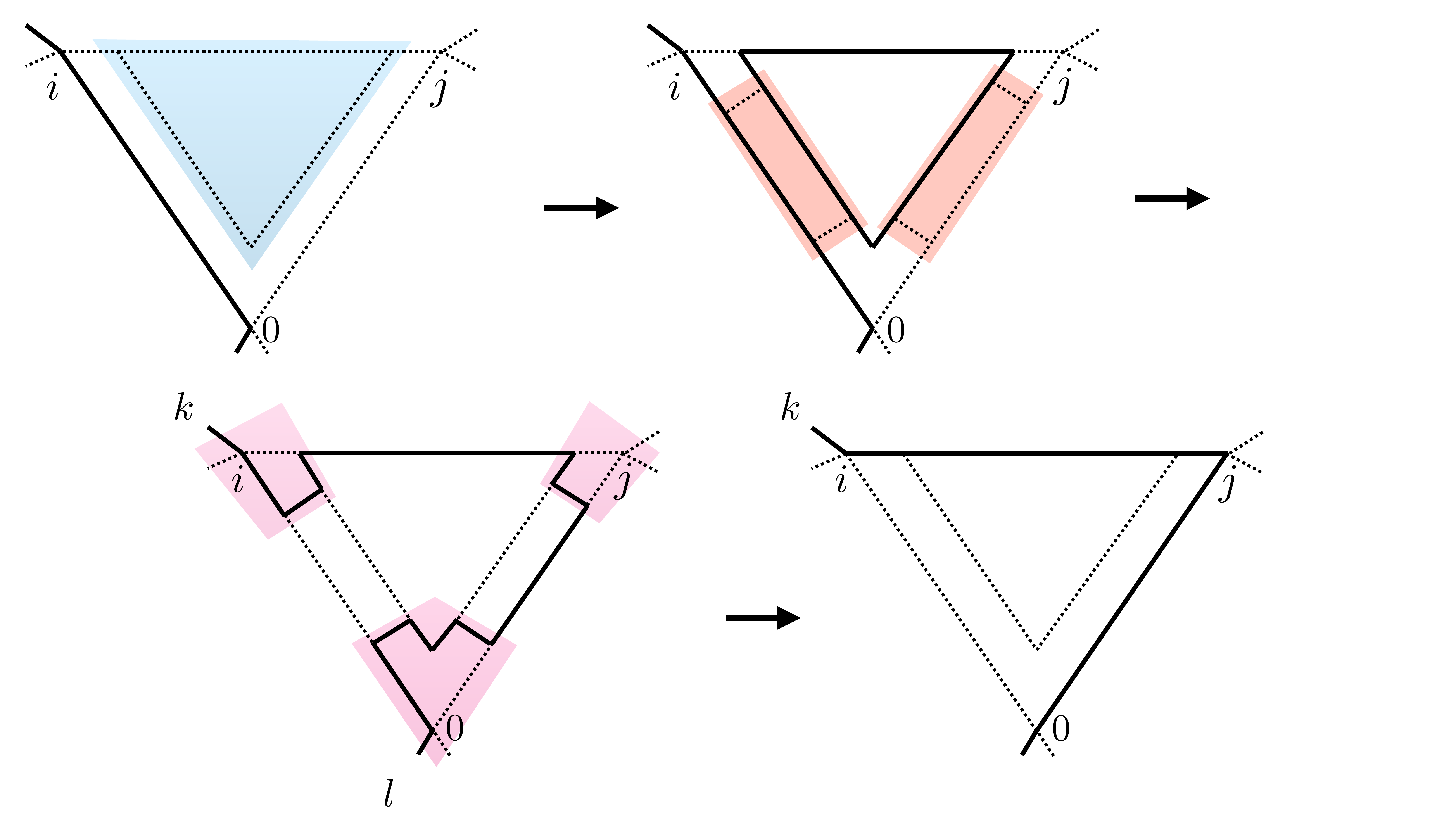}
\caption{
	A construction of an operator $M_{ij}$ satisfying the requisite conditions, 
	described in \cref{subsec:data}.  
	In this construction, $M_{ij}$ is obtained as the composition of a face operator (blue),
	which nucleates a loop of gauge flux in the interior of the plaquette, 
	with controlled shallow circuits, 
	which splice this loop into the perimeter of the plaquette.
	The edge operators (red) are controlled by the occupation numbers of the corresponding edges, 
	and the vertex operators (purple) are controlled by the local configurations near the corresponding vertices.  
	All distances in the figure are much greater than the correlation length.  
	This process is designed in such a way that \cref{eq:Mcondition} is satisfied, 
	as can be checked by comparing the local reduced density matrices 
	on both sides of this \cref{eq:Mcondition}.
}
\label{fig:F_move_figure2}
\end{figure}

\subsection{Definition of the loop self statistics $\mu$} \label{subsec:indicator}

The loop statistics are defined using the process illustrated in \cref{fig:invariant_figure}.
There is some arbitrariness in the choice of initial configuration, 
but, for concreteness, we begin with a starting configuration $\cnf_1$, 
consisting of the edges $(14),(10),(40)$ being occupied 
(the top left configuration in \cref{fig:invariant_figure}).
Then, we apply to $|\cnf_1\rangle$ a sequence of $36$ membrane operators or their inverses.
Specifically, in each step, one of two possibilities occurs: 
1) an edge $(ij)$ changes from being unoccupied to being occupied or 
2) an edge $(ij)$ changes from being occupied to being unoccupied.
We act with $M_{ij}$ and $M_{ij}^{-1}$ in these two cases respectively.
In this way we generate a sequence of states~$|\cnf_j\rangle$, $j=1,2,\ldots, 37$, with~$\cnf_{37}=\cnf_1$, 
{\em i.e.},
the configuration comes back to itself at the end.  
However, the orientation along the loop reverses at the end of the process, 
as illustrated in \cref{fig:invariant_figure}.
Note that since for each step in the process 
there is always an occupied edge $e$ which is not touched by the membrane operator acting at that step, 
we can consistently define what it means for the orientation to not change during a step 
by requiring that the orientation along $e$ be fixed.
Note also that this orientation is not in any way a physical observable, 
{\em i.e.,} $|\cnf_1\rangle = |\cnf_{37}\rangle$.

\begin{figure}[b]
\centering
\includegraphics[width=6.5in, trim={0ex 50ex 0ex 0ex}, clip]{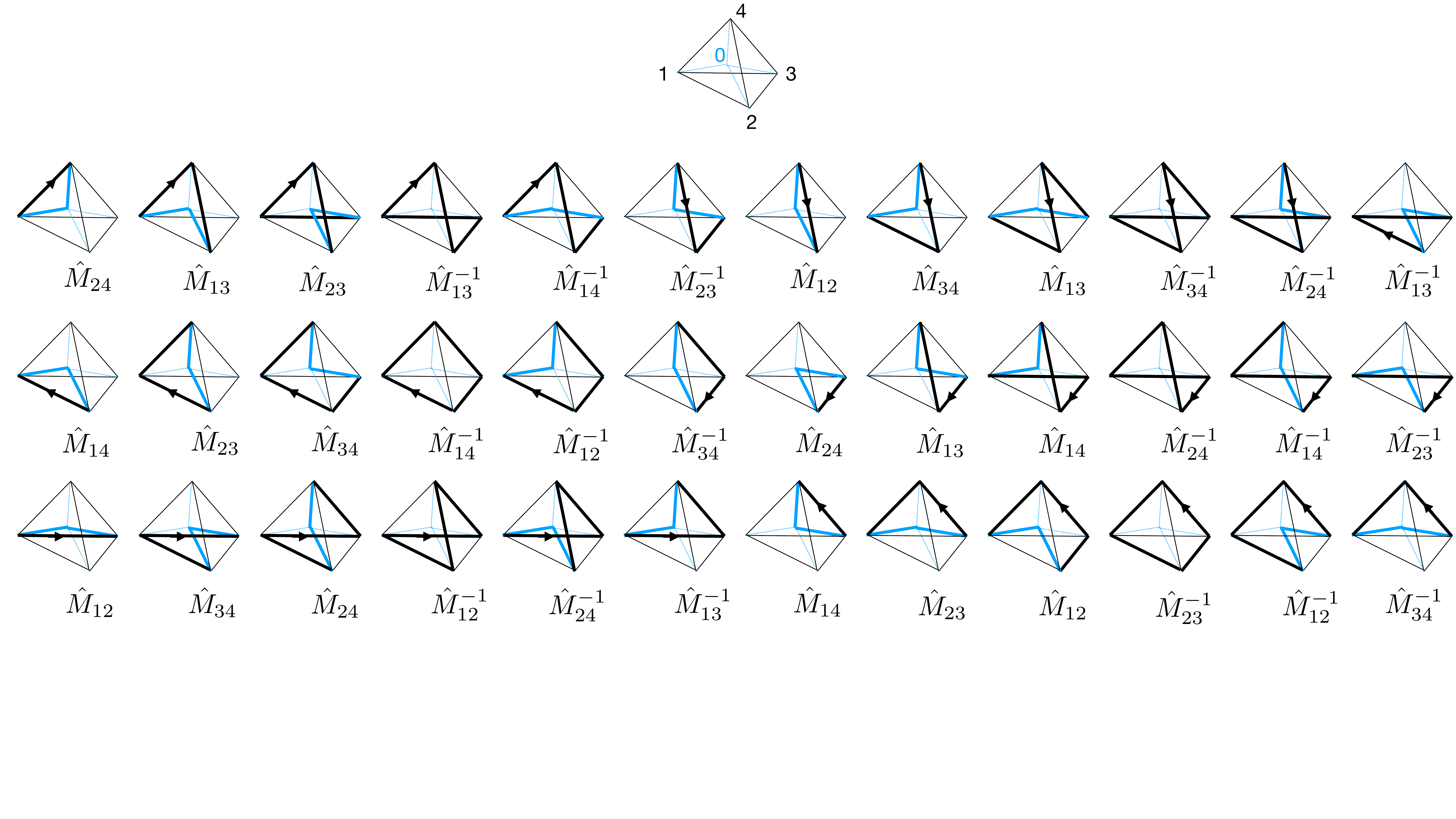}
\caption{
	Sequence of moves used in defining the loop self-statistics.  
	We start with the configuration $|\cnf_1\rangle$ in the upper left, 
	and, at each step, apply the operator below the configuration to obtain the next configuration 
	(reading left to right and up to down).  
	This sequence of operators is written as a product in \cref{eq:defhI}.
	Note that the loop comes back to itself at the end, 
	but the orientation along it reverses.  
	At the top we illustrate our labeling scheme for the vertices of the tetrahedron.
}
\label{fig:invariant_figure}
\end{figure}

Concretely, the process in \cref{fig:invariant_figure} is implemented by the following operator:
\begin{align} \label{eq:defhI}
\bM = &M_{34}^{-1}M_{12}^{-1}M_{23}^{-1}M_{12}M_{23}M_{14} M_{13}^{-1}M_{24}^{-1}M_{12}^{-1}M_{24}M_{34}M_{12}\\ \nonumber
 &M_{23}^{-1}M_{14}^{-1}M_{24}^{-1}M_{14}M_{13}M_{24} M_{34}^{-1}M_{12}^{-1}M_{14}^{-1}M_{34}M_{23}M_{14} \\ \nonumber
&M_{13}^{-1}M_{24}^{-1}M_{34}^{-1}M_{13}M_{34}M_{12} M_{23}^{-1}M_{14}^{-1}M_{13}^{-1}M_{23}M_{13}M_{24}
\end{align}
We define the loop self-statistics $\mu$ by
\begin{align} \label{eq:defI}
\bM |\cnf_1\rangle = \mu |\cnf_1\rangle .
\end{align}

\subsection{$\mu$ is well defined when the gauge charge is a fermion.}

We have to check that $\mu$ is well defined, 
{\em i.e.,} $\mu$ does not depend on the various arbitrary choices made above.
First, we will check that, for a given fixed set of~$\{|\cnf\rangle \}$, 
$\mu$ is independent of the particular choices of~$M_{ij}$.
Second, we will check that, when the gauge charge is a fermion, 
$\mu$ is independent of the choice of the $\{ |\cnf\rangle \}$.
From these two facts, 
it will follow that $\mu$ is an invariant of the phase, 
rather than just of a particular Hamiltonian.
The ground states of two Hamiltonians in the same phase can be related by a shallow circuit
(or, more generally, a short time evolution of a quasi-local pseudo-Hamiltonian evolution),
and conjugating by this circuit allows us to turn membrane operators~$M_{ij}$ 
associated with one Hamiltonian into those associated with the other,
and the latter are a valid choice.
Since the operator $\bM$ also ends up being conjugated, 
its eigenvalue $\mu$ does not change.

\subsubsection{$\mu$ against choices of $M_{ij}$}

Consider two different choices of membrane operators, 
$\{M_{ij}\}$ and $\{M'_{ij}\}$, 
and suppose that $|\cnf\rangle$ is a configuration state 
that is acted on by $M_{ij}$, 
or is the result of acting with $M_{ij}^{-1}$, in the expression for $\bM$.  
Then we claim that
\begin{align} \label{eq:hMprime}
	M'_{ij}|\cnf\rangle = u_{ij}^i(\cnf) u_{ij}^j(\cnf) u_{ij}^0(\cnf) M_{ij}|\cnf\rangle
\end{align}
where $u_{ij}^i(\cnf)$, $u_{ij}^j(\cnf)$, and $u_{ij}^0(\cnf)$ 
are $U(1)$ phases that depend only on the local occupation numbers,
in the configuration $\cnf$, 
of the edges that end at vertices $i,j$, and $0$, respectively.
We prove~\cref{eq:hMprime} in~\cref{app:pfMprime}.
Intuitively, it just states that the phase ambiguity associated to $M_{ij}$ 
depends locally on the configuration $|\cnf\rangle$ that $M_{ij}$ is acting on.

Now, let us compute the loop self-statistics~$\mu$
associated with the new set of membrane operators~$M'_{ij}$, 
{\em i.e.},
insert~$M'_{ij}$ in place of~$M_{ij}$ in the definition of~$\bM$ in~\cref{eq:defhI}
and act on~$|\cnf_1\rangle$, as in~\cref{eq:defI}. 
This expression is unaltered by inserting projectors $|\cnf_i\rangle\langle \cnf_i|$, 
for an appropriate~$i$,
between any two of the operators appearing in the expression for $\bM$, 
since, by~\cref{eq:Mcondition}, the operators $M'_{ij}$ just map between the configurations~$|\cnf_i\rangle$, 
up to an overall $U(1)$ phase.
We then use~\cref{eq:hMprime} or its conjugate
\begin{align} \label{eq:hMprimeconjugate}
	\langle \cnf | \left(M'_{ij}\right)^{-1} 
	=
	{\overline{u_{ij}^i}}(\cnf) {\overline{u_{ij}^j}}(\cnf) {\overline{u_{ij}^0}}(\cnf) \langle \cnf | M_{ij}^{-1}
\end{align}
on every one of the $36$ operators that appears in the expression.
We thus see that the new value of~$\mu$, computed using the $M'_{ij}$, 
is equal to the old value, computed using the $M_{ij}$, 
times a product of various $U(1)$ phases $u_{ij}^i(\cnf)$, $u_{ij}^j(\cnf)$, $u_{ij}^0(\cnf)$, and their complex conjugates.
We show in \cref{sec:app_details2} that these phases cancel,
{\em i.e.},
each phase appears the same number of times as its complex conjugate.
The core reason for this is that the expression in~\cref{eq:defhI} for~$\bM$ 
was engineered in such a way that, locally near each vertex,
each change in the configuration appears exactly the same number of times as its inverse.

\subsubsection{$\mu$ against choices of $\{ |\cnf\rangle \}$ when the gauge charge is a fermion}
\label{sec:againstChargeDecoration}

Now let us examine the dependence of~$\mu$ on the choice of~$\{ |\cnf\rangle \}$.  
First, note that if $U$ is a shallow circuit of local unitaries, 
then replacing~$\{ |\cnf\rangle \}$ with~$\{U |\cnf\rangle \}$ does not change~$\mu$.  
Indeed, we can simply conjugate the operators~$M_{ij}$ by~$U$ 
to obtain valid membrane operators for~$\{U |\cnf\rangle \}$.
Now suppose that $\{ |\cnf'\rangle \}$ is some other arbitrary set of configuration states, 
satisfying all of the properties in \cref{subsec:data}, 
and let us try to deform~$\{ |\cnf\rangle \}$ into~$\{ |\cnf'\rangle \}$.  
By possibly repeatedly applying shallow circuits,
we can bring~$\{ |\cnf\rangle \}$ near~$\{ |\cnf'\rangle \}$, 
and by applying shallow circuits on neighborhoods of the edges 
we can ensure that the reduced density matrices of~$\{ |\cnf\rangle \}$ and~$\{ |\cnf'\rangle \}$ 
on occupied edges are identical.  
However, there is a potential obstruction to the most natural way of completing this deformation on the vertices, 
since it may be the case that the disagreement near the vertices is by a gauge charge.  
In this case, there is no local unitary, acting near the vertices, that connects the two states, 
since they lie in different topological superselection sectors.%
\footnote{
	We exclude here the possibility of Cheshire charge~\cite{Cheshire},
	since our interest will be mostly in the case of fermionic gauge charges, which cannot condense on loops.
}  
Of course, by modifying the unitary on the edges by gauge charge string operators 
it may be possible to get rid of these extra gauge charges and connect the two sets of configurations.

To investigate this question, 
let us therefore consider the topological (shallow-circuit)
equivalence classes of states~$\{ |\tc \rangle \}$ 
which agree with~$\{ |\cnf\rangle \}$ on the interiors of the edges but may disagree on the vertices.
We may imagine that the states~$\{ |\tc \rangle \}$ are just the states $\{ |\cnf\rangle \}$, 
but with additional gauge charge decoration $x_v(\cnf) = 0,1$ at the various vertices $v=0,1,2,3,4$, 
which depends on the configuration~$\cnf$ only through the local portion of~$\cnf$ near~$v$.
This means that $x_v$ must be a function of the occupancy numbers on incident edges.
Such a function is a polynomial in binary variables~$y_j$ where $j$ ranges over all vertices other than~$v$, 
{\em e.g.}, $x_{v=0} = y_1 + y_2 + y_1 y_2 \in \ZZ_2[y_1,y_2,\ldots]/ (y_j^2 + y_j)$.
Since the number of occupied incident edges is always~$0$ or~$2$ in our set of configurations,
the occupancy variables obey the condition that~$y_j y_k y_\ell = 0$ for any distinct $j,k,\ell$
and~$\sum_j y_j = 0$.
In addition, if $y_j=0$ for all $j$, then $x_v$ must be zero;
with no incident flux tube near~$v$, we can detect a gauge charge at~$v$.
So, $x_v$ has no constant term.
Hence, the most general function~$x_v(y_j)$ is
a $\ZZ_2$-linear combination of quadratic functions;
a linear function $x_v = y_j$ is equal to $x_v = y_j^2 = y_j(\sum_{k \neq j} y_k)$.
Therefore, we introduce coefficients~$[avb] = [bva] \in \ZZ_2$ 
associated with each corner at~$v$ in a triangle~$avb$ so that
\begin{align}
	x_v(\cnf) &= \sum_{\text{triangle }avb} [avb] \cdot \delta_{avb}(\cnf)\label{eq:xv}\\
\text{where }	\delta_{avb}(\cnf) &= (\text{quadratic function~}y_a y_b) =
		\begin{cases} 
			1 & \text{if both edges $av$ and $vb$ are occupied in }\cnf,\\
			0 & \text{otherwise}.
		\end{cases}\nonumber
\end{align}
Such a decoration is subject to the consistency condition 
that every loop contain an even number of $\ZZ_2$ gauge charges.
This requirement translates to conditions on the coefficients~$[avb]$:
\begin{align}
	[jk0] + [k0j] + [0jk] &= 0 \label{eq:jk0}\\
	[ijk] + [jkl] + [kli] + [lij] &= 0 \label{eq:ijkl}
\end{align}
where $i,j,k,l \in \{1,2,3,4\}$ are distinct.
\Cref{eq:jk0} comes from loops around triangles
and \cref{eq:ijkl} from loops over four edges in~\cref{fig:all_configurations_figure}.
We have some loop configurations that occupy $5$ edges,
but we will see that we do not need to impose another set of constraints from these configurations.
Assuming this system of equations, we prove in~\cref{app:consistency} that
\begin{align}
	t = [avb] + [bvc] + [cva] \in \ZZ_2 \label{eq:constt}
\end{align}
is a constant independent of any distinct~$a,b,c,v$.

Next, we show that a gauge charge decoration~$x_v$ that is determined by a consistent set of coefficients~$[abc]$,
is realized by modified membrane operators.
For clarity, we may imagine that these gauge charges are offset from the vertices by a common spatial vector
which is long compared to the correlation length but short compared to the size of the tetrahedron.
Then, the modified membrane operators are
\begin{align} \label{eq:modified}
	\tilde M_{ij} = M_{ij} S_{0i}^{ [0ij] + t\cnf(0i) } S_{0j}^{ [0ji] + t\cnf(0j) } S_{ij}^{t\cnf(ij)}
\end{align}
where $S_{ab}$ is a string operator inserting a gauge charge at~$a$ and another at~$b$.
The appearance of the edge occupation number~$\cnf(ab) = \cnf(ba) = 0,1$ in the exponents
is a shorthand for~$S_{ab}$ being controlled on whether edge~$ab$ is occupied.
Let us explain why this modification realizes the charge decoration.
We have to check that the changes in charge decoration numbers 
are in accordance with those given by~$x_v$.
Upon the action by $\tilde M_{ij}$, 
the charge decoration number changes by
\begin{align}
	\Delta x_0 &= [0ij] + t\cnf(0i) + [0ji] + t\cnf(0j) 
			= [i0j] + t\cnf(0i) + t\cnf(0j) & \text{by \cref{eq:jk0},}\nonumber\\
	\Delta x_i &= [0ij] +  t\cnf(0i) + t\cnf(ij),\\
	\Delta x_j &= [0ji] + t\cnf(0j) + t\cnf(ij).\nonumber
\end{align}
This may be summarized as
\begin{align}
	\Delta x_v = [avb] + t\cnf(av) + t\cnf(vb).
\end{align}
On the other hand, 
if we toggle the edge occupation numbers on the three edges of a triangle~$avb$, 
then \cref{eq:xv} implies that
\begin{align}
	\Delta x_v 
	&= [avb] \Delta \delta_{avb} + 
		\sum_{k \neq a,v,b}  
			\left([avk] \Delta \delta_{avk} + 
			[kvb] \Delta \delta_{kvb}\right)
	\\
	&= [avb] \big((\cnf(av) + 1)(\cnf(vb) + 1) - \cnf(av)\cnf(vb) \big) +
		\sum_{k \neq a,v,b} 
			\left([avk] \cnf(vk) + [kvb] \cnf(kv)\right)\nonumber\\
	&= [avb](\cnf(av) + \cnf(vb) + 1) + \sum_{k \neq a,v,b}([avk]+[kvb])\cnf(kv) \nonumber\\
	&= [avb](\cnf(av) + \cnf(vb) + 1) + \sum_{k \neq a,v,b}(t + [avb])\cnf(kv) &\text{by \cref{eq:constt}} \nonumber\\
	&=[avb](\cnf(av) + \cnf(vb) + 1) + (t + [avb])(\cnf(av)+\cnf(bv))&\text{for } \sum_{l\neq v} \cnf(lv) = 0\nonumber\\
	&=[avb] + t\cnf(av) + t\cnf(bv).\nonumber
\end{align}
This completes the proof that our modified membrane operator realizes a given charge decoration.
Since charge decoration is realized by some string operators
whose end points are at vertices where flux loop passes through,
the requirement that any loop configurations occupying $5$ edges must have an even number of gauge charge decorations,
which we did not impose when we solved~\cref{eq:jk0,eq:ijkl},
is automatically satisfied.

In~\cref{app:consistency} we find that 
there are exactly two classes of solutions of~\cref{eq:jk0,eq:ijkl},
distinguished by the constant~$t=0,1$ of~\cref{eq:constt}.
When $t=0$, there is a new set of coefficients~$[ab] = [ba] \in \ZZ_2$,
where each~$[ab]$ is associated with an edge~$ab$, 
such that $[abc] = [ab] + [bc]$.
Then, the modified membrane operators are\footnote{When applied to our calculation of $\mu$, the vertex $a=0$.}
\begin{align}
	\tilde M_{\text{triangle }abc} 
	&= 
	M_{\text{triangle }abc} \, S_{ab}^{[ab]+[bc]} S_{ac}^{[ac] + [cb]}\\
	&\cong
	M_{\text{triangle }abc} \, S_{ab}^{[ab]} S_{ac}^{[ac]} S_{bc}^{[bc]}\nonumber\\
	&=
	U \,M_{\text{triangle }abc} \,U^\dagger \nonumber\\
\text{where }	U &=
	\prod_{\text{edge }yz} S_{yz}^{[yz] \cnf(yz)} .
\end{align}
Here, $\cong$ in the second line denotes 
modification of the membrane operator by $(S_{ab}S_{ac}S_{bc})^{[bc]}$; this modification does not affect $\mu$ due to~\cref{eq:hMprime}.
The appearance of $\cnf(yz)$ in the exponent of~$U$
means that the operator $S_{yz}^{[yz]}$ is controlled on the edge occupation on~$yz$.
Since $S_{yz}$ is a string-like shallow quantum circuit,
$U$ is also a shallow quantum circuit.
Therefore, any charge decoration with $t=0$ does not affect~$\mu$.

Since any two solutions with $t=1$ differ by a $t=0$ solution, it remains to confirm that $\mu$ is unaffected under any one particular $t=1$ solution.
One such $t=1$ solution is given by $[012] = 1$, $[034] = 1$, 
$[0jk] = 0$ for all other distinct $j,k \in \{1,2,3,4\}$,
and $[ijk] = 1 + [0jk] + [0ji]$ for all distinct $i,j,k \in \{1,2,3,4\}$.
It is straightforward to see that, for this charge decoration, the new value of~$\mu$, calculated from the decorated~$\tilde M_{ij}$ membrane operators, will differ by at most a sign from the old value.
Indeed, this sign difference has two contributions: 
(i)~the anticommutation of the gauge string operators with the membrane operators which they intersect, 
and 
(ii)~in the case of the gauge charges being fermions, 
the anticommutation of the gauge string operators~$S_{i0}$ with each other 
since all of them share the vertex $0$ in common.

\begin{figure}
\centering
\includegraphics[width=2.5in]{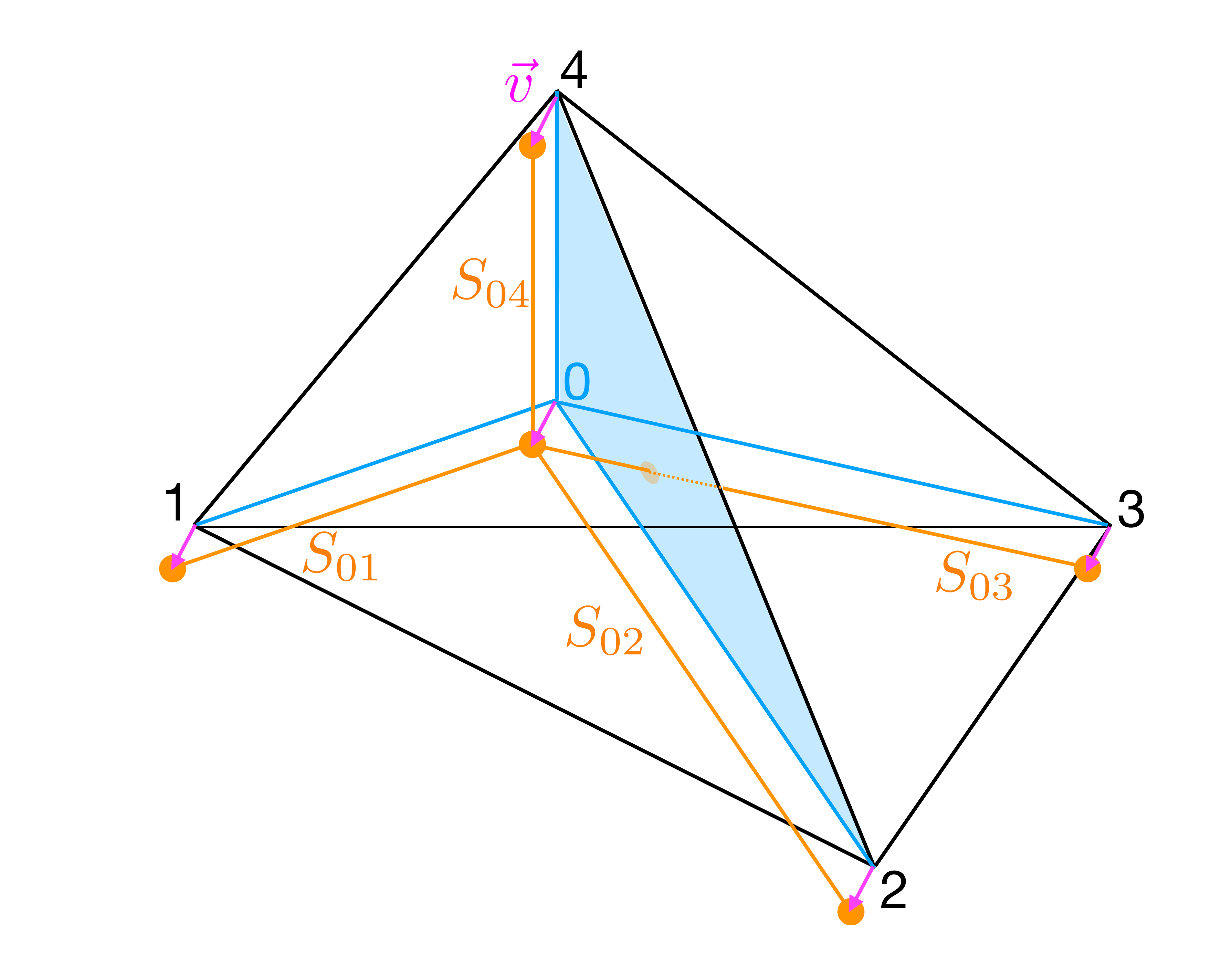}
\caption{
	There is a unique triple of distinct $j,i,k$, 
	such that the only non-trivial commutation relation is~$S_{0j}$ with~$M_{ik}$.  
	In the case shown in the figure, 
	this non-trivial commutation relation is between~$S_{03}$ and~$M_{24}$.
}
\label{fig:v_figure}
\end{figure}

\begin{figure}[b]
\centering
\includegraphics[width=6in]{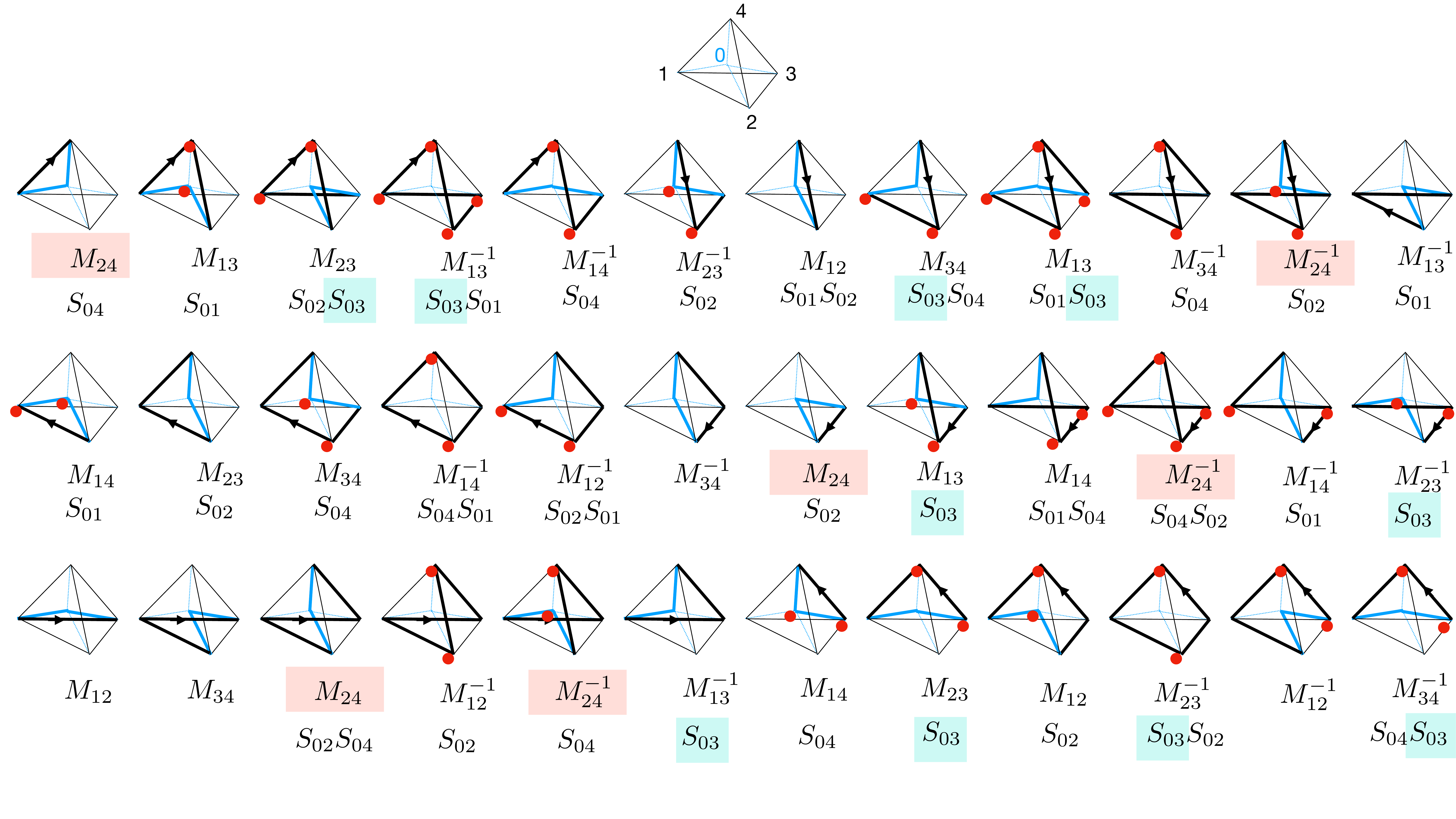}
\caption{
	The red dots indicate the gauge charges decorating the various vertices, 
	for the non-trivial $t=1$ decoration described in the text.  The red rectangles highlight the steps at which the membrane operator~$M_{24}$ or its inverse is applied.  
	The blue rectangles highlight an application of the string operator~$S_{03}$.  
	Commuting such a membrane operator past such a string operator 
	gives rise to a minus sign in a $\ZZ_2$ gauge theory.  
	As is apparent from the figure, 
	it takes an odd number of anticommutations to cancel off all of the membrane operators against each other 
	and all of the string operators against each other, 
	so the overall sign contribution from such anticommutations is $-1$.
	Note that the string operators were chosen to square to $+1$.
	Furthermore, we can assume that, in the case of fermionic gauge charges, the only pairs of anti-commuting string operators are $(S_{01},S_{02})$ and $(S_{03},S_{04})$ (one can always get this from any other choice by multiplying some of the string operators by gauge charge detection operators at~$0$).  We can check from the figure that the sign from anti-commuting $S_{01},S_{02}$ past each other, in such a way as to cancel all of these operators, is~$1$, whereas the sign from anti-commuting $S_{03},S_{04}$ past each other is~$-1$, leading to an overall sign of~$-1$ due to the fermionic statistics of the gauge charges.}
\label{fig:decorated3}
\end{figure}

Consider the first contribution.
Note that for a generic displacement~$\vec v$, 
there will be a unique triple of distinct~$j,i,k$, 
such that the only nontrivial commutation relation is~$S_{j0}$ with~$M_{ik}$, 
as illustrated in~\cref{fig:v_figure} with~$(ik)=(24)$ and~$j=3$.
As we check in~\cref{fig:decorated3}, 
these commutation relations contribute a factor of~$-1$ to the loop statistics.

The second contribution, due to the statistics of the gauge charges,
is nontrivial only in the case of fermionic gauge charges,
in which case the string operators~$S_{i0}$ anticommute for certain pairs.
We find in~\cref{fig:decorated3} that the product of all of the fermionic string operators 
appearing in our decorated $36$-step process yields a factor of~$-1$.
Hence, the second contribution is~$-1$ in the case of fermionic gauge charges, 
and $+1$ in the case of bosonic gauge charges.

Thus, in the case of bosonic gauge charges~$\mu$ is not well defined, 
since its value changes by an overall factor of~$-1$ 
when the loop configurations are nontrivially decorated by gauge charges.
However, in the case of fermionic gauge charges 
the two contributions to the sign always cancel and $\mu$ is well defined.  
Furthermore, for the usual $3+1$-dimensional fermionic toric code, 
{\em i.e.}, the FcBl model, $\mu=1$, as can be checked 
by an explicit calculation in the $3+1$d Walker--Wang model 
based on the premodular category $\{1,f\}$, where $f$ is a fermion.

We show in \cref{app:mu} that, if $\mu$ is well defined, it must always be equal to $\pm 1$, 
and in the discussion we give an argument, based on some physical assumptions, 
that any stand-alone $3+1$d realization of the fermionic toric code must have $\mu=1$.
Futhermore, in the next sections, we construct a $4+1$d invertible exactly solved Hamiltonian (the ``FcFl model'')
which realizes a $3+1$d fermionic toric code with~$\mu=-1$.
Taken together, these facts imply that our FcFl model is not shallow circuit equivalent to a product state.  
On the other hand, we will see that two copies of it are shallow circuit equivalent to a product state, 
so that the FcFl model is a $\ZZ_2$-classified invertible phase of matter.


\section{Exactly solved model} \label{sec:WW}

In this section we construct a $4+1$d exactly solved Hamiltonian 
which realizes an anomalous $\mu=-1$ fermionic toric code on its boundary.
Our construction is analogous to that of Walker and Wang~\cite{Walker_2011}, in the sense that it bootstraps a boundary topological order into a bulk Hamiltonian in one higher dimension.  However, our dimensions are shifted up by one from the case discussed in~\cite{Walker_2011}, and the input data is a braided fusion $2$-category rather than a premodular category.
Because such generalizations of Walker--Wang models have not been studied before, 
we first warm up by constructing simpler exactly solved models 
which realize nonanomalous bosonic and fermionic toric codes on their respective boundaries.
As we will explicitly check, 
these simpler models will be short-range entangled, 
{\em i.e.}, their ground states on geometries with no boundary will be small-depth (shallow) circuit disentanglable.
The model realizing the anomalous fermionic toric code, on the other hand, 
will be invertible but not shallow circuit disentanglable.

A guiding principle to define explicit Hamiltonians below is as follows.
We would like the ground state to be a superposition of 
vacuum-to-vacuum processes involving
topological excitations of an input theory,
where the amplitude is the corresponding transition amplitude.
If the input theory is in $3+1$-dimensional spacetime,
then each component of our ground state wavefunction
is a picture of a dynamical process
drawn in a $4$-dimensional canvas.
The input theory has point-like and loop-like excitations 
(generally in every dimension up to space codimension~$2$),
and a picture of ours consists of worldlines and worldsheets.
Fluctuation operators of these worldmembranes 
will be generators of our picture,
and the sum of these picture generators will be our $4+1$d Hamiltonian.
Whenever the fluctuation changes the topology of worldmembranes,
the topological interaction of the input theory 
gives nontrivial change in the transition amplitude,
which must be encoded in the fluctuation operators.
The point of our construction below is
that desired fluctuation operators can be constructed
such that they all commute
and the ground state is nothing else but what we want.

We assume that we have a cellulation of a $4$-dimensional space
which refines to a triangulation.
We will use the Poincar\'e dual cellulation that refines to a triangulation as well,
where a $k$-dimensional cell of the original cellulation 
corresponds to a $(4-k)$-dimensional dual cell of the Poincar\'e dual.
The simplest choice would be the $4$-dimensional cubic lattice 
whose $0$-cells are identified with~$\ZZ^4 \subset \RR^4$.
To distinguish the two cellulations,
we call the former the \emph{primary} cellulation consisting of primary cells,
and the latter the \emph{secondary}.
As usual, the sizes of all primary and secondary cells are uniformly bounded from below and above.
The sizes of cells define our lattice spacing, the smallest length scale in our construction.

Suppose there is one qubit~($\CC^2$) on each primary $2$-cell and on each secondary $1$-cell.  These degrees of freedom can be interpreted as $3$-form and $2$-form gauge fields, respectively; a short discussion of this gauge theory interpretation is given in \cref{sec:discussion}.
In view of the guiding principle above,
the primary $2$-cells and secondary $1$-cells 
will support worldsheets and worldlines
of line-like and point-like excitations of an input theory, respectively.
\Cref{fig:terms} might be helpful 
to understand the construction of Hamiltonian terms.

\begin{figure}
\includegraphics[width=\textwidth, trim={0ex 70ex 25ex 0ex}, clip]{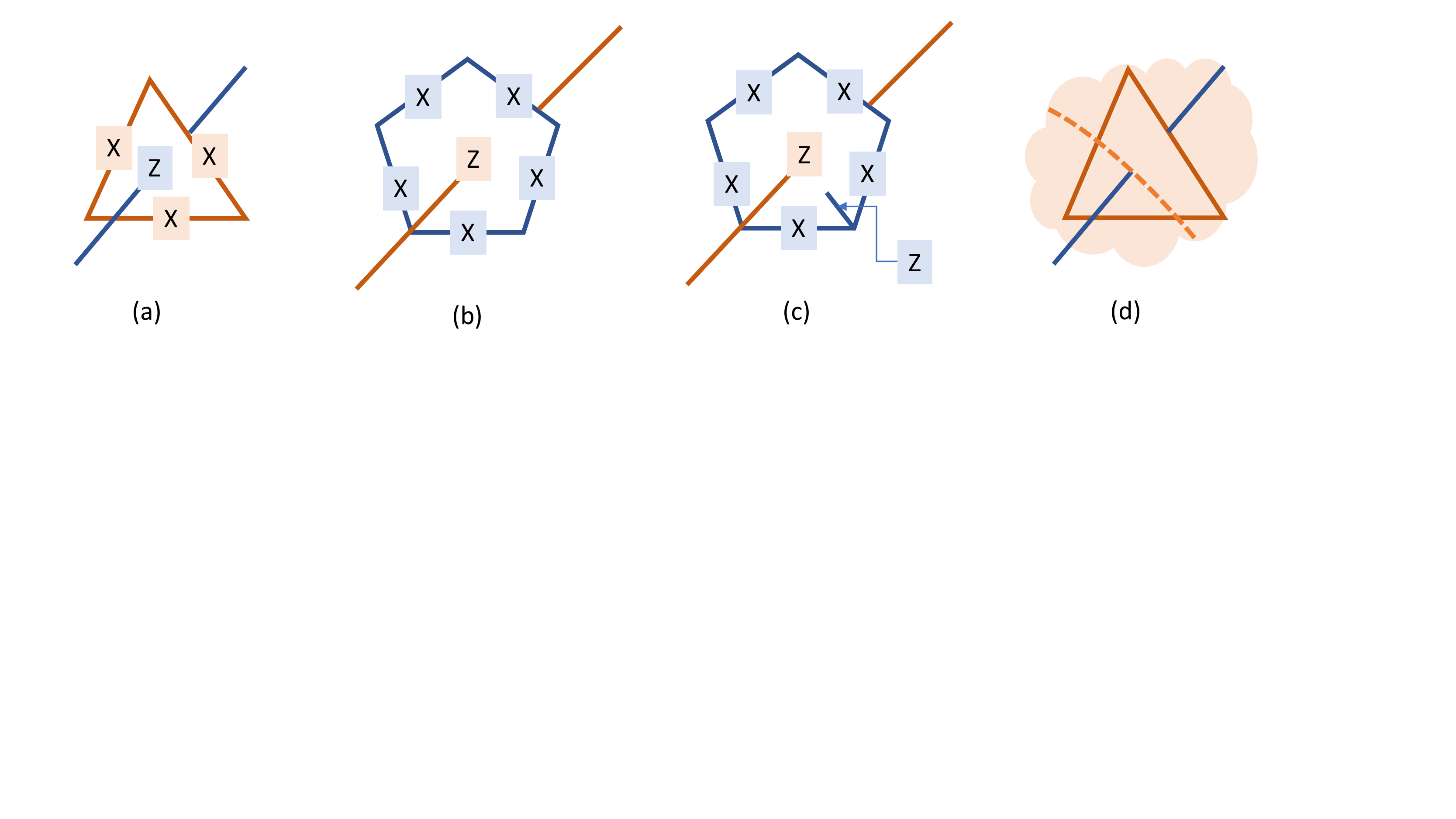}
\caption{
Red segments in the figure denote a primary $2$-cell;
they are drawn as if they were one-dimensional,
but they are two-dimensional and the fourth direction is not shown.
Blue segments are secondary $1$-cells.
($\mathsf a$)
A worldsheet fluctuation operator~$T_{Bl}$.
The $Z$ factor on the blue secondary $1$-cell, Poincar\'e dual
to the primary $3$-cell bounded by the depicted primary $2$-cycle,
captures the mutual braiding between a point-like excitation
and a line-like excitation of the input theory, 
the $3+1$-dimensional toric code.
The $X$ tensor factors implement the fluctuation of worldsheet.
($\mathsf b$)
A worldline fluctuation operator~$L_{Bc}$.
A secondary $2$-cell is not necessarily a pentagon.
($\mathsf c$)
A worldline fluctuation operator~$L_{Fc}$.
The extra $Z$ factor on the secondary $1$-cell
whose projection under~$\phi$ 
intersects the interior of the secondary $2$-cell
that bounds the secondary $1$-cycle with $X$ factors,
captures the fact that the worldline is a trajectory of a fermion.
($\mathsf d$)
This is meant to illustrate the worldsheet fluctuation operator~$T_{Fl}$,
which consists of the same operator content as~$T_{Bl}$,
and, in addition, extra nonPauli operators that captures 
the change in twisting of the orientation domain wall (orange dashed curve)
of the primary $2$-cycle.
}
\label{fig:terms}
\end{figure}

\subsection{Bosonic charge and bosonic loop excitations (BcBl)}\label{sec:bcbl}

\paragraph{Hamiltonian.}

We define a \emph{worldsheet} fluctuation operator $T_{Bl}$ for each primary $3$-cell~$e^3$, 
which is Poincar\'e dual to a secondary $1$-cell~$e_1$,
by the tensor product of Pauli~$X$ over all the primary $2$-cells $f^2 \in \bd e^3$ 
at the boundary of the $3$-cell~$e^3$ and Pauli~$Z$ on~$e_1$.
If the primary cellulation is the cubic lattice, each worldsheet fluctuation operator
is a product of $7$~Pauli operators.
In addition, we define a \emph{worldline} fluctuation operator $L_{Bc}$ for each secondary $2$-cell~$f_2$,
which is Poincar\'e dual to a primary $2$-cell~$f^2$,
by the tensor product of Pauli $X$ over all the secondary $1$-cells $\ell_1 \in \bd f_2$
at the boundary of~$f_2$ and Pauli~$Z$ on~$f^2$.
If the primary cellulation is the cubic lattice,
each worldline fluctuation operator is a product of $5$~Pauli operators.
We define a Hamiltonian $H^{WW}_{BcBl}$ to be the negative sum of 
all the worldsheet and worldline fluctuation operators.
\begin{align}
	H^{WW}_{BcBl} = - \sum_{f_2: \text{cells}} L_{Bc}(f_2) - \sum_{e^3: \text{cells}} T_{Bl}(e^3)
\end{align}
It is natural to include terms that 
enforce closedness of worldlines and worldsheets,
but they turn out to be redundant.
A worldline closedness is enforced by demanding that
the product~$\pi(v_0)$ of~$Z$ along the boundary of every secondary $0$-cell~$v_0$
should take eigenvalue~$+1$.
But this product of~$Z$ around~$v_0$ 
is equal to~$\pi(v_0) = \prod_{e^3: v_0 \in \bd e_1} T_{Bl}(e^3)$
where $e_1$ is Poincar\'e dual to~$e^3$.
Likewise, the worldsheet closedness is enforced 
by demanding that $\pi(a^1) = \prod_{f_2: a^1 \in \bd f^2} L_{Bc}(f_2)$
should take eigenvalue~$+1$
where $f_2$ is Poincar\'e dual to~$f^2$.

\paragraph{Commutativity.}

It is obvious that any pair of worldsheet fluctuation operators~$T_{Bl}$ commute
since the tensor factors on the primary qubits are all~$X$
and those on the secondary are all~$Z$.
The same is true for any pair of worldline fluctuation operators~$L_{Bc}$.
The less obvious case is when a worldsheet fluctuation operator~$T_{Bl}$ on~$e^3$
meets a worldline fluctuation operator $L$ on~$f_2$
where the Poincar\'e dual~$f^2$ of~$f_2$ is on the boundary of~$e^3$.
But $f^2 \in \bd e^3$ happens precisely 
when the Poincar\'e dual $e_1$ of $e^3$ is on the boundary of~$f_2$.
Hence, when the $X$ factor of~$T_{Bl}$ on~$f^2$ anticommutes with the $Z$ factor of~$L$,
the $X$ factor of~$L$ on~$e_1$ anticommutes with $Z$~factor of~$T_{Bl}$.
Therefore, $H^{WW}_{BcBl}$ consists of commuting terms.

\paragraph{Disentangling circuit.}

Let $U$ be a two-qubit unitary:
\begin{align}
U =
\left\{\begin{matrix}
\ket{++}\\
\ket{+-}\\
\ket{-+}\\
\ket{--}
\end{matrix}\right\}
\begin{pmatrix}
1 & 0 & 0 & 0 \\
0 & 1 & 0 & 0 \\
0 & 0 & 1 & 0 \\
0 & 0 & 0 & -1 
\end{pmatrix}
\left\{\begin{matrix}
\bra{++}\\
\bra{+-}\\
\bra{-+}\\
\bra{--}
\end{matrix}\right\},
\qquad
\begin{cases}
 U(X \otimes I) U^\dagger = X \otimes I \\
 U(I \otimes X) U^\dagger = I \otimes X \\
 U(Z \otimes I) U^\dagger = Z \otimes X \\
 U(I \otimes Z) U^\dagger = X \otimes Z 
 \end{cases}
 \label{eq:U}
\end{align}
where the braced column matrices denote the basis 
in terms of~$\ket{\pm} = \pm X\ket{\pm}$.
Note that $U$ is invariant under exchange of the two qubits in its support.
Consider a shallow quantum circuit
\begin{align}
	V = \prod_{e_1 \sim f^2} U_{e_1,f^2} \label{eq:V}
\end{align}
where~$e_1 \sim f^2$ ranges over all pairs of a secondary $1$-cell $e_1$ 
and a primary $2$-cell $f^2$ such that the Poincar\'e dual of~$e_1$ has~$f^2$ in its boundary.
This product is well defined because $U$ is diagonal in the $X$ basis.
Under conjugation by~$V$, 
every worldsheet fluctuation operator becomes a single-qubit operator~$Z$ on the associated secondary $1$-cell,
and every worldline fluctuation operator becomes~$Z$ on the associated primary $2$-cell.
Hence, $V H^{WW}_{BcBl} V^\dagger$ is the negative sum of single-qubit operator~$Z$ over all qubits in the system.
This implies in particular that the ground state of~$H^{WW}_{BcBl}$ is
the unique common eigenstate of all the worldsheet and worldline 
fluctuation operators with eigenvalue~$+1$ and is disentangled by~$V$.

In fact, the disentangling circuit is supplied by a general fact as follows.
\begin{lemma}\label{lem:disentanglingU}
Suppose there are $n$ qubits
and let $\stab = \{P_1,P_2,\ldots,P_n\}$ 
be a set of Pauli operators of form
\begin{equation}
P_j = Z(j) \otimes \bigotimes_{i \in \Gamma(j)} X(i)
\end{equation}
where the arguments denote the qubit the operator acts on
and $\Gamma(j)$ is a subset of \mbox{$\{1,2,\ldots, n\}\setminus\{j\}$}.
If $P_j$ commutes with $P_k$ for all~$j,k$,
then some product of~$U$ defined in~\cref{eq:U} maps $P_j$ to a single-qubit operator~$Z(j)$
for every~$j$.
\end{lemma}
This can be thought of as a characterization of so-called graph states~\cite{Hein2003}.
\begin{proof}
If $i \in \Gamma(j)$, then $P_i$ has the factor $Z(i)$ anticommuting with $X(i)$ of $P_j$.
The commutativity demands that this anticommutation must be canceled by another anticommutation
which can only be given by a factor $X(j)$ of $P_i$.
This means that $j \in \Gamma(i)$.
Hence, we have an undirected graph with $n$ nodes where there is an edge between~$i$ and~$j$
if and only if $i \in \Gamma(j)$ or equivalently $j \in \Gamma(i)$.
If we apply $U$ over every edge of this graph, then $\stab$ is disentangled.
\end{proof}

\paragraph{Amplitudes in the ground state.}

Given a primary $1$-cell~$e^1$, consider its Poincar\'e dual~$e_3$
and consider the product~$\prod L$ of all the worldline fluctuation operators on the boundary~$\bd e_3$.
The $X$ factors of these operators are on the boundary of $b_2 \in \bd e_3$
and hence all cancel since $\bd \bd e_3 = 0 \bmod 2$.
The remaining factors are~$Z$ on the primary $2$-cells whose boundary includes~$e^1$.
The product~$\prod L$ taking an eigenvalue~$+1$ 
means that the ground state consists of configurations of primary $2$-cells
that must define a $2$-cycle with $\ZZ_2$~coefficients.

Likewise, the product~$\prod T_{Bl}$ of all the worldsheet fluctuation operator on the boundary~$\bd h^4$
of a primary $4$-cell~$h^4$ has surviving factors~$Z$ on all secondary $1$-cells 
whose boundary contain the Poincar\'e dual of~$h^4$.
The product~$\prod T_{Bl}$ taking an eigenvalue~$+1$
means that the ground state consists of configurations of secondary $1$-cells
that must define a $1$-cocycle with $\ZZ_2$~coefficients.

The Hamiltonian terms drive fluctuations in the worldsheet and worldline configurations.
In the basis where $Z$ is diagonal, 
the qubit states~$\ket 1$ and~$\ket 0$ on a primary $2$-cell indicate
that the worldsheet of a loop excitation has and has not swept that location in spacetime,
respectively.
Likewise, the qubit state~$\ket 1$ and~$\ket 0$ on a secondary $1$-cell represent
the occupancy of the worldline of a charge on that spacetime location.
Hence, the Hamiltonian terms correctly capture the property that 
if the worldline of a charge links with the worldsheet of a loop,
there must be an amplitude factor of $-1$ relative to the configuration where they are not linked.
Therefore $H^{WW}_{BcBl}$ is a Walker--Wang Hamiltonian in $4+1$d
with respect to the $3+1$d toric code topological order.

\paragraph{Linking number.}

Given any exact secondary $1$-chain~$a_1$,
let~$\tilde a_2$ be any secondary $2$-chain whose boundary is~$a_1$.
Likewise, given any exact primary $2$-chain~$b^2$,
let~$\tilde b^3$ be any primary $3$-chain whose boundary is~$b^2$.
Since a secondary chain is canonically a cochain,
we can consider mod~$2$ intersection numbers~$\Int_2(\tilde a_2, b^2)$ 
and~$\Int_2(a_1, \tilde b^3)$ by the evaluation of the cochain on the chain.

The amplitude~$\pm 1$ for the configuration of~$a_1$ and~$b^2$
can be computed starting from the vacuum
by applying various fluctuation operators,
and a choice of set of fluctuation operators determines~$\tilde a_2$
and~$\tilde b^3$.
The amplitude is precisely the mod~$2$ intersection numbers, 
which must be the same:
if we first build $b^2$ by worldsheet fluctuation operators
and then insert $a_1$ by worldline fluctuation operators,
then we will be computing~$\Int_2(\tilde a_2, b^2)$;
if we first build~$a_1$ and then insert~$b^2$,
then we will be computing~$\Int_2(a_1, \tilde b^3)$.
Hence, the existence of the ground state of~$H^{WW}_{BcBl}$
shows that the intersection number is well defined independent
of~$\tilde a_2$ and~$\tilde b^3$.
We conclude that we may \emph{define} the mutual linking number modulo~$2$
as a \emph{function} of the exact cycles:
\begin{align}
	\Lnk( a_1, b^2) = \Int_2( \tilde a_2, b^2 ) = \Int_2( a_1, \tilde b^3) \label{eq:Lnk}
\end{align}
Note that this conclusion only uses Poincar\'e duality with $\ZZ_2$ coefficients,
which is valid even for nonorientable manifolds.

\subsection{Fermionic charge and bosonic loop excitations (FcBl)}\label{sec:fcbl}

We keep using the worldsheet fluctuation operator of $H^{WW}_{BcBl}$.
The worldline fluctuation operator, on the other hand, must be modified
since we want the worldline to be one of a fermion, not of a boson.
The canonical way to describe a fermion via worldlines
is to consider \emph{framed} worldlines.
If a worldline closes, the holonomy of the frame along the line 
is valued in $\pi_1(SO(3))\cong \ZZ_2$
and the closed worldline has quantum amplitude~$\pm 1$.

In this section we take a simplified approach, 
motivated by the idea of a ``blackboard'' framing determined by a projection to $2$~dimensions.
Although we could in principle just work with a $4$-dimensional hypercubic lattice 
with a generic linear projection to $2$d --- thereby directly generalizing~\cite{BCFV} ---
we instead find it useful to work in the more general setting of an arbitrary cellulation 
of a general $4$-dimensional spatial manifold.

Consider therefore a piecewise linear map 
\begin{align}
	\phi : \mathcal K \to \RR^2
\end{align}
called a \emph{projection},
from the $2$-skeleton $\mathcal K$ of the secondary cellulation down to $\RR^2$.  
We require that $\phi$ should map every secondary $1$-cell to a straight line segment of nonzero length,
and the images of $1$-cells are transverse to one another.
These conditions are generically satisfied:
we can simply project all the $0$-cells to points in generic position on $\RR^2$,
and connect two projected points if the pair is the boundary of some $1$-cell.
We also require that the image of a $2$-cell~$f_2$ is defined as in~\cref{fig:phi}.
That is, if $v_0(1),v_0(2),\ldots,v_0(n)$ are the vertices of~$f_2$,
then we subdivide~$f_2$ into $n-2$ triangles~$t_j$, each formed by~$v_0(1),v_0(j),v_0(j+1)$ for $j=2,\ldots,n-1$,
and map each~$t_j$ injectively to the triangle in $\RR^2$ formed by $\phi(v_0(1)),\phi(v_0(j)),\phi(v_0(j+1))$.\footnote{
	Previously in~\cite[\S II]{4dbeyond} we introduced a similar projection,
	but required that every 2-cell is mapped injectively.
	We no longer require this injectivity.
}
So, every $2$-cell is projected to a polygon, which may be folded.
\begin{figure}
	\includegraphics[width=2in, trim={5ex 95ex 140ex 5ex}, clip]{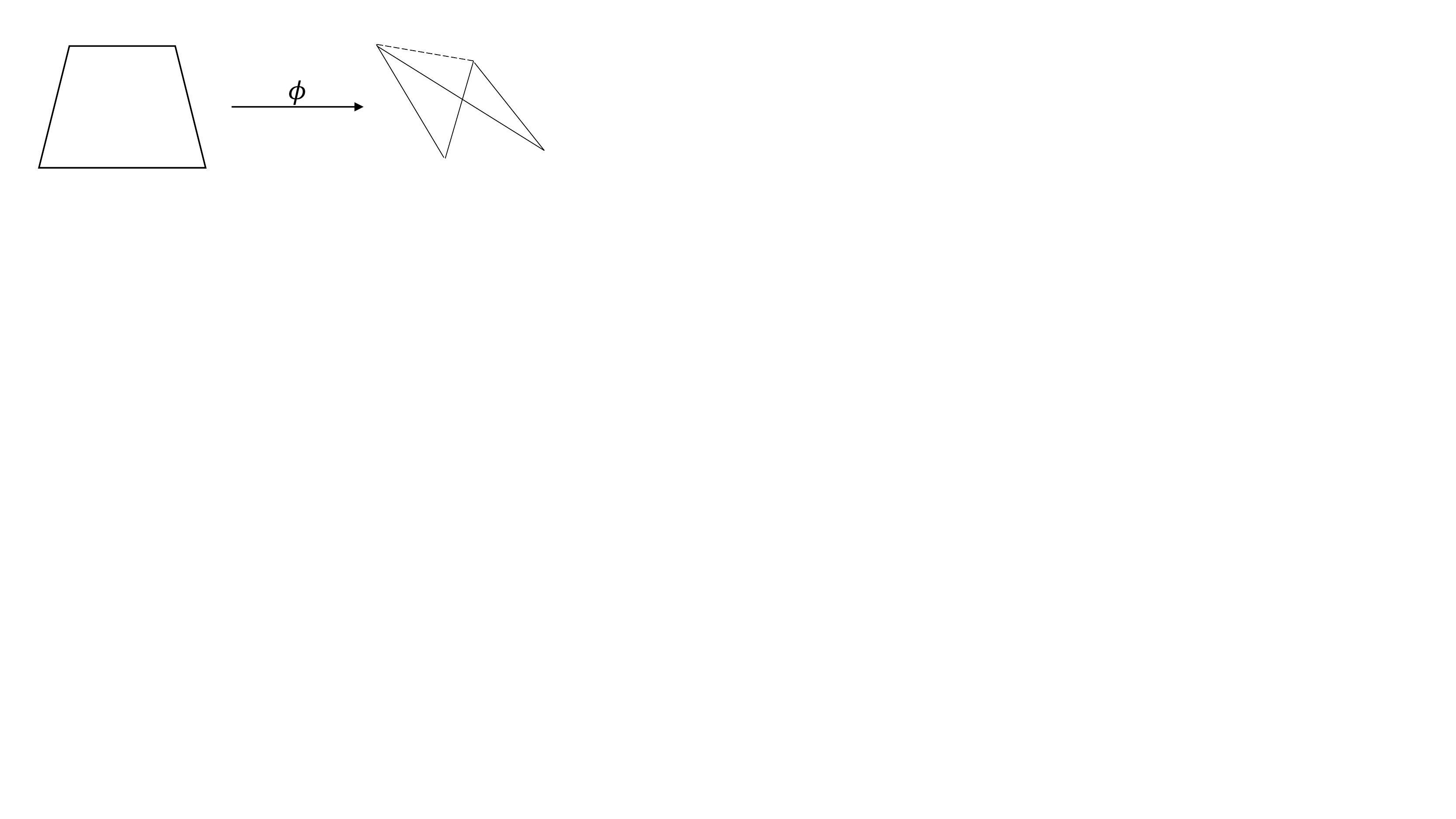}
	\caption{
		A $2$-cell that has four or more sides may be mapped to a folded figure under $\phi$.
	}
	\label{fig:phi}
\end{figure}

\paragraph{Hamiltonian.}

We have to define a worldline fluctuation operator for each secondary $2$-cell~$f_2$.
First, assume that $\phi$ is injective on $f_2$ ({\em e.g.}, $f_2$ is a triangle).
Take the worldline fluctuation operator $L_{Bc}$ of~$H^{WW}_{BcBl}$,
which consists of Pauli~$X$ along the boundary of~$f_2$ and 
Pauli~$Z$ on the Poincar\'e dual $f^2$ of~$f_2$.
We multiply this operator by Pauli~$Z$ tensor factors
on the secondary $1$-cells~$e_1$ such that
there is a factor~$Z$ on~$e_1$ if and only if
$\phi(e_1)$ intersects the interior of~$\phi(f_2)$.
This defines a term $L_{Fc}(f_2)$.
Given two such $2$-cells $f_2$ and $f_2'$, $L_{Fc}(f_2)$ and $L_{Fc}(f_2')$ then commute for a nontrivial reason.  
We will not give the argument here because it is verbatim the same as in~\cite[\S II.C.a]{4dbeyond}.

Generally, even if $\phi|_{f_2}$ is not injective,
the projected polygon~$\phi(f_2)$ defines a subdivision~$(f_2)_1,\ldots,(f_2)_n$ of~$f_2$
such that each subcell $(f_2)_i$ is injectively mapped under $\phi$.
For example, in~\cref{fig:phi} a quadrilateral decomposes into two triangles.
Except for the Pauli~$Z$ on $f^2$,
we follow the prescription for each $(f_2)_i$
as if it were a genuine $2$-cell and the folding lines (the dashed line in~\cref{fig:phi})
were occupied by qubits.
Then, we multiply all these to obtain 
$\hat L(f_2) = \prod_i \left[ L_{Fc}((f_2)_i) \setminus Z(f^2)\right]$.
Since the factors are commuting, the multiplication order does not matter and $\hat L(f_2)$ is unambiguous.
The product $\hat L(f_2)$ does not contain any $X$ factor on the folding lines
because a folding line is shared by two subcells.
However, there could be $Z$ factors on folding lines in $\hat L(f_2)$.
We define
\begin{align}
L_{Fc}(f_2) = Z(f^2) \bra{0_{\text{folding lines}} } \hat L(f_2) \ket{0_{\text{folding lines}} }. \label{eq:lfc}
\end{align}
where $\ket{0_{\text{folding lines}} }$ is the $+1$ eigenvalue eigenstate of all of the $Z$ operators on the folding lines.  That is, we remove any $Z$ factor on the folding line 
but keep any signs that may have been accumulated while forming~$\hat L(f_2)$.

The Hamiltonian~$H^{WW}_{FcBl}$ is the negative sum of all the worldsheet fluctuation operators,
which are the same as in~$H^{WW}_{BcBl}$,
and all the worldline fluctuation operators,
which is the product of the worldline fluctuation operator of~$H^{WW}_{FcBl}$ 
and the extra $Z$ factors determined by~$\phi$.
\begin{align}
	H^{WW}_{FcBl} = - \sum_{f_2: \text{cells}} L_{Fc}(f_2) - \sum_{e^3: \text{cells}} T_{Bl}(e^3)
\end{align}
As in~$H^{WW}_{BcBl}$, 
we do not have to include the worldline closedness terms~$\pi(v_0)$.
Also, it turns out that we do not have to include worldsheet closedness terms;
this is not immediately obvious, 
but will be handled by a disentangling circuit below.

\paragraph{Commutativity.}

The worldsheet fluctuation operators are the same as those of the BcBl model, so we already know that they commute with each other.  Any worldsheet fluctuation operator commutes with 
any worldline fluctuation operator because of the same reason as in $H^{WW}_{BcBl}$;
the extra $Z$ factors in the FcBl worldline fluctuation operator are 
on the secondary cellulation,
on which the worldsheet fluctuation operator has $Z$ factors only.
Any two worldline fluctuation operators commute 
because each is a product of commuting operators associated with sub-$2$-cells.
Therefore, $H^{WW}_{FcBl}$ consists of commuting terms.

\paragraph{Disentangling circuit.}

The unitary~$V$ of~\cref{eq:V} transforms the worldsheet fluctuation operator
to the single-qubit operator $Z$.
Hence, $H ' = V H^{WW}_{FcBl} V^\dagger$ 
has a completely disentangled ground state on the secondary cellulation.
Next, we focus on the worldline fluctuation operator $L=L_{Fc}(b_2)$ after~$V$.  
Letting $b^2$ denote the primary $2$-cell that is Poincar\'e dual to $b_2$,
\cref{eq:U} implies that $VLV^\dagger$ consists of three kinds of Pauli tensor factors:
$Z$ on~$b^2$,
$Z$ on some secondary $1$-cells $e_1$ near $b^2$ determined by $\phi$
(denoted as $e_1 \propto b^2$),
and
$X$ on primary $2$-cells~$p^2$ if they are on the boundary of (Poincar\'e dual of) 
an odd number of such $e_1$'s (denoted as $p^2 \in \db e_1$).
On the ground state of~$H'$, 
the $Z$ on the secondary $1$-cells~$e_1$ have eigenvalue~$+1$,
so we may ignore them. Hence,
\begin{align}
	VLV^\dagger \bmod Z \text{ on secondary qubits} 
	= Z(b^2) \prod_{e_1 \propto b^2}\prod_{p^2 \in \db e_1} X(p^2) .
\end{align}
Now, there is exactly one $VLV^\dagger$ term for each primary qubit 
on which $VLV^\dagger$ act by $Z$.
The other factors of $VLV^\dagger$ is $X$ (ignoring $Z$ on secondary qubits).
Any $VLV^\dagger$ term commutes with another, 
and \cref{lem:disentanglingU} guarantees that $V H^{WW}_{FcBl} V^\dagger$ can be disentangled.
We have completed the disentangling circuit.

A consequence of the disentangling circuit is
that the ground state of $H^{WW}_{FcBl}$ is unique.
Namely, it is
the superposition of all cycles in the trivial $2$-homology and $1$-cohomology 
class with an appropriate amplitude~$\pm 1$.
This implies that the sign is a \emph{function} of cycles.
Since the mutual linking in~\cref{eq:Lnk} is a function of cycles,
we conclude there is a well-defined amplitude $(-1)^{\Frm}$
assigned to each worldline~$a_1$ (in the trivial homology class),
regardless of which $2$-chain bounding this worldline is used to compute the amplitude.
This assignment is defined whenever we have a projection~$\phi$.

In \cref{app:frm} we elaborate on~$\Frm$ to 
construct a commuting Pauli Hamiltonian~$H_{Fc}$,
which is a topologically ordered model with a fermionic point-like excitation
on any combinatorial manifold of dimension two or higher.

\subsection{Fermionic charge and fermionic loop excitations (FcFl)}

In this subsection we finally construct a Hamiltonian for a non-trivial invertible $4+1$d bosonic phase.
Again we find it useful to do this for a general cellulation of a general $4$-manifold, 
though we can easily specialize to a hypercubic lattice.
We refer to our model as the FcFl model, 
because there is a sense in which the loop is fermionic.  
Specifically, there is an extra factor ($\pm 1$) 
in the amplitude for the worldsheet of fermionic loop excitations,
in addition to the framing of fermionic worldline and the mutual linking between worldlines and worldsheets.
This extra $\pm 1$ factor can be read off as follows.
On a $2$-cycle representing the worldsheet, 
there is an orientation domain wall, which is a $1$-cycle $W^1$.
We regard~$W^1$ as a worldline of a (fictitious) fermionic charge
and read off the associated amplitude of~$W^1$.
The cycle~$W^1$ may or may not link with the worldsheet,
but there is no contribution to the quantum amplitude from this linking;
the imagined fermionic charge along~$W^1$ is gauge-neutral.

The orientation domain wall of a surface is only defined as a homology class;
however, we always choose a preferred cycle~$W^1$.
Our prescription to this end is determined by globally defined (local) orientations for each primary $2$-cell.
Given any primary $2$-chain~$b^2$ with $\ZZ_2$ coefficients
we promote it to a $2$-chain~$\xi(b^2)$ with $\ZZ_4$ coefficients,
\begin{align}
	\xi : \ZZ_2[\text{$2$-chains}] \to \ZZ_4[\text{$2$-chains}],
\end{align}
by the rule $\ZZ_2 \ni 0 \mapsto 0 \in \ZZ_4$ and $\ZZ_2 \ni 1 \mapsto 1 \in \ZZ_4$,
where every cell in a chain with $\ZZ_4$ coefficients carries its orientation.
Then, for a closed $2$-chain~$b^2$ we define
\begin{align}
W^1(b^2) = \tfrac12 \bd \xi(b^2)
\end{align}
where $\bd$ is the mod~$4$ boundary operator
and $\tfrac12$ is well defined because $\bd \xi(b^2)$ vanishes mod~$2$.

\paragraph{Hamiltonian.}

There is one qubit on each primary $2$-cell and one qubit on each secondary $1$-cell.
We continue to use the worldline fluctuation operators $L_{Fc}$ of $H^{WW}_{FcBl}$,
for which we need to fix a projection~$\phi_2$ from the $2$-skeleton of the secondary cellulation.
Also, we fix another projection~$\phi^2$ from the $2$-skeleton of the primary cellulation;
$\phi$ is piecewise linear, mapping $1$-cells to transverse straight line segments. 
See~\cref{sec:fcbl}.
Such a projection exists for the standard hypercubic cellulation of $\RR^4$.
The worldsheet fluctuation operator~$T_{Fl}$ on a primary $3$-cell~$e^3$ 
({\em i.e.}, a secondary $1$-cell~$e_1$)
is then defined as follows.
\begin{align}
	T_{Fl} = D \cdot \underbrace{Z( e_1 ) \cdot \prod_{f^2 \in \bd e^3} X(f^2)}_{T_{Bl}}
\end{align}
The operator~$D$ is a diagonal operator in the $Z$ basis, acting on qubits on primary $2$-cells,
and calculates the amplitude change due to the orientation domain wall change.
We define $D$ on the subspace of closed chains below;
if there is open boundary detected on the support of~$D$,
then we let $D$ have eigenvalue~$0$ on that state.
The change $\Delta W^1$ of the orientation domain wall 
can be computed using that fact that for all $2$-chains $x^2$ and $y^2$,
\begin{align}
\xi(x^2+y^2) = \xi(x^2)+\xi(y^2)\quad\text{ if }\quad x^2 \cap y^2 = \emptyset, \label{eq:xi-linear-nonintersecting}
\end{align}
and is given by, with $s^2 \equiv b^2 \cap \bd e^3$,
\begin{align}
	\Delta W^1 
	&= W^1(b^2 + \bd e^3) - W^1(b^2) \nonumber\\
	&= \tfrac12 \bd \left(\xi\big( (b^2 + s^2) + (\bd e^3 + s^2) \big) - \xi\big( (b^2 + s^2) + s^2 \big)\right) \label{eq:DeltaW1}\\
	&= \tfrac12 \bd \left(\xi(\bd e^3 + s^2) - \xi(s^2)\right). \nonumber
\end{align}
Hence, $\Delta W^1$ is local.
Reading the existing orientation domain wall,
and gluing $\Delta W^1$ in with interpretation 
that they are fermionic worldlines, materialized by the projection~$\phi^2$,
we compute the change in the amplitude locally.
This is $D$.
Therefore, $T_{Fl}$ is local.

In addition to the worldline and worldsheet fluctuation operators,
we also include
\begin{align}
\pi(a^1) = \prod_{f^2: a^1 \in \bd f^2} Z(f^2), \qquad
\pi(v_0) = \prod_{e_1: v_0 \in \bd e_1} Z(e_1)
\end{align}
in $H^{WW}_{FcFl}$ which enforce closed primary $2$-chains and closed secondary $1$-chains.
So, the full Hamiltonian reads
\begin{align}
H^{WW}_{FcFl} = - \sum_{a^1: \text{cells}} \pi(a^1) - \sum_{b_2: \text{cells}} L_{Fc}(b_2) - \sum_{e^3: \text{cells}} T_{Fl}(e^3) - \sum_{v_0: \text{cells}} \pi(v_0).
\end{align}
The term $T_{Fl}$ is not a tensor product of Pauli matrices.

\paragraph{Commutativity.}
The $\pi$ terms commute as they are diagonal,
and they also commute with $L_{Fc}$ and $T_{Fl}$ 
as the fluctuation operators preserve the subspace of closed chains.
The worldline fluctuation operators $L_{Fc}$ commute with each other as we remarked before,
which relies on~\cite[\S II.C.a]{4dbeyond}.
A worldline fluctuation operator $L_{Fc}$ and a worldsheet fluctuation operator $T_{Fl}$ commute
for the same reason as in $H^{WW}_{FcBl}$ since the diagonal factor $D$ of $T_{Fl}$ 
does not overlap with $X$ factors of $L_{Fc}$.
It remains to check the commutation of the operators $T_{Fl}$ with themselves.
It suffices to check this on the subspace of closed chains
since $T_{Fl}$ annihilates other states.
But for a state with closed chains,
the action of $T_{Fl}$ is to change the cycle configuration 
and insert a sign that is a \emph{function} of cycle configurations.
Hence, $T_{Fl}$ commute among themselves, and $H^{WW}_{FcFl}$ is a commuting Hamiltonian.

\paragraph{Unique ground state of $H^{WW}_{FcFl}$.}

We have defined $H^{WW}_{FcFl}$ such that there is a ground state where all terms (without the overall minus sign)
assume $+1$ on that ground state.
This does not immediately imply that there is only one ground state.
The uniqueness of the ground state is proven for $H^{WW}_{BcBl}$ and $H^{WW}_{FcBl}$
since we have the circuits that disentangle entire spectra of these Hamiltonians,
but we do not yet know the uniqueness of the ground state of $H^{WW}_{FcFl}$.
However, we can still show the uniqueness, independent of disentangling transformations.

Any ground state must be in the subspace of closed chains;
otherwise, the $\pi$ terms will not assume $+1$ and the energy will be higher.
Then, the worldline and worldsheet fluctuation operators 
hybridize all configurations in a (co)homology class $[h] \in H_2(\mathcal K^4;\ZZ_2) \oplus H^1(\mathcal K^4;\ZZ_2)$.
Suppose that 
\begin{itemize}
\item[($\star$)] if a product of~$L_{Fc}$ and~$T_{Fl}$ is diagonal in $Z$ basis,
{\em i.e.}, the product does not change the configuration,
then it acts as $+1$ on the space of all cycles in~$[h]$.
\end{itemize}
We have proved this condition for products over null-homologous cycles;
the existence of one ground state is a proof.
Then, we can consider a basis of closed chains with the amplitudes in place implied by the fluctuation operators.
In this basis $L_{Fc}$ and $T_{Fl}$ have off-diagonal elements that are either~$+1$ or~$0$.
The Perron--Frobenius theorem implies that 
there is a unique ground state in $[h]$ 
under the assumption~($\star$).
If ($\star$) is violated for~$[h]$, then there is no ground state in~$[h]$.

It remains to show that~($\star$) is violated for all but the zero homology class.
Suppose a primary $2$-cycle~$s^2$ (worldsheet) is in a nontrivial homology class.
Then, there is a secondary $2$-cycle~$s_2$ that intersects this $2$-cycle at odd number of positions
by Poincar\'e duality.
Take the product~$\prod L$ of all worldline fluctuation operators~$L_{Fc}$ on~$s_2$.
Since $L_{Fc}$ on any secondary $2$-cell~$b_2$ has $X$ factors on~$\bd b_2$,
the product~$\prod L$ has no $X$ factors, and $\prod L$ is a diagonal operator
with $Z$ factors on every primary $2$-cell dual to secondary $2$-cells on $s_2$.
These $Z$ factors probe the parity of the intersection number between $s^2$ and $s_2$,
which is odd.
But this is contradictory to the fact that each $L_{Fc}$ has eigenvalue~$+1$
on a hypothetical ground state with a nontrivial primary $2$-cycle $s^2$.
A similar argument rules out any nontrivial homology class for secondary $1$-cycles;
the product of worldsheet fluctuation operators~$T_{Fl}$ over nontrivial primary $3$-cycle
detects any nontrivial secondary $1$-cycles.
Therefore, any ground state is in the trivial homology sector,
and $H^{WW}_{FcFl}$ has only one ground state.

\subsection{Disentangling two FcFl states}\label{sec:disentangle2FcFl}

We are going to construct a shallow quantum circuit
that maps the stacked tensor product of two FcFl states to the stacked tensor product of two FcBl states.
Since we have a shallow disentangling circuit for an FcBl state,
the two FcFl states will be disentangled.
We put one FcFl state on our usual system with qubits on primary $2$-cells and secondary $1$-cells.
The other FcFl state is put on the dual system of the same cellulation,
{\em i.e.}, on a system with qubits on secondary $2$-cells and primary $1$-cells.
We do not try to find a shallow circuit between the two FcFl states,
but we note that on $\RR^4$ with the standard hypercubic cellulation
the two states are equal up to a small ``diagonal'' translation,
and hence are considered two identical states laid side by side.

Let $c$ denote any configuration of a primary $2$-cycle and a secondary $1$-cycle in the trivial homology and cohomology class.
Similarly, let $c'$ denote any configuration of a secondary $2$-cycle and a primary $1$-cycle.
We write $c = a \oplus b$ and $c' = a' \oplus b'$
where $a$ is the secondary (or primary if primed) $1$-cycle 
and $b$ is the primary (or secondary if primed) $2$-cycle.
We have
\begin{align}
	\ket{2 \text{FcFl}} = \sum_{c,c'} \eta^f(c,c')  \ket{c,c'}
\end{align}
where the sign $\eta^f$ is determined by three contributions:
\begin{itemize}
	\item linking of $a$ with $b$ and that of $a'$ and $b'$;
	\item framing of $a$ and $a'$;
	\item framing of $W^1(b)$ and $W_1(b')$.
\end{itemize}
If the last contribution was absent, the wavefunction becomes that of the two FcBl states $\ket{2 \text{FcBl}}$.
Now, we introduce a nominally different basis for the closed chain subspace.
Define an involution~$\omega$ by
\begin{align}
(c,c') &= a' \oplus b \oplus b' \oplus a \nonumber\\
\mapsto \omega(c,c') &= (a' + W^1(b)) \oplus b \oplus b' \oplus (a + W_1(b'))\\
\mapsto (\omega \circ \omega)(c,c') &= (c,c') . \nonumber
\end{align}
This involution~$\omega$ is simply a permutation on the set $\{ (c,c') \}$ of all closed chains.
We use the $\omega$-permuted basis to write
\begin{align}
	\ket{2 \text{FcBl}} = \sum_{c,c'} \eta^b(c,c') \ket{\omega(c,c')}
\end{align}
where the sign $\eta^b$ is determined by two contributions:
\begin{itemize}
	\item linking of $a + W_1(b')$ with $b$ and linking of $a' + W^1(b)$ with $b'$;
	\item framing of $a + W_1(b')$ and framing of $a' + W^1(b)$.
\end{itemize}

Therefore, the difference $\eta^f \eta^b = (-1)^{q+r_1+r^1}$ for a given pair $c,c'$ is
given by the linking number difference $q$ and the framing difference $r_1 + r^1$:
\begin{align}
	q(c,c') &= \Lnk(W_1(b'), b) + \Lnk(W^1(b),b')  \nonumber \\
	r^1(c,c') &= \Frm(a'+W^1(b)) + \Frm(a') + \Frm(W^1(b)) \\
	r_1(c,c') &= \Frm(a+W_1(b')) + \Frm(a) + \Frm(W_1(b')) \nonumber
\end{align}

We will show below that each of $q,r^1,r_1$ is computed by a shallow quantum circuit.
The circuit for~$q$ will require global orientation of the $4$-space.
That is, we construct a shallow quantum circuit $Q$ such that
\begin{align}
	Q \eta^f(c,c')\ket{c,c'} = \eta^b(c,c') \ket{c,c'}
\end{align}
for all $c,c'$.
Then,
\begin{align}
	\ket{2\text{FcBl}} = \underbrace{\sum_{c,c'} \ket{\omega(c,c')} \bra{c,c'}}_{\Omega} Q \ket{2\text{FcFl}}.
\end{align}
Here, $\Omega$ can be realized by a shallow quantum circuit.
One can apply control-$\sqrt X$ or control-$\sqrt{X}^\dagger$ from each $2$-cell to each of its boundary $1$-cells,
depending on the local orientation differences.
The displayed expression for $\Omega$ does not determine the circuit;
we are only specifying the action of $\Omega$ on $2$-cycles of trivial (co)homology.

After showing that each of $q,r^1,r_1$ is computed by a shallow quantum circuit in the subsections below, we will have completed the construction of a disentangling circuit for two FcFl states on any \emph{orientable} $4$-space with a cellulation which refines to a triangulation.

\subsubsection{Circuit for $q(c,c')$}


We claim that, \emph{if the ambient $4$-space is oriented},
\begin{align}
	\tfrac 1 2 \Int_4(b,b') = \Lnk(W_1(b'), b) + \Lnk(W^1(b),b') \label{eq:IntLnk}
\end{align}
where $\Int_4$ is the mod~$4$ intersection number which we will define shortly.
It is well defined to take the half because $\Int_4(b,b') = 0 \bmod 2$
if $b$ is null-homologous.
$\Int_4(b,b')$ is defined by regarding $b'$ as a 2-cochain over $\ZZ_4$ and evaluating it at $b$ as a 2-chain over $\ZZ_4$.
Geometrically, this can be described as follows.
Each $2$-cell carries a preferred local orientation,
so at an intersection point we have two mutually transverse $2$-dimensional ordered bases.  
Concatenating these, starting with the primary cell's basis, 
gives an ordered local basis of the full $4$-dimensional space, 
which can be compared to the ambient orientation.  
This gives a sign and the mod~$4$ intersection number is the mod~$4$ sum of all the signs.

Consider the ``error'' $\Delta(b,b') = \tfrac 1 2 \Int_4(b,b') + \Lnk(W_1(b'), b) + \Lnk(W^1(b),b')$ in our claim.
We observe that $\tfrac 1 2 \Int_4(b,b') + \Lnk(W^1(b),b')$ is
independent of how we choose preferred local orientations for each $2$-cell of $b$.
Indeed, if we change the local orientation for a cell nonintersecting with $b'$,
then the intersection number is unaffected and there is no new linking.
If we change the local orientation for an intersecting cell,
then $\Int_4$ changes by $\pm 2 = 2 \bmod 4$
and the newly created piece of $W^1$ links with $b'$.  Thus both $\tfrac 1 2 \Int_4(b,b')$ and $\Lnk(W^1(b),b')$ change by $1$, whereas $\Lnk(W_1(b'), b)$ is unaffected, so $\Delta(b,b')$ does not change.  Therefore, $\Delta(b,b')$ is independent of local orientations of $2$-cells.

We further claim that $\Delta(b,b')$ is constant
even if deform $b$ or $b'$ by the boundary of a $3$-cell.
Suppose we deform $b$ by $\bd e^3$. 
Any deformation on $b'$ will be handled by symmetry.
Since our $\Delta$ is independent of local orientations of $2$-cells,
we may choose them so that \mbox{$W^1(\bd e^3) = \frac 1 2 \bd \xi \bd e^3 = 0$.}
Recall that $\xi$ is a lifting from $2$-chains over $\ZZ_2$ to $2$-chains over $\ZZ_4$ determined by local orientations.
This means that $e^3$ can be given an orientation, 
{\em i.e.}, a choice of a mod~$4$ chain denoted by~$\xi e^3$,
such that 
\begin{align}
	\bd \xi e^3 = \xi \bd e^3. \label{eq:bdxie3}
\end{align}
Here, the boundary operator~$\bd$ on the left-hand side is mod~$4$ and that on the right-hand side is mod~$2$. 
Then, by~\cref{eq:DeltaW1} we have, with a mod~$2$ chain~$u = b \cap \bd e^3$,
\begin{align}
	W^1(b + \bd e^3) - W^1(b) 
	&= \tfrac 1 2 \bd\left( \xi(\bd e^3 + u) - \xi u \right) \nonumber\\
	&= \tfrac 1 2 \bd\left( \xi \bd e^3 - \xi u - \xi u \right) \\
	&=\tfrac 1 2 \bd ( -2 \xi u)  = \bd u, \nonumber
\end{align}
where in the second line we used
\begin{align}
	\xi(f^2 + g^2) = \xi f^2 - \xi g^2 \qquad \text{for any } f^2 \supseteq g^2. \label{eq:xi-linear-subset}
\end{align}
By standard equivalence between linking and intersection, the changes in the linking numbers are
\begin{align}
\Lnk(W_1(b'),b+\bd e^3) - \Lnk(W_1(b'),b) &= \Lnk(W_1(b'),\bd e^3) = \Int_2(W_1(b'), e^3),\nonumber\\
\Lnk(W^1(b+\bd e^3), b') - \Lnk(W^1(b),b') &= \Lnk(\bd u, b') = \Int_2(u,b').
\end{align}
To calculate the change in $\Int_4$,
let us be more explicit about the lifting of coefficients from~$\ZZ_2$ to~$\ZZ_4$
by writing
$\Int_4(b,b') = \Int_4[ \xi b, \xi b' ]$.
Then, with the mod~$2$ chain~$u = b \cap \bd e^3$,
\begin{align} 
&\Int_4[\xi(b+\bd e^3), \xi b'] - \Int_4[\xi b, \xi b']  \mod 4
\nonumber\\
&= 
\Int_4[ \xi( (b+u)+(\bd e^3 + u) ), \xi b')  - \Int_4[\xi b, \xi b'] 
&\text{decomposition of }b+\bd e^3
\nonumber\\
&=
\Int_4[ \xi(b + u) + \xi(\bd e^3 + u), \xi b']  - \Int_4[\xi b,\xi b'] 
&\text{\cref{eq:xi-linear-nonintersecting}}
\nonumber\\
&=
\Int_4[ \xi b - \xi u,\xi b'] + \Int_4[ \xi\bd e^3 - \xi u, \xi b']  - \Int_4[\xi b,\xi b'] 
&\text{\cref{eq:xi-linear-subset}}
\nonumber\\
&=
-2\Int_4[\xi u, \xi b'] + \Int_4[\bd \xi e^3, \xi b']
&\text{\cref{eq:bdxie3}}
\nonumber\\
&=
-2\Int_4[\xi u, \xi b'] + \Int_4[\xi e^3,\bd \xi b'] 
&\Int\text{= cochain evaluation}
\label{eq:lotsofintersectionnumbers}\\
&=
-2\Int_4[\xi u, \xi b'] + \Int_4[\xi e^3, 2W_1(b')] 
&\text{by def. of $W_1$}
\nonumber \\
&=
2 \Int_2(u,b') + 2\Int_2(e^3, W_1(b')) 
&(2 x) \bmod 4 = 2(x \bmod 2)
\nonumber
\end{align}
where the reader might want to recall that we write $W_1$ as it is a chain of secondary $1$-cells,
which are dual of primary $3$-cells,
so it can be regarded as a $3$-cocycle.
Therefore, $\Delta$ is constant under changes of the $2$-cycles by the boundary of any $3$-chain. 

Finally, since the cycles that enter the FcFl ground state are all (co)homologically trivial,
it suffices to evaluate $\Delta(b,b')$ for zero cycles $b,b'$,
for which it is obvious that $\Delta(b,b') = 0$.
We have proved~\cref{eq:IntLnk} and
\begin{align}
	q(c,c') = \tfrac12 \Int_4(b,b').
\end{align}

The mod~$4$ intersection number is computable by a depth~$1$ quantum circuit.
Indeed, the local preferred orientations of $2$-cells $f^2$ and its dual $f_2$
determines the contribution $\pm 1$ to the intersection number when they are both occupied by $2$-cycles $b,b'$.
This information is fixed once and for all, and the quantum gate $\diag(1,1,1,\pm i)$ reads this off.

\subsubsection{Circuit for $r_1$ and $r^1$}\label{sec:circuitBiFrm}

The construction will be the same for primary and secondary cellulations, 
so we only discuss~$r^1$.
By~\cref{eq:Frm} and~\cref{eq:association} we know that
$r^1( c,c' ) = \BiFrm(W^1(b), a')$,
which we show can be evaluated by a shallow circuit.
In fact, the extension~$\overline\BiFrm$ can be evaluated by a shallow circuit:
\begin{lemma}\label{lem:circuitR}
	Under the assumption of~\cref{prop:longS},
	there exists a shallow circuit~$R$ such that 
	$R \ket{a^1 , b^2} = (-1)^{\overline\BiFrm(a^1,W^1(b^2))} \ket{a^1 , b^2}$
	for all cycles~$a^1$ and~$b^2$.
\end{lemma}
\begin{proof}
\cref{rem:BiFrmByExt} says that $\overline\BiFrm$ may be evaluated 
by arbitrarily extending the $2$-skeleton~$\mathcal K$ of our $4$-space.
We introduce an auxiliary $0$-cell $\infty$ and project it to a generic position 
``infinitely'' away (which can be specified finitely but we will not) from the image of any other cell.
We build a cone over the $2$-skeleton~$\mathcal K$ with an apex~$\infty$.
We can obviously extend~$\phi$ over the $2$-skeleton of the cone to have $\phi^e$,
and all the line segments from a $0$-cell of~$\mathcal K$ to~$\infty$ on the projection plane~$\RR^2$
are almost parallel to each other since $\infty$ is far away.
Since $\mathcal K$ is locally finite,
the topology of line segments in the vicinity of $\phi(\mathcal K) \subset \RR^2$ stabilizes as $\infty$ is pushed to ``infinity.''
This extension fulfills the assumption of~\cref{prop:longS},
since all $1$-cycles become null-homologous.

Now, the cycle~$a$ is the boundary of the cone~$\tilde a$, a $2$-chain, over~$a$ itself.
For any $2$-cell~$f \in \tilde a$ in the cone, 
the Pauli factors of the operator~$L_{Fc}(f)$ on the $1$-cells of~$\mathcal K$
(those that do not touch~$\infty$)
are locally determined.
In particular the thorns of $S(a)$ are locally determined.
Since every coefficient of cells in~$W^1(b^2)$ is locally determined,
we find a desired shallow circuit~$R$.
\end{proof}


\section{Boundary topological order}\label{sec:bdtheory}

For any of our models BcBl, FcBl, and FcFl,
we have shown that the bulk does not admit any deconfined topological excitation
since the stack of two copies of a model can be completely disentangled.
We have also shown that the unique ground state on a closed $4$-manifold
is a superposition of null-homologous cycles.
However, in the presence of boundaries 
we will show that the boundary hosts nontrivial topological excitations and
that the ground state becomes degenerate whenever the $3$-dimensional boundary 
has nontrivial first, and hence second, homology.

Before proceeding with the detailed construction, let us make some comments.  First, the boundary Hamiltonians discussed below are not unique, and in general we expect the various models to admit many different boundary topological orders.  However, we will show that, in particular, the FcFl model admits an anomalous boundary topological order, {\em i.e.}, one which is not possible strictly in $3+1$ dimensions (modulo an assumption that there are no nontrivial $3+1$-dimensional fermionic invertible phases).  This fact by itself is sufficient to show that the FcFl model defines a nontrivial invertible phase of matter.  For if it did not, we could find a shallow depth disentangler for it, and by truncating this disentangler appropriately we could then disentangle the $3+1$-dimensional surface from the $4+1$-dimensional bulk and obtain a standalone $3+1$-dimensional model of the anomalous topological order, a contradiction.

Imagine we have a closed $4$-space which we will cut out to open up a $3$-dimensional boundary.
We choose a connected set of \emph{primary}
$4$-cells, and declare that the qubits contained in (the closure of) the chosen $4$-cells are our system qubits.
This choice of a connected subset of qubits
leaves a complementary set of \emph{external} qubits.
By construction, the system boundary is a closed $3$-manifold
and inherits the cell structure of the $4$-space.
Namely, the system boundary is filled with primary $3$-cells,
each of which is the intersection of a system primary $4$-cell and an external primary $4$-cell.

Given our Hamiltonian, be it BcBl, FcBl, or FcFl, on a closed $4$-space,
we define a new Hamiltonian as
\begin{align}
	H = -\sum_j h_j 
	\quad \longrightarrow \quad 
	H_\text{w/bd} = -\sum_j \bra{0^{\otimes \text{external}}} h_j  \ket{0^{\otimes \text{external}}} .
\end{align}
In other words,
if a term $h_j$ contains a $Z$ factor at an external qubit, we replace it with a number~$1$,
but if a term $h_j$ contains an $X$ factor at an external qubit, we drop the term.
One may regard the prescription as introducing (infinitely) strong single-qubit terms $-Z$ on external qubits.
The Hamiltonian with boundary is still commuting.
Any surviving pair of two terms that are modified by the qubitwise sandwiching
must have had diagonal factors on external qubits,
and the commutativity of the modified terms 
is inherited from that of the unmodified terms.

The original Hamiltonian~$H$ contains (in the algebra generated by its terms)
cycle-enforcing terms~$\pi(a^1) = \prod_{f^2: \bd f^2 \ni a^1} Z(f^2)$ 
and~$\pi(v_0) = \prod_{e_1: \bd e_1 \ni v_0} Z(e_1)$.
These terms survive under the qubitwise sandwiching
to become cycle-enforcing terms on the nonexternal, system qubits.
It is possible that all but one factor of a closed-chain enforcing term
may be sandwiched out to produce a single-qubit~$Z$,
in which case, effectively, the set of external qubits is enlarged.
Indeed, for any secondary $1$-cell~$e_1$ (a primary $3$-cell) at the system boundary,
there exists an external primary $4$-cell~$v_0$ such that $e_1 \in \db v_0$.
If $4$-cell $v_0$ had no other system qubit on its boundary,
the cycle-enforcing term $\pi(v_0)$ becomes a single-qubit $Z$ under the sandwiching,
and the qubit at~$e_1$ is disentangled out.
If we consider a hypercubic lattice with a flat cubic lattice boundary,
then every secondary qubit within the boundary is disentangled out.

\subsection{Gapped boundary}

Our Hamiltonian with boundary is gapped in a strong sense:
if $O$ is an operator on a not-too-big ball-like subsystem (excluding external qubits)
that commutes with every term of~$H_\text{w/bd}$,
then $O$ belongs to the complex algebra generated by terms of $H_\text{w/bd}$ near the support of~$O$.
Colloquially, this is to say that any commuting term that can be added to~$H_\text{w/bd}$ is already in $H_\text{w/bd}$.
Suppose that $K$ is any operator supported on a ``small'' ball-like region,
not necessarily commuting with any Hamiltonian term.
We multiply~$K$ by~$t_j$, the projector onto $+1$-eigenspace of~$h_j$ of~$H_\text{w/bd}$,
on the left and right for all~$j$ 
such that the support of~$t_j$ overlaps with that of~$K$,
to obtain an operator 
\begin{align}
	K' = \underbrace{(\prod_j t_j )}_{\Pi(K)} K \underbrace{(\prod_j t_j )}_{\Pi(K)}.
\end{align}
The support of~$K'$ is only slightly larger than that of~$K$,
and thus still supported on a ball-like region,
but $K'$ now commutes with all Hamiltonian terms.
Our claim implies that 
$K'$ is a $\CC$-linear combination of products of Hamiltonian terms near the support of~$K$.
This is the local topological order condition,
that is used to prove the gap stability against perturbations~\cite{BravyiHastingsMichalakis2010stability}.
It demands that if $K$ is any observable on a small box,
then the expectation value of $K$ on any true, global ground state 
be the same as that on any state of form~$\Pi(K)\ket \psi$ 
where $\ket \psi$ is arbitrary, bearing no reference to $H_\text{w/bd}$.
Indeed, since~$K'$ belongs to the algebra generated by operators, 
each of which takes a definite eigenvalue on $\Pi(K)$,
we see that 
\begin{align}
\Pi(K) \cdot K \cdot \Pi(K) = \Pi(K) \cdot K' \cdot \Pi(K) = c(K) \Pi(K) \label{eq:KLcondition}
\end{align}
for some scalar~$c(K)$.

We have not specified how ``small'' or ``not-too-big'' subsystems have to be where $O$ or $K$ is supported.
In fact, we only need that the $\ell$-neighborhood of a region we consider contains only null-homologous (co)cycles,
where $\ell$ is the maximum diameter of the support of a Hamiltonian term.
Since~\cref{eq:KLcondition} is the Knill-Laflamme criterion for error correction~\cite{KnillLaflamme1996},
the ground state subspace of~$H_\text{w/bd}$ obeys
the ``homogeneous'' topological order condition~\cite{Haah2020},
a geometry-free form of an error correction property.

Let us prove the claim.
Given an operator~$O$ that commutes with every Hamiltonian term,
we expand~$O$ in the Pauli operator basis,
\begin{align}
	O &= \sum_j P_j \\
	0 =[O,\pi]&= \sum_j [P_j, \pi] \nonumber
\end{align}
where~$P_j$ are distinct tensor products of Pauli matrices,
and~$\pi$ is either~$\pi(v_0)$ or $\pi(a^1)$.
Each commutator, if nonvanishing, is a Pauli operator $2P_j \pi$,
and is thus different from any other appearing in the sum.
Since Pauli operators are an orthonormal operator basis (under Hilbert--Schmidt inner product),
we see that every commutator must vanish.
We can write each $P_j$ as a product of some Pauli~$Z$ operators and some Pauli~$X$ operators up to a phase.
The $X$ part of~$P_j$ must form a cycle because $P_j$ commutes with~$\pi(a^1)$ and~$\pi(v_0)$ for all~$a^1$ and~$v_0$.
Being supported on a ball-like region, the cycle must be null-homologous.
But the fluctuation operators of~$H_\text{w/bd}$
explore all null-homologous cycles within the ball-like region,
and hence, the $X$ part can be canceled by multiplying $P_j$ by fluctuation operators.
This multiplication results in $P_j'$, a diagonal operator in $Z$ basis,
which may have a larger support than~$P_j$,
but not more than the size of the fluctuation operators.

Our next goal is to show that 
if $D$ is any diagonal operator on a ball-like region which commutes with every Hamiltonian term,
then it belongs to the algebra generated by Hamiltonian terms.
A diagonal operator~$D$ commutes with a fluctuation operator
if and only if $D$ commutes with the $X$ part of fluctuation operators.  
Now, because $Z$ is diagonal,
each Pauli summand~$D_k$ of~$D = \sum_k D_k$ in the Pauli basis expansion 
is a product of $Z$'s.  
By the same argument as before, 
each such $D_k$ must individually commute with the $X$ part of fluctuation operators.
Thus, each such $D_k$ may be viewed as a dual chain that is closed relative to the spatial boundary.
Because the support of $D_k$ is topologically trivial, 
this dual chain is exact.
However, the terms~$\pi(v_0)$ and~$\pi(a^1)$ 
are precisely the generator of null (co)homology cycles, 
relative to the spatial boundary,
which shows that each~$D_k$, and hence~$D$, belongs to the algebra generated by~$\pi$'s.

Therefore, each $P_j'$ belongs to the algebra of $\pi$'s, so
$P_j$ belongs to the algebra of $\pi$'s and the fluctuation terms,
and, finally, $O$ belongs to the algebra of the Hamiltonian terms.  Note that a similar prescription introduces a gapped boundary to the $2+1$d and any higher dimensional toric code.

\subsection{Operators for topological excitations}

\subsubsection*{$1$-dimensional operators}

Consider a secondary $1$-cycle~$a_1$ near the system boundary which may be homologically nontrivial.
Since we have a direction normal to the surface,
there is a homologous shift~$a_1'$ of~$a_1$ into the exterior of the system,
also near the surface.  So, $a_1$ is entirely in the interior of the system, and $a_1'$ is entirely in the exterior.

\paragraph{For bosonic charges.}

The null-homologous cycle~$a_1 + a_1'$ is the boundary of some $2$-chain~$b_2$,
and the product~$L_{Bc}(b_2)$ of~$L_{Bc}$ over~$b_2$ is an $X$ string along $a_1 + a_1'$
times $Z$ factors on the primary $2$-cells, 
each of which is Poincar\'e dual to cells of~$b_2$.
Let us strip off the $X$ factors of~$L_{Bc}(b_2)$ along~$a_1'$ and apply the qubitwise sandwiching on external qubits,
to obtain a system operator~$S_{Bc}(a_1)$.
This string operator~$S_{Bc}(a_1)$ commutes with every term of~$H^{WW}_{BcBl\text{ w/bd}}$;
the stripped off part from~$L_{Bc}(b_2)$ is not seen by terms of~$H^{WW}_{BcBl\text{ w/bd}}$.
When $a_1$ is homologically nontrivial,
$S_{Bc}(a_1)$ is not generated by terms of~$H^{WW}_{BcBl\text{ w/bd}}$.

If we truncate a long~$S_{Bc}(a_1)$, then at the end point the cycle-enforcing term is violated.
This excitation is topological and deconfined as witnessed by the following $3$-dimensional operator.
In the $4$-space in which our system lives,
there is a primary $3$-sphere~$r^3$ that encloses~$v_0$ but not~$v_0'$.
The product~$M(r^3)$ of secondary qubit operator~$Z$ over~$r^3$ 
is a product of cycle-enforcing terms~$\pi(v_0)$ over the $4$-ball bounded by~$r^3$,
and $M(r^3)$ anticommutes with~$S_{Bc}(a_1)$.
The sphere $r^3$ is in fact a $3$-dimensional hemisphere if we neglect the external qubits,
where the equatorial boundary is a $2$-sphere near the system boundary.
Since there is no restriction on the geometry (size) of these operators,
the excitation at~$v_0$ is topological.
We will later find a strictly $2$-dimensional operator 
that does not penetrate into the $4$-dimensional bulk.

Our construction of~$S_{Bc}(a_1)$ depends on~$a_1'$ and~$b_2$,
but the resulting operator is unique up to $\pi(a^1)$'s within the system
if the cellulation is so refined 
that the distance between~$a_1$ and its shift~$a_1'$ 
is smaller than the injectivity radius~$r_\text{inj}$ of the system boundary manifold.
Indeed, the difference~$\Delta S_{Bc}$ resulting from these arbitrary choices, 
before we take the qubitwise sandwiching,
is a product of primary-qubit~$Z$ which as a whole commutes with every term of~$H^{WW}_{BcBl}$.
We know that $\Delta S_{Bc}$ must be a product of terms of~$H^{WW}_{BcBl}$
because we have a Hamiltonian-disentangling circuit~\cref{eq:V}.
The group of all primary-qubit diagonal Pauli operators in the group of terms of~$H^{WW}_{BcBl}$
consists of~$L_{Bc}(c_2)$ for secondary $2$-cycles~$c_2$.
Since $\pi(a^1)$'s viewed as secondary $2$-chains generate all null-homologous secondary $2$-cycles,
it remains to show that if $\Delta S_{Bc} = L_{Bc}(c_2)$, then $c_2$ is null-homologous.
To see this, just note that the cycle~$c_2$ is contained in the $r_\text{inj}$-neighborhood of~$\Supp(a_1)$,
which deformation-retracts to~$\Supp(a_1)$,
that contains no nontrivial $2$-cycle, so $c_2$ is null-homologous.

The quasiparticle at the end of truncated string operator is a boson 
as seen by a T-junction exchange process in~\cref{eq:exchangestatistics};
any segment of $S_{Bc}$ commutes with any other.

\paragraph{For fermionic charges.}

For FcBl, a similar string operator~$S_{Fc}$ near the boundary 
exists.  It is likewise a product of~$X$ along a secondary cycle~$a_1$,
but is additionally decorated with $Z$ thorns.
The construction is the same as that for string operators of~$H^{WW}_{BcBl\text{ w/bd}}$;
take a pair of a secondary $1$-cycle within the system and its shift to the exterior,
pick a secondary $2$-chain bounded by the pair to multiply~$L_{Fc}$ over,
and read off the system part.
We have an extra property that the secondary part of~$S_{Fc}$ localizes along~$a_1$
by~\cref{lem:cancelthorns}.
The $Z$ factors on primary $2$-cells that were responsible for linking parities do not localize,
but they are unique up to~$\pi(a^1)$'s as we have shown in the construction of~$S_{Bc}$ above.
The string operator~$S_{Fc}$ 
that is first localized by~\cref{lem:cancelthorns} and then truncated by the system boundary,
is equivalent to one that is first truncated and then localized.
These two procedures commute 
because the cycle-enforcing terms~$\pi(v_0)$ that localize the string before the system boundary truncation
can be first truncated and then localize the string operator.
The same string operators are applicable for FcFl;
we have used the worldline fluctuation operators only
which are the same for both FcBl and FcFl.

The string operators of FcBl indeed transport fermions 
as one can easily check by a T-junction process with~\cref{eq:exchangestatistics}.
A concrete expression for string operators on a hypercubic lattice is shown in~\cref{fig:stringop}(c).
Note that the fourth direction (into the $4$-dimensional bulk) is irrelevant 
for the evaluation of exchange statistics of boundary quasiparticles,
because any potential $Z$ thorns in the fourth direction will commute with all other operators, as there are no $X$'s in the fourth direction.  Also, factors of~$Z$ on primary qubits are not relevant for this computation.

\subsubsection*{$2$-dimensional operators}

Let $b^2$ be a primary $2$-cycle that sits within the system boundary,
which may represent a nontrivial homology of the system boundary.
Similarly to the string operator construction above,
we consider a translate~$(b^2)'$ of~$b^2$ along the direction normal to the system boundary
towards the exterior.
We know that $b^2 + (b^2)' = \bd c^3$ for some primary $3$-chain~$c^3$.
Before the qubitwise sandwiching,
the product of~$T_{Bl}$ over~$c^3$ 
is a Pauli operator with $X$ factors on $b^2 \sqcup (b^2)'$
and some $Z$ factors on secondary $1$-cells in between.
The product of~$T_{Fl}$ over~$c^3$ is more complicated,
but only in the diagonal part.
That is, for $T=T_{Bl}$ and $T=T_{Fl}$ we have
\begin{align}
	\prod_{a_1 \in c^3} T(a_1)
	= \left(\prod_{a_1 \in c^3} Z(a_1)\right) 
	\cdot D \cdot 
	\left(\prod_{f^2 \in b^2} X(f^2) \right) \cdot
	\left( \prod_{g^2 \in (b^2)'} X(g^2) \right) \label{eq:c3T}
\end{align}
where $D$ is some operator on the primary qubits which is diagonal in $Z$ basis.
In the case of $T_{Bl}$, $D$ is the identity.  
This follows from the fact that each term~$T(a_1)$ 
maps a configuration of primary $2$-cycle to another with some sign,
where the change in the cycle is the boundary of primary $3$-cell, Poincar\'e dual to~$a_1$.
The change in cycles is expressed by the $X$ factors in~\cref{eq:c3T}.
To define membrane operators, we strip off the $X$ factors on~$(b^2)'$ 
and apply the qubitwise sandwiching on the remaining operator:
\begin{align}
	M(b^2) = 
	\bra{0^{\otimes \text{external}}} 
		\left(\prod_{e_1 \in c^3} Z(e_1)\right) D
	\ket{0^{\otimes \text{external}}}
	\left(\prod_{f^2 \in b^2} X(f^2) \right).
\end{align}
Starting with~$T_{Bl}$, we obtain~$M_{Bl}(b^2)$
that consists of $X$ factors on~$b^2$ that is within the system boundary
and some $Z$ factors on secondary $1$-cells, also within the system boundary.
With~$T_{Fl}$, the operator~$M_{Fl}(b^2)$ differs from~$M_{Bl}(b^2)$ only in the diagonal operator on primary qubits.
For either case, Bl or Fl,
on the hypercubic lattice $4$-space and the cubic lattice boundary,
no $Z$ factors on secondary $1$-cells survive after qubitwise sandwiching.

The membrane operator~$M(b^2)$ is $2$-dimensional
and commutes with every Hamiltonian term of~$H_\text{w/bd}$.
The commutativity is obvious for the worldsheet fluctuation operators in these Hamiltonians.
The part of~$M_{Bl}(b^2)$ seen by the worldline fluctuation operators
is unchanged under stripping and qubitwise sandwiching, 
so the commutativity is retained.
The constructed operator~$M(b^2)$ 
complements the $1$-dimensional string operators we constructed above;
the anticommutation relations come from the primary-qubit $X$ factors of~$M(b^2)$ and $Z$ factors of~$S(a_1)$.
In other words, the membrane operators over primary homology $2$-cycles 
and the string operators over secondary homology $1$-cycles
are a generating set for the algebra acting on the ground state subspace.
Here, the homologies are those of the boundary $3$-manifold.

\subsection{Loop self-statistics $\mu=-1$ on boundary of the FcFl model}

Any truncation of~$M_{Bl}$ satisfies the requirements of~\cref{subsec:data} to compute the loop self-statistics~$\mu$,
since $M_{Bl}$ is a tensor product of single-qubit unitary operators.
It is obvious that $\mu_{Bl} = +1$,
since $X$ factors are on primary cells and $Z$ factors are on the secondary cells,
so any truncated versions commute with one another.

To compute the loop self-statistics~$\mu_{Fl}$ with $M_{Fl}$,
we resort to a representation of~$M_{Fl}$ inside the boundary theory of a model Fc$\times$FcBl,
which will be shown to be isomorphic by a shallow quantum circuit to Fc$\times$FcFl.
Here, Fc is a topologically ordered model in $4+1$d defined in~\cref{rem:Fc}
that has a fermionic charge.
We will employ the circuits developed in~\cref{sec:disentangle2FcFl}.

\paragraph{Fc$\times$FcBl $\cong$ Fc$\times$FcFl.}

Recall that the amplitude~$\braket{b^2 , a_1 | \text{FcFl}}$ is given
by the product of three signs:
the linking between~$a_1$ and~$b^2$,
the frame parity of~$a_1$,
and the frame parity of~$W^1(b^2)$.
On the other hand, the amplitude~$\braket{b^2 , a_1 | \text{FcBl} }$ is given
by the product of two signs:
the linking between~$a_1$ and~$b^2$
and the frame parity of~$a_1$.
For a fixed primary homology $1$-cycle~$h^1$,
the amplitude~$\braket{x^1 | \text{Fc};[h^1]}$ with $x^1 \in [h^1]$
is given by the frame parity~$\overline\Frm(x^1)$ as constructed in~\cref{prop:globalFrm}.
Given a projection from our $2$-skeleton down to $\RR^2$,
$\overline\Frm$ may not be unique, precisely when the degree-$1$ homology of the $2$-skeleton is nonzero.
We choose an arbitrary $\overline\Frm$ extending~$\Frm$.
As in~\cref{sec:disentangle2FcFl} we consider
\begin{align}
\braket{x^1 , b^2 , a_1 ~|~ \text{Fc}\times\text{FcFl}}
&=
(-1)^r \braket{x^1+W^1(b^2) , b^2 , a_1 ~|~ \text{Fc}\times\text{FcBl}}\\
\text{where }\quad 
r
&=
\overline\BiFrm(x^1, W^1(b^2)) . \nonumber
\end{align}
We know from~\cref{lem:circuitR} that $r: (y^1,b^2) \mapsto \overline\BiFrm(y^1, W^1(b^2))$
where $y^1$ is any primary $1$-cycle and $b^2$ is any primary $2$-cycle,
is computable by a shallow circuit~$R$.
We also know  that
\begin{align}
\Omega|_\text{cycles} = \sum_{x^1 \in [h^1], b^2 \in [0], a_1 \in [0]}
\ket{x^1+W^1(b^2) , b^2 , a_1}\bra{x^1 , b^2 , a_1}
\end{align}
is implemented by a shallow circuit,
where $\Omega$ does not change the homology class of the primary $1$-cycle
since $W^1$ is always null-homologous.
Therefore, we have a shallow-circuit equivalence
\begin{align}
\Omega R \ket{\text{Fc}\times\text{FcFl}} = \ket{\text{Fc}\times\text{FcBl}}
\end{align}
regardless of what superselection sector the ground state $\ket{Fc}$ is in.

In the presence of spatial boundary,
the circuit~$R$ has gates that straddle system qubits and external qubits;
there could be an external primary $2$-cell that 
meets a system primary $2$-cell along a system $1$-cell, 
and some gate of $R$ needs to access both $2$-cells to infer~$W^1$.
However, being a diagonal circuit, $R$ becomes a unitary circuit
under the qubitwise sandwiching, 
which is nothing but restriction of domain of the unitary gates,
and hence the sandwiched circuit locally computes~$\overline \BiFrm$.
The layer of $\Omega$ does not have any gates that straddle system qubits and external qubits,
and $\Omega$ restricts to a shallow circuit under the qubitwise sandwiching.
Hence, the shallow-circuit equivalence between Fc$\times$FcFl and Fc$\times$FcBl
continues to hold in the presence of spatial boundary.

\paragraph{Representation of $T_{Fl}$ in the bulk.}

We will find a representation of open membrane operators of FcFl
via operators in Fc$\times$FcBl.
(The worldsheet fluctuation terms~$T_{Fl}$ of $H^{WW}_{FcFl}$ were rather implicit, 
which is why it is easier to work with $H^{WW}_{FcBl}$.)
Before we consider open membrane operators of~FcFl,
we first find a representation of the worldsheet fluctuation operator~$T_{Fl}$.
Then, we will be able to read off an open membrane operator representation.

Suppose there is a local orientation for each primary $2$-cell.
Consider an operator
\begin{align}
	T'(g^2) &\equiv X(g^2) S(W^1(g^2)) \label{eq:Tprime}
\end{align}
associated with each primary $2$-cycle~$g^2$,
where $X(g^2) \equiv \prod_{f^2 \in g^2} X(f^2)$ 
is a product of Pauli~$X$ over all $2$-cells of~$g^2$.
Such an operator gives
\begin{align}
	T'(g^2) \ket{x^1, b^2} (-1)^{\overline\Frm(x^1)}
	&= 
	\ket{x^1 + W^1(g^2), b^2 + g^2}(-1)^{\overline\Frm(x^1)+\Frm(W^1(g^1))+\overline\BiFrm(x^1,W^1(g^2))}\nonumber\\
	&=
	\ket{x^1 + W^1(g^2), b^2 + g^2}(-1)^{\overline\Frm(x^1+W^1(g^1))}.
\end{align}
Then, we will have
\begin{align}
&
R \Omega T'(g^2) \Omega R \ket{x^1,b^2} (-1)^{\overline\Frm(x^1) + \Frm(W^1(b^2))} \nonumber\\
=&
R \Omega T'(g^2) \Omega \ket{x^1, b^2} (-1)^{\overline\Frm(x^1+W^1(b^2))}\nonumber\\
=&
R \Omega T'(g^2) \ket{x^1 + W^1(b^2), b^2} (-1)^{\overline\Frm(x^1+W^1(b^2))}\\
=&
R \Omega \ket{x^1 + W^1(b^2)+W^1(g^2), b^2 + g^2} (-1)^{\overline\Frm(x^1+W^1(b^2)+W^1(g^2))}\nonumber\\
=&
R \ket{x^1 + W^1(b^2)+W^1(g^2) + W^1(b^2+g^2), b^2 + g^2} (-1)^{\overline\Frm(x^1+W^1(b^2)+W^1(g^2))}\nonumber\\
=&
\ket{x^1 + W^1(b^2)+W^1(g^2) + W^1(b^2+g^2), b^2 + g^2} \nonumber\\
&\qquad\cdot (-1)^{\overline\Frm(x^1+W^1(b^2)+W^1(g^2))+\overline\BiFrm(x^1 + W^1(b^2)+W^1(g^2), W^1(b^2 + g^2))}\nonumber\\
=&
\ket{x^1 + W^1(b^2)+W^1(g^2) + W^1(b^2+g^2), b^2 + g^2} \nonumber\\
&\qquad \cdot (-1)^{\overline\Frm(x^1+W^1(b^2)+W^1(g^2)+W^1(b^2+g^2))+\overline\Frm(W^1(b^2 + g^2))}\nonumber
\end{align}
where in the second from the last equality
we use $\overline\BiFrm(u^1,u^1) = 0$ for any $1$-cycle~$u^1$.
That is,
\begin{align}
(\Omega R)^\dagger T'(g^2)(\Omega R) &\ket{x^1,b^2} (-1)^{\overline\Frm(x^1) + \Frm(W^1(b^2))} \nonumber\\
= &\ket{(x^1)', b^2 + g^2 } (-1)^{\overline\Frm((x^1)') + \Frm(W^1(b^2+g^2))} .
\end{align}
where $(x^1)' = x^1 + W^1(b^2)+W^1(g^2) + W^1(b^2+g^2)$.
Since $(x^1)'$ is always homologous to~$x^1$,
if we sum over all $1$-cycles in the homology class~$[x^1]$,
we have
\begin{align}
	(\Omega R)^\dagger T'(\bd c^3)Z(c^3)(\Omega R) &\ket{\text{Fc}} \ket{b^2, a_1} (-1)^{\Frm(W^1(b^2))} \nonumber\\
	= &\ket{\text{Fc}}\ket{b^2 + \bd c^3, a_1} (-1)^{\Frm(W^1(b^2+\bd c^3))+\Int(c^3,a_1)}.
\end{align}
This is the correct relation of amplitudes in the ground state~$\ket{\text{FcFl}}$.
Therefore, $T'(\bd c^3)Z(c^3)$ is a correct representation of the worldsheet fluctuation operator~$T_{Fl}$
acting on the state of form~$\ket{\text{Fc}}\ket{b^2,a_1}$ where $b^2, a_1$ are cycles.
Note that it is important that the first factor~$\ket{\text{Fc}}$ is the ground state of~$H_{Fc}$.

We wish to express $T'(g^2)$ as a product of operators~$O$, each of which is associated with a $2$-cell:
\begin{align}
	T'(g^2) \equiv X(g^2) S(W^1(g^2)) = \pm \prod_{f^2 \in g^2} O(f^2) \label{eq:repr-opreq}
\end{align}
where we allow for the product to be ambiguous up to a sign. 
The challenge is to find operators $O(f^2)$ \emph{and} a local orientation of all $2$-cells~$f^2$
such that the product of $O(f^2)$ over an arbitrary null-homologous $2$-cycle~$g^2$ gives $S(W^1(g^2))$ 
as in~\cref{eq:repr-opreq}.  To this end,
we assume that the primary cellulation is a triangulation.
We use a set of coefficients~$[xyz]=[zyx] \in \ZZ_2$ where $x,y,z$ are any distinct primary $0$-cells
subject to conditions that 
\begin{align}
[xyz]+[yzx]+[zxy] = 0,\label{eq:t1conditions}\\
[xwy]+[ywz]+[zwx] = 1\nonumber
\end{align}
for any distinct~$x,y,z,w$.
These coefficients are similar to those used in~\cref{sec:againstChargeDecoration} with $t=1$,
but previously the coefficients~$[xyz]$ were defined over five vertices.
In~\cref{app:consistency} we show that there exists a solution to \cref{eq:t1conditions}
over arbitrarily many vertices,
and any two solutions differ by~$\Delta[xyz] = [xy] + [yz]$
where $[xy]=[yx]\in \ZZ_2$ are free parameters.
Ignoring secondary qubits for now,
we define an open membrane operator~$O(xyz)$ for any primary triangle~$xyz$:
\begin{align}
	O(xyz) \equiv X(xyz) S(xy)^{ [xyz] + \cnf(xy) } S(xz)^{ [xzy] + \cnf(xz) } S(yz)^{\cnf(yz)}
	\label{eq:Odef}
\end{align}
where $S(xy) = S(yx)$ are fermion string segments on primary $1$-cell~$xy$ 
constructed in~\cref{rem:fermionstringsegments},
and $\cnf(xy)=1$ if and only if the underlying configuration of primary $2$-chain has $\ZZ_2$-boundary at~$xy$.
The form of~$O(xyz)$ is virtually the same as~$\tilde M$ of~\cref{sec:againstChargeDecoration},
but here a triangle~$xyz$ ranges over all primary $2$-cells,
as opposed to the case of $\tilde M$, where we only looked at a finite, small number of triangles.
Every fermion string segment is a product of Pauli~$X$ on~$xy$ and some $Z$s around it and squares to~$+1$.
Due to~$\cnf(xy)$ in the exponents, the open membrane operator~$O(xyz)$ is not a product of Pauli operators.
Note that there is a distinguished vertex~$x$ in~\cref{eq:Odef}.  We assume that a distinguished vertex is chosen arbitrarily for each triangle.

The open membrane operators commute up to a sign regardless of distinguished vertices of triangles:
\begin{align}
O(xyz)O(yzw)O(xyz)^{-1}O(yzw)^{-1} = \pm 1.
\end{align}
So, the product of~$O$ over a primary $2$-cycle is well defined up to a sign.
Now, we claim that $O(f^2 = xyz) = O(xyz)$ satisfies~\cref{eq:repr-opreq}.
We will prove this claim below.

\paragraph{A local orientation for triangles.}

In the definition~\cref{eq:Odef} of open membrane operators,
we have not specified orientation for triangles;
the choice of a distinguished vertex of a triangle has nothing to do with the orientation of the triangle.
Hence, our goal is to define a local orientation for all triangles
such that \cref{eq:repr-opreq} is met.  There is essentially a unique local orientation for triangles such that the operator $O(g^2)$ that we have defined (up to sign) takes the form of \cref{eq:repr-opreq}, 
{\em i.e.}, has the fermion string operators precisely on the $W^1(g^2)$ line determined by the orientation.  If $xyz$ and $yzw$ are the only triangles of a $2$-cycle~$b^2$ which touch edge~$yz$,
then other open membrane operator factors are ``too far'' from $yz$.
So, the product $O(xyz)O(yzw) = \pm O(yzw) O(xyz)$ has to act 
by the fermion string operator on~$yz$
if and only if the orientations of $xyz$ and $yzw$ do not agree along~$yz$.
Put differently, 
each of the orientations of $xyz$ and $yzw$ induces an ordering on the edge~$yz$,
and the fermion string operator has to be produced upon multiplication of the two open membrane operators
if and only if these induced orderings on~$bc$ are the same.
Given such a local orientation, \cref{eq:repr-opreq} immediately follows.
Note that $O(xyz)^2 = \pm S(xy)S(yz)S(zx)$ 
is the closed fermion string operator along the boundary of triangle~$xyz$.

To find the requisite local orientation,
we define 
\begin{align}
\chi(xyz|yzw) = \begin{cases}
1 & \text{if }\bra{0(yz)} O(xyz)O(yzw) \ket{0(yz)} = 0,\\
0 & \text{otherwise}
\end{cases}
\end{align}
where $\ket{0(yz)}$ is an eigenstate of single-qubit Pauli~$Z$ on edge~$yz$.
The product~$O(xyz)O(yzw)$ may have $S(yz)$ factor,
in which case the factor gives an $X$ factor on edge~$yz$
and \mbox{$\chi(xyz|yzw) =1$.}
Obviously, $\chi$ is symmetric, {\em i.e.}, $\chi(yzw|xyz) = \chi(xyz|yzw)$.
We find it convenient to define $\chi$ for arbitrary pairs of triples of vertices.
Each string operator is some operator acting on an edge,
and for the definition of~$\chi$
we use neither the fact that the string operator transports a fermion,
nor that it is a tensor product of Pauli operators.
With this abstraction,
although we have only defined the open membrane operators~$O(xyz)$
for primary $2$-cells~$xyz$,
we may use the same formula~\cref{eq:Odef} to define $O$ for an arbitrary triple of vertices,
regardless of their distances.
This can be regarded as an embedding of our primary $2$-skeleton into a very high dimensional simplex.
To make sense of~\cref{eq:Odef} we have to introduce a qubit, unless it is already there,
for each triple (triangle) of vertices and for each pair (edge) of vertices.

Then, for any distinct vertices $x,y,z,u,v,w$, we will show that
\begin{align}
	\chi(uvx|uvy) + \chi(uvy|uvz) + \chi(uvz|uvx) &= 1 &\text{(edge)}\label{eq:Ocond}\\
	\chi(wxy|wyz) + \chi(wyz|wzx) + \chi(wzx|wxy) &= 0 &\text{(vertex)}\nonumber
\end{align}
Note here that the vertices in the argument of $\chi$ may be arbitrarily far apart.
Once we show this, \cref{lem:alignment} will imply that there exists a local orientation for all triangles
such that $S(yz)$ appears as a factor in the product $O(xyz)O(yzw)$
if and only if the orientations of $xyz$ and $yzw$ do not agree along~$yz$.
This is what we have sought.

Now, we prove~\cref{eq:Ocond} from the definition~\cref{eq:Odef} of~$O$ using~\cref{eq:t1conditions}.
We put a dot above a vertex symbol, {\em e.g.}, $\dot x y z$,
to denote that the dotted vertex~$x$ is distinguished for~\cref{eq:Odef}.
It is not difficult to enumerate all cases to check~\cref{eq:Ocond},%
\footnote{
The following enumerates all dotting possibilities
for a vertex-sharing triple of triangles without an edge common to all three, up to re-labeling:
$	\dot wxy, \dot wyz, \dot wzx$ /
$	\dot wxy, \dot wyz, w\dot zx$ /
$	\dot wxy, w\dot yz, wz\dot x$ /
$	\dot wxy, w\dot yz, w\dot zx$ /
$	\dot wxy, wy\dot z, w\dot zx$ /
$	w\dot xy, w\dot yz, w\dot zx$ /
$	wx\dot y, w\dot yz, w\dot zx$.
}
but the vertex condition follows without any further calculation.
Since relevant edges~$wx,wy,wz$ in the vertex condition are distinct,
we simply consider the product $O(wxy)O(wyz)O(wzx)O(xyz) = O(\bd (wxyz))$
where $wxyz$ is a $3$-simplex.
The calculation of~\cref{sec:againstChargeDecoration}
implies that if we multiply $O$ one at a time to cover $\bd(wxyz)$,
the loop~$\ell$ along the boundary of any intermediate $2$-chain
contains an even number of charges,
and as $\ell$ disappears eventually,
all charges must disappear, too.
This means that the remaining string operator in $O(\bd(wxyz))$ must be $\ZZ_2$-closed,
which implies the vertex condition of~\cref{eq:Ocond}.
As for the edge condition,
the following enumerates all dotting possibilities for edge-sharing triples of triangles:
\begin{align}
	uv\dot x, uv\dot y, uv\dot z;\qquad
	uv\dot x, uv\dot y, \dot u v z;\qquad
	uv\dot x, \dot u v y, \dot u v z;\\
	uv\dot x, \dot u v y, u \dot v z;\qquad
	\dot u v x, \dot u v y, \dot u v z;\qquad
	\dot u v x, \dot u v y, u \dot v z. \nonumber
\end{align}
It is straightforward to prove the edge condition%
\footnote{
This is the step where we really need the second condition~($t=1$) of~\cref{eq:t1conditions};
if $[xwy]+[ywz]+[zwx] = 0$, then the edge condition would not hold.
}
by checking all of the above possibilities, using
\begin{align}
	\chi(uv\dot x| uv \dot y) &= 1, & \chi(uv \dot x| \dot u v y) &= 1 + [uvy],\\
	\chi(\dot u v x| u \dot v y) &= 1 + [uvx]+[vuy],& \chi(\dot uvx | \dot uvy) &= 1+[uvx]+[uvy] = [xvy].\nonumber
\end{align}

\paragraph{$\mu_{Fl} = -1$.}

We have found a local orientation for all primary $2$-cells and
shown that the open membrane operators~$O$ 
multiply over a $2$-cycle to leave a string operator along the $W^1$ of the $2$-cycle,
implying that $O(b^2)$ of~\cref{eq:Odef} fulfills~\cref{eq:repr-opreq}.
Hence, we have found an explicit representation of~$M_{Fl}$
in the presence of an auxiliary ground state~$\ket{\text{Fc}}$.
If we change the orientation of some triangle,
then the open membrane operator at that triangle 
is multiplied by unconditional closed fermion string operator along its perimeter.
Therefore, we have obtained a set of (big) open membrane operators~$M_{ij}$ satisfying 
all requirements of~\cref{subsec:data} ---
simply multiply $O$ over some large triangle.

Then, applying our indicator definition in~\cref{subsec:indicator},
we see that nontrivial commutators among open membrane operators
come from nontrivial commutators among decorating open fermion string operators.
The overall effect of these nontrivial commutators of fermion string operators
is calculated in~\cref{fig:decorated3},
and the result is that $\mu_{Fl} = -1$.

\section{Discussion}\label{sec:discussion}

In this work we defined an invariant of 3+1d fermionic $\ZZ_2$ gauge theories, 
the loop self-statistics~$\mu$.
We demonstrated that $\mu$ can be measured using a process that rotates the gauge flux loop 
by an angle of $\pi$ around its diameter, 
in a way that is independent of the arbitrary choices made in the process.  
In other words, all of the nonuniversal geometric phases cancel out.  
Crucially, the gauge charge has to be a fermion in order for~$\mu$ to be well defined; 
if the gauge charge is a boson, $\mu$ is only defined up to sign.  

For the ordinary 3+1d fermionic toric code, we saw rather easily that~$\mu=+1$.
Less trivially, we also constructed a 4+1d model~$H^{WW}_{FcFl}$ (the FcFl model)
which realizes a $\mu=-1$ fermionic $\ZZ_2$ gauge theory on its boundary.
We demonstrated that this FcFl model is invertible 
by explicitly disentangling two stacked copies of it into a product state.
Further, we argued that any stand-alone 3+1d fermionic $\ZZ_2$ gauge theory must have~$\mu=+1$,
for otherwise we would be able to construct a nontrivial invertible fermionic phase in 3+1d,
which is believed not to exist.  
Hence, the FcFl model in 4+1d realizes a $\ZZ_2$ classified nontrivial bosonic phase of matter, 
stable without any symmetry.  
We argued that a $\ZZ_2$ fermionic gauge theory with~$\mu=-1$ 
can be obtained by Higgsing all-fermion QED, 
so, by a standard argument, 
the FcFl model can realize all-fermion QED on its boundary.

Let us now discuss the connection of our model to the work of~\cite{Freyd_2020} in a little more detail.  Our exactly solved models have two types of degrees of freedom: those living on secondary $1$-cells, and those living on primary $2$-cells.  These can be interpreted as $4-1=3$-form and $4-2=2$-form gauge fields respectively.  Ref.~\cite{Freyd_2020} refers to these as $E$ and $M$ respectively, and shows in \cite[Eq.~(7)]{Freyd_2020} that a general $3+1$ dimensional action for these is given by $H^5(K(\mathbb Z_2,2) \times K(\mathbb Z_2,3),{\rm U}(1))$, which by the K\"unneth formula is equal to a direct sum of the three tensor factors $H^5(K(\mathbb Z_2,3),{\rm U}(1)) \simeq \mathbb \Z_2$, $H^3(K(\mathbb Z_2,3),{\rm U}(1)) \otimes H^2(K(\mathbb Z_2,2),{\rm U}(1)) \simeq \mathbb \Z_2$ and $H^5(K(\mathbb Z_2,2),{\rm U}(1)) \simeq \mathbb \Z_2$.  These are generated by the classes $(-1)^{{\rm{Sq}}^2 E}, (-1)^{EM}, (-1)^{M {\rm{Sq}}^1 M}$ respectively.  The first of these determines whether the charge is a boson (as in the BcBl model) or a fermion (as in the FcBl and FcFl models).  The second class just states that the charge braids nontrivially with the magnetic flux loop.  Finally, the last class, $(-1)^{M {\rm{Sq}}^1 M}$, is just equal to the loop self-statistics $\mu$.

We would also like to make the following observation, noted by a referee of an earlier draft of this work.  Our loop self-statistics invariant can be written in a somewhat more symmetric way by re-ordering certain terms in the process.  The result can be checked to still be universal and equal to our invariant.  The expression is:

\begin{align*}
{\bf M} = &M_{24}[M_{13},M_{23}]M_{14}^{-1} \cdot M_{12}[M_{34},M_{13}]M_{24}^{-1} \cdot M_{23}[M_{14},M_{34}]M_{12}^{-1} \\
&M_{13}[M_{24},M_{14}]M_{23}^{-1} \cdot M_{34}[M_{12},M_{24}]M_{13}^{-1} \cdot M_{14}[M_{23},M_{12}]M_{34}^{-1}
\end{align*}

There are several avenues for future exploration.  
First, it would be nice to have a direct argument that $\mu=+1$ in any stand-alone 3+1d model, 
{\em i.e.},
one which does not rely on the classification of invertible fermionic phases in 3+1d.
In fact, it is likely sufficient to make a weaker assumption: 
namely, that any putative invertible fermionic phase in 3+1d admits a gapped 2+1d boundary 
that is possibly topologically ordered.  
For, in that case, one can cut a ball out around the $0$ vertex in~\cref{fig:invariant_figure},
large compared to the correlation length.  
This ensures that the only locations in which the circuit calculating~$\mu$ 
is potentially nontrivial are along nonintersecting edges from the outer vertices to the surface of the ball,
and by examining the process defining~$\mu$, 
one sees that all of these cancel.  
Thus, in a sense, the circuit calculating~$\mu$ localizes onto the $0$ vertex.  
Another potential approach to showing that $\mu=+1$ in any stand-alone 3+1d model 
is to argue that, if $\mu=-1$ say, then an open membrane operator on a disc 
has to square to an operator on the disc's boundary 
which is the worldline operator for a gauge neutral fermion.
But, such a particle does not exist in the theory, 
leading to a contradiction.
It remains to make any of these arguments fully rigorous.

Another avenue for further exploration is to construct continuum versions of these models, 
and in particular to find an action formulation.
In fact, thinking about 2d world-sheets embedded in 4d 
and the various ways one can put Pin structures on these 
motivated much of our early thinking on this problem.  
We would expect any such model to be closely related to that discussed in~\cite{Wang_2019}.
It would also be nice to make an explicit connection 
between our FcFl model and the $2$-form Chern-Simons field theory description 
of the 4+1d bulk that realizes all-fermion QED on its boundary~\cite{Swingle_2015}.

\begin{acknowledgments}
We thank Michael Freedman for useful conversations about Pin structures and Whitney umbrella singularities.  We thank Chong Wang for an interesting discussion.  LF is supported by NSF DMR-1939864.
\end{acknowledgments}

\appendix

\section{Details in defining the loop self-statistics}\label{sec:app_details}

\subsection{Proof of \cref{eq:hMprime}}\label{app:pfMprime}

Consider two different choices of membrane operators~$M_{ij}$ and~$M'_{ij}$.
Note that our convention is that $i<j$, 
which will be notationally important in the following discussion.  
Let us write~$M_{ij} = M_{ij;\text{int}} M_{ij;\text{bdry}}$ 
where $M_{ij;\text{int}}$ acts in the interior of the plaquette~$(ij0)$ 
and $M_{ij;\rm{bdry}}$ acts near the boundary.
Similarly, $M'_{ij} = M'_{ij;\text{int}} M'_{ij;\text{bdry}}$.  
Then, $M'_{ij;\text{int}} M_{ij;\text{int}}^{-1}$ is a membrane operator 
which preserves the ground state reduced density matrix inside the plaquette~$(ij0)$
and creates some topologically trivial excitations at its boundary.
By virtue of being topologically trivial,
these can be removed by a shallow circuit $G_{\rm{bdry}}^{-1}$ acting near this boundary, 
so the operator 
\begin{align}
	F_{ij}^{ij0} \equiv G_{\rm{bdry}}^{-1} M'_{ij;\rm{int}} M_{ij;\rm{int}}^{-1}
\end{align}
has support on the interior of the plaquette~$(ij0)$ 
and maps the ground state to itself up to a $U(1)$ phase.
In fact, since all of the configuration states $|\cnf\rangle$ 
look like the ground state in the interior of the plaquette~$(ij0)$, 
all of these $|\cnf\rangle$ are eigenstates of $F_{ij}^{ij0}$ 
with the same eigenvalue.
Now,
\begin{align}
F_{ij} 
&\equiv M'_{ij} M_{ij}^{-1} 
\nonumber\\
&= 
M'_{ij;\rm{int}} M'_{ij;\rm{bdry}}M_{ij;\rm{bdry}}^{-1} M_{ij;\rm{int}}^{-1} 
\\&= 
\left(M'_{ij;\rm{int}}M'_{ij;\rm{bdry}}M_{ij;\rm{bdry}}^{-1}(M'_{ij;\rm{int}})^{-1} \right) \left(M'_{ij;\rm{int}} M_{ij;\rm{int}}^{-1}\right) 
\nonumber\\ &=  
\left(M'_{ij;\rm{int}}M'_{ij;\rm{bdry}}M_{ij;\rm{bdry}}^{-1}(M'_{ij;\rm{int}})^{-1} G_{\rm{bdry}} \right) F_{ij}^{ij0}
\nonumber
\end{align}
The operator in parentheses of the last line is a shallow circuit
supported near the boundary of the plaquette~$(ij0)$.
Since every~$\ket \cnf$ is an eigenstate of~$F_{ij}^{ij0}$
and
any state of form~$M_{ij} \ket{\cnf'}$ is an eigenstate of~$F_{ij}$,
we see that any state of form~$M_{ij} \ket{\cnf'}$ is an eigenstate of this circuit.
We can partition this circuit as
\begin{align} \label{eq:a3}
	\left(M'_{ij;\rm{int}}M'_{ij;\rm{bdry}}M_{ij;\rm{bdry}}^{-1}(M'_{ij;\rm{int}})^{-1} G_{\rm{bdry}} \right)
	=
	F_{ij}^0 F_{ij}^i F_{ij}^j F_{ij}^{i0} F_{ij}^{j0} F_{ij}^{ij} 
\end{align}
where $F_{ij}^{i0}$, $F_{ij}^{j0}$, and~$F_{ij}^{ij}$
are supported on neighborhoods of the interiors of edges~$(i0)$, $(j0)$ and~$(ij)$, respectively, 
and $F_{ij}^0$, $F_{ij}^i$, and~$F_{ij}^j$ 
are supported on neighborhoods of the vertices~$0$, $i$, and~$j$, respectively, 
as illustrated in \cref{fig:Fij_figure}.
Furthermore, we can ensure that, 
first, the states $|\cnf\rangle$ are eigenstates of $F_{ij}^{ij}$ 
with a common eigenvalue~$f_{ij}^{ij}$,
and second, the states~$|\cnf\rangle$ are 
eigenstates of~$F_{ij}^{i0}$ with eigenvalues~$f_{ij}^{i0}(\cnf)$ 
that depend on~$\cnf$ only through the occupation number of the edge~$(i0)$, 
and similarly for $F_{ij}^{j0}$.
Indeed, we can pick any partition we like, 
and then modify $F_{ij}^{i0}$, $F_{ij}^{j0}$, and $F_{ij}^{ij}$ locally 
near the endpoints of the corresponding intervals as needed 
to remove local excitations.  
Note that if there are gauge charges left at any of the endpoints, 
they will be left behind at \emph{both} endpoints of all of these operators; 
in this case we can multiply $G_\text{bdry}$ by a closed gauge string operator 
surrounding the plaquette to get to a situation where there are no gauge charges left at the endpoints.

Thus, finally we have
\begin{align} 
	F_{ij} 
	= 
	F_{ij}^0 F_{ij}^i F_{ij}^j F_{ij}^{i0} F_{ij}^{j0} F_{ij}^{ij} F_{ij}^{ij0} 
	\label{eq:Fdecomp}
\end{align}

\begin{figure}
\centering
\includegraphics[width=2.0in]{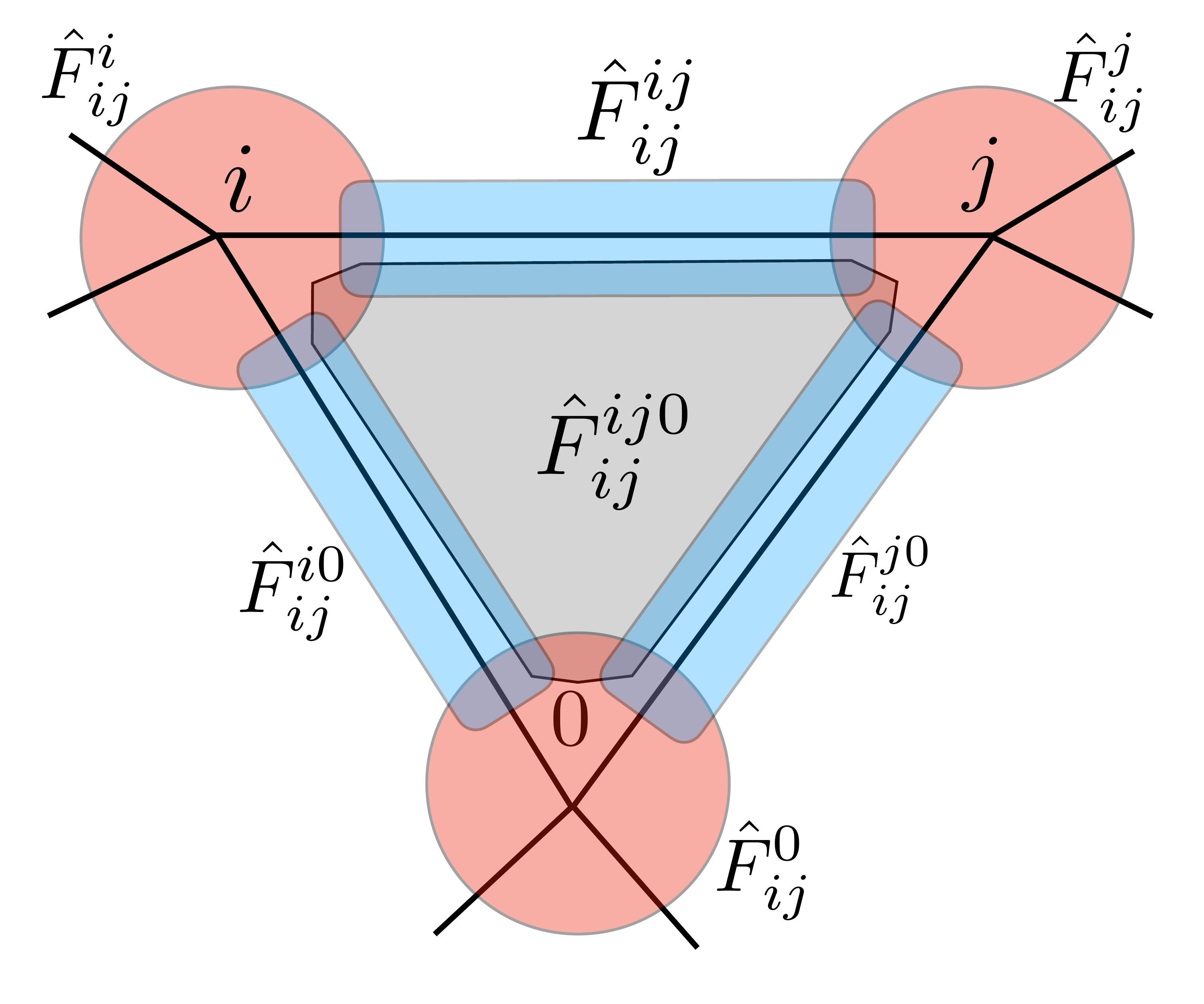}
\caption{}
\label{fig:Fij_figure}
\end{figure}

Now, suppose that $|\cnf\rangle$ is a configuration state 
that gets acted on by~$M_{ij}$, or is the result of acting with $M_{ij}^{-1}$, 
in the expression for $\bM$.
Then the above facts taken together imply that $M_{ij} |\cnf\rangle$, 
which is equal to some other configuration state~$|\cnf'\rangle$ up to a $U(1)$ phase, 
is an eigenstate of all seven of the operators appearing on the right-hand side of \cref{eq:Fdecomp}. 
The corresponding eigenvalues have the following properties.  
The eigenvalues~$f_{ij}^{ij0}$ and~$f_{ij}^{ij}$ are independent of $\cnf'$.  
The eigenvalues~$f_{ij}^{i0}(\cnf')$ and~$f_{ij}^{j0}(\cnf')$ depend on~$\cnf'$ 
only through the occupation numbers of the edges~$(i0)$ and $(j0)$, respectively.
Finally, the eigenvalues~$f_{ij}^0(\cnf')$, $f_{ij}^i(\cnf')$, and~$f_{ij}^j(\cnf')$ depend on~$\cnf'$ 
only through how this configuration looks locally near the vertices~$0$, $i$, and~$j$, respectively.
To see these,
we simply note that for any two configurations~$\cnf', \cnf''$ 
which look locally the same in a ball-like region~$R$ 
({\em e.g.}, the neighborhood of a vertex or of the interior of an edge), 
the corresponding states~$|\cnf'\rangle$ and~$|\cnf''\rangle$ 
can be transformed into each other by a unitary supported away from~$R$; 
such a unitary will commute with any operator supported on~$R$.

We thus have
\begin{align}
M'_{ij}|\cnf\rangle 
=
F_{ij}^0(\cnf') F_{ij}^i(\cnf') F_{ij}^j(\cnf') F_{ij}^{i0}(\cnf') F_{ij}^{j0}(\cnf') F_{ij}^{ij} F_{ij}^{ij0} M_{ij}|\cnf\rangle
\end{align}
Let us define $u_{ij}^i(\cnf) = F_{ij}^i(\cnf') F_{ij}^{i0}(\cnf')$, 
$u_{ij}^j(\cnf)=F_{ij}^j(\cnf') F_{ij}^{j0}(\cnf')$, 
and $u_{ij}^0(\cnf) = F_{ij}^0(\cnf') F_{ij}^{ij} F_{ij}^{ij0}$.  
Note that these still depend only on how $\cnf$ looks locally near the vertices~$i$, $j$, and~$0$, respectively.  
We finally have:
\begin{align*}
	M'_{ij}|\cnf\rangle = u_{ij}^i(\cnf) u_{ij}^j(\cnf) u_{ij}^0(\cnf) M_{ij}|\cnf\rangle
\end{align*}
which is \cref{eq:hMprime}.

\subsection{$\mu$ is independent of the choice of $M_{ij}$} \label{sec:app_details2}

Now let us prove that the various phases in the expression for the loop self-statistics in terms of the $M'_{ij}$ 
cancel in pairs with their complex conjugates.  
Fix a particular~$(ij)$.
Then, as shown in~\cref{fig:decorated}, 
each~$M_{ij}$ occurs precisely three times in the expression for~$\bM$, as does its inverse.  
Then we have three instances for each of 
$u_{ij}^i, u_{ij}^j, u_{ij}^0, {\overline{u_{ij}^i}}, {\overline{u_{ij}^j}}, {\overline{u_{ij}^0}}$, 
evaluated on different configurations~$\cnf$. 
The fact that these all cancel is illustrated, for the various choices of $(ij)$, 
in~\cref{fig:M12_M13,fig:M14_M23,fig:M24_M34}, 
which are color coded in the same way as~\cref{fig:decorated}, 
and show only the portion of the configuration near the relevant $(ij0)$~plaquette.
Specifically, consider first $u_{ij}^0$ and~${\overline{u_{ij}^0}}$.  
Then the factors of~$u_{ij}^0(\cnf)$ generated on each of the three configurations at the top row 
cancel with the corresponding factors of~${\overline{u_{ij}^0}}$ generated on the configurations directly below.  
This is because, when one acts with $M_{ij}^{-1}$ on each of the configurations in the bottom row, 
the local configuration near~$0$ becomes identical to that of the corresponding configuration directly above; 
recall that, according to~\cref{eq:hMprimeconjugate}, 
the $\cnf$ that appears in the factor of~${\overline{u_{ij}^0}}$ 
generated by the action of~$\left(M'_{ij}\right)^{-1}$ 
is the configuration obtained after~$\left(M'_{ij}\right)^{-1}$ has acted.
Similarly, the factors of $u_{ij}^i$ and ${\overline{u_{ij}^i}}$ 
cancel between the configurations on the top and bottom row 
connected by a line labeled with $i$, 
and likewise for~$u_{ij}^j$ and~${\overline{u_{ij}^j}}$.

\begin{figure}
\centering
\includegraphics[width=6in, trim={0ex 43ex 0ex 0ex},clip]{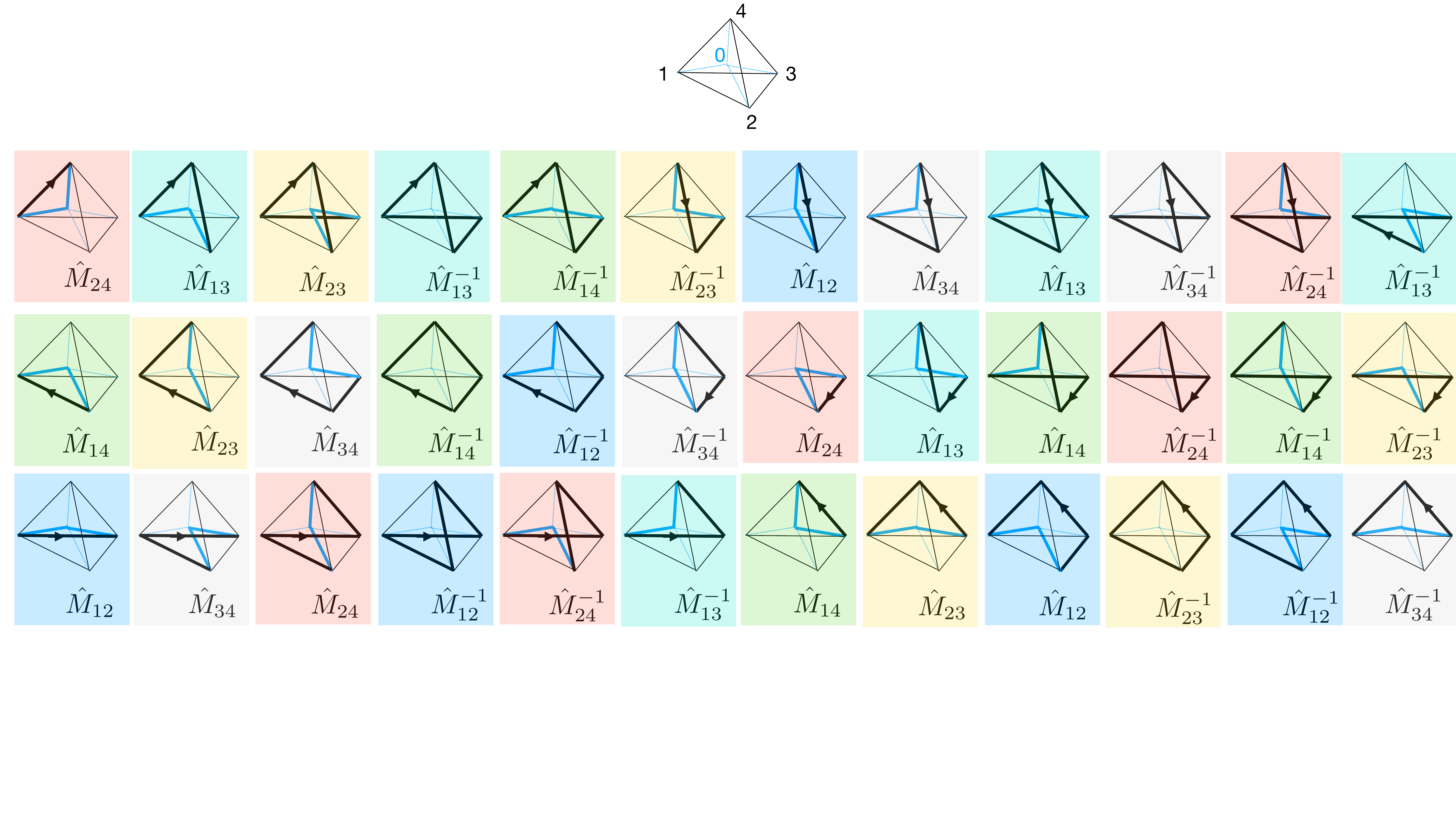}
\caption{There are six different colors, corresponding to the six possible choices of $(ij)$.  Each corresponding $M_{ij}$ occurs precisely three times, as does its inverse.  We start with the configuration in the top left, and act on each configuration state with the operator directly below it.}
\label{fig:decorated}
\end{figure}

\begin{figure}
\centering
\includegraphics[width=0.8\textwidth, trim={20ex 50ex 20ex 50ex},clip]{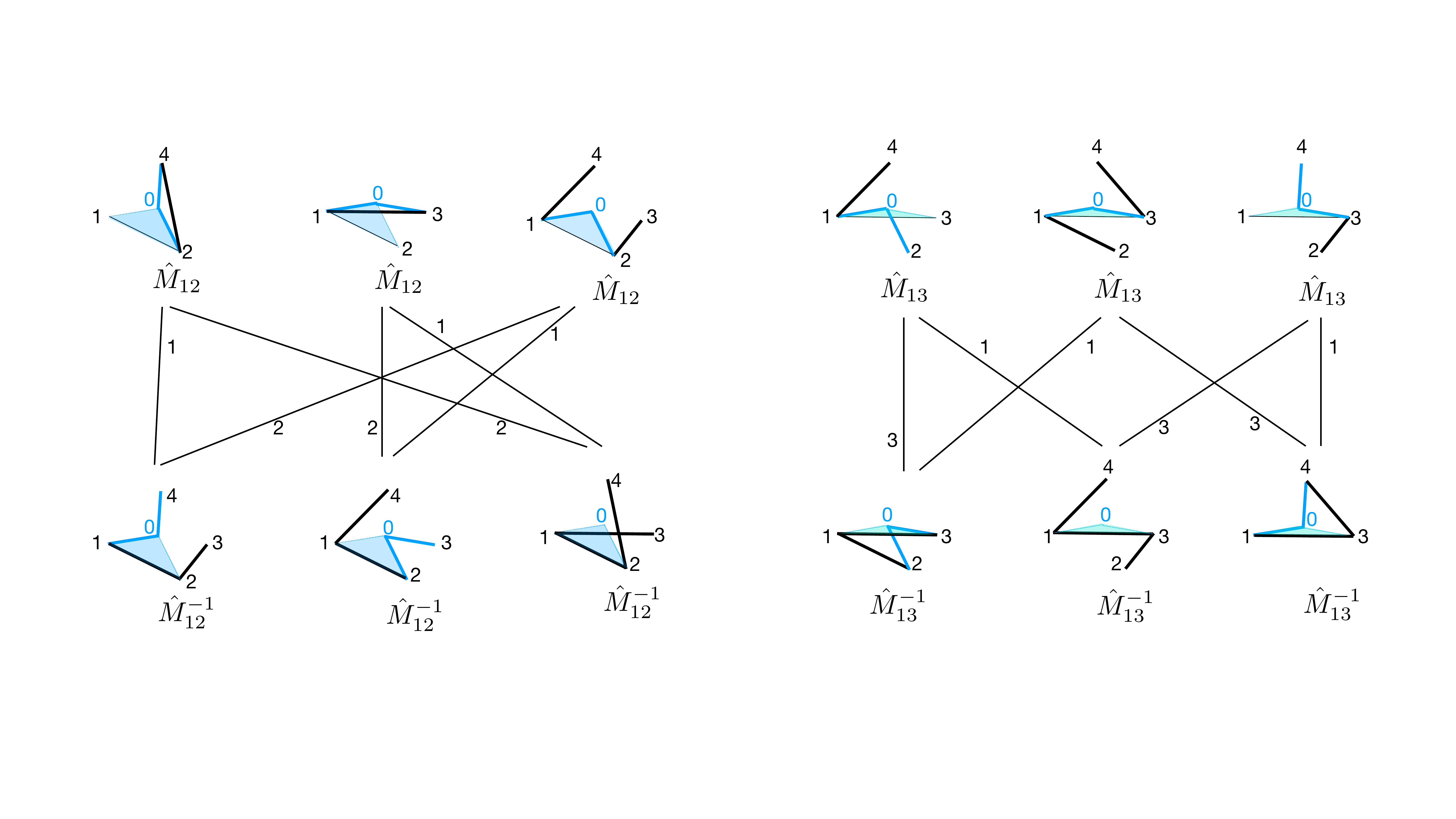}
\caption{Cancellation of phases for $(ij)=(12)$ and $(13)$}
\label{fig:M12_M13}
\end{figure}

\begin{figure}
\centering
\includegraphics[width=6in, trim={0ex 45ex 0ex 50ex},clip]{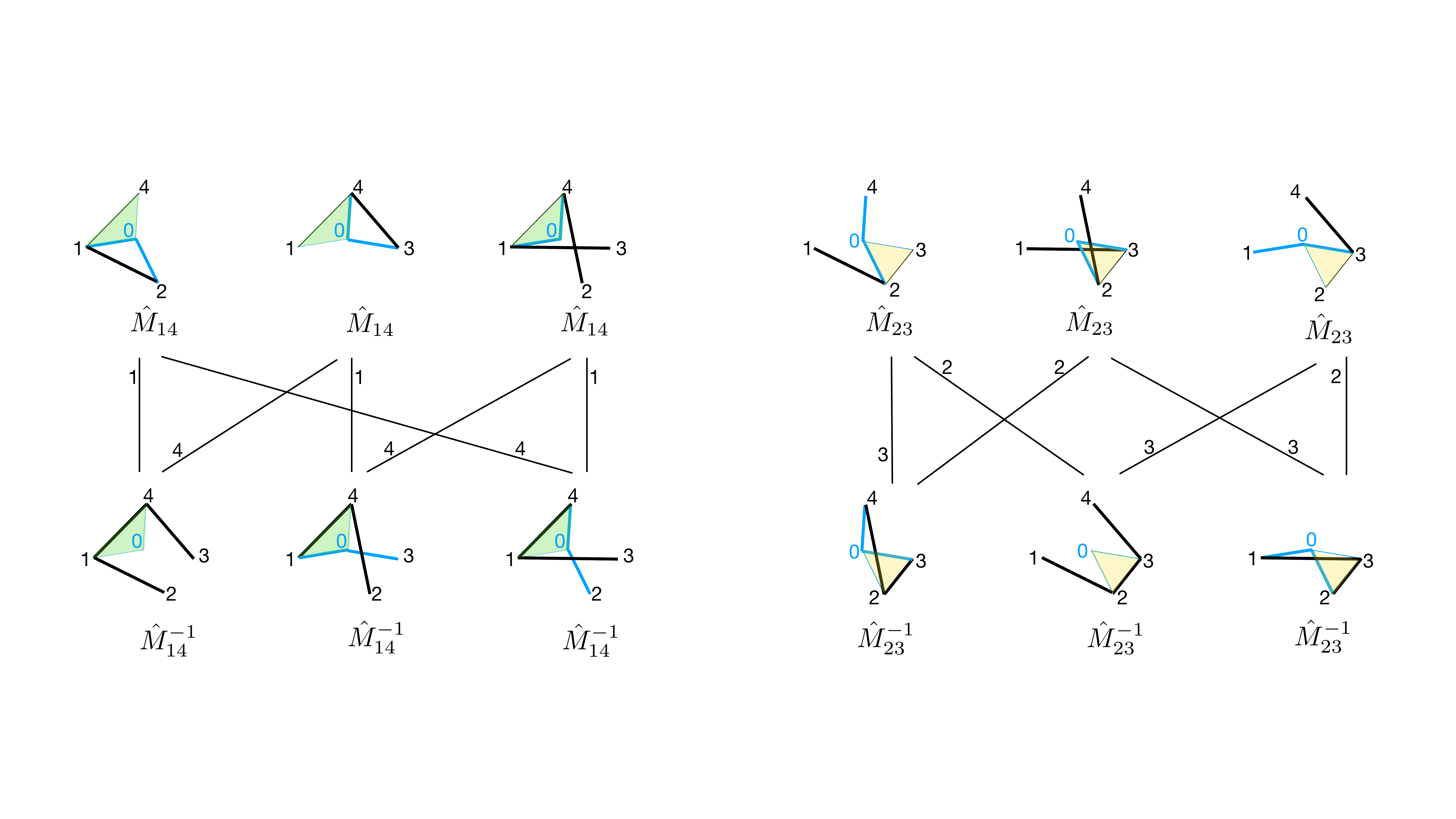}
\caption{Cancellation of phases for $(ij)=(14)$ and $(23)$}
\label{fig:M14_M23}
\end{figure}

\begin{figure}
\centering
\includegraphics[width=6in, trim={0ex 45ex 0ex 50ex},clip]{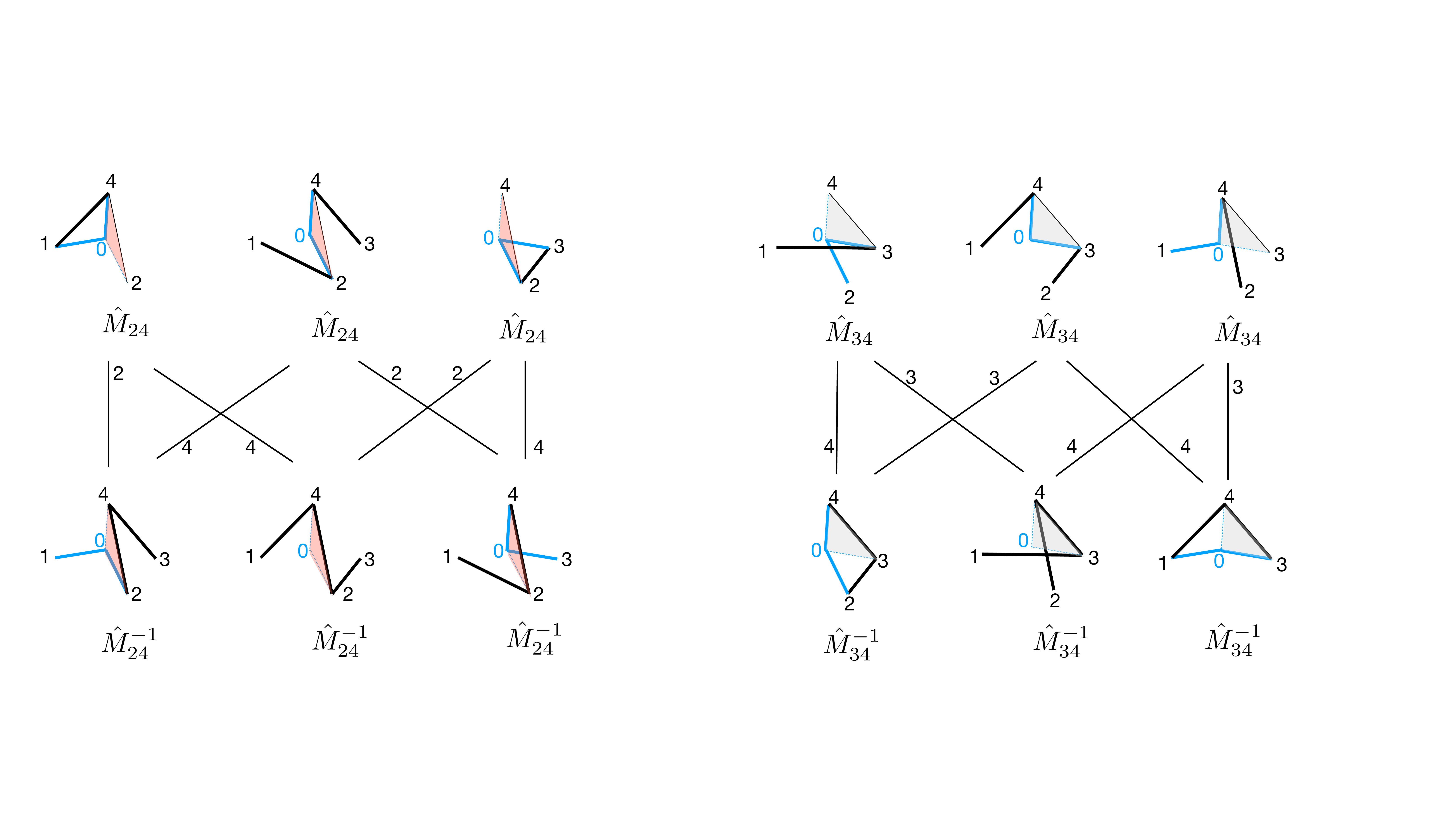}
\caption{Cancellation of phases for $(ij)=(24)$ and $(34)$}
\label{fig:M24_M34}
\end{figure}

\section{Decoration coefficients}\label{app:consistency}

Let $V = \{0,1,2,\ldots\}$ be a label set of $5$ or more elements.
Whenever $a,b,c$ are distinct elements of the label set,
we denote by $[abc]=[cba]$ a value in $\ZZ_2 = \{0,1\}$.
Suppose for any distinct $a,b,c,d \in V$
\begin{align}
	0 &= (ab0) \equiv [ab0] + [b0a] + [0ab] \qquad \text{ if } a,b \neq 0, \label{eq:03c}\\
	0 &= (abcd) \equiv [abc] + [bcd] + [cda] + [dab]. \label{eq:4c}
\end{align}
Regard $[0ab] = [ba0]$ as primitive variables for distinct~$a,b$;
we have not checked whether they are indeed independent variables.
Solve for $[a0b]$ and $[abc]$ in terms of these primitive variables:
\begin{align}
[a0b] &= [0ba] + [0ab],\\
[abc] &= [bc0] + [c0a] + [0ab] = [0cb] + [0ac] + [0ca] + [0ab].\nonumber
\end{align}
Then, it readily follows that for all distinct $a,b,c \in V$
\begin{align}
	(abc) = [abc]+[bca]+[cab] = 0. \label{eq:3c}
\end{align}
The condition~$(abcd) = 0$ puts a constraint on our primitive variables:
\begin{align}
	0 = (abcd) = [a0b] + [b0c] + [c0d] + [d0a] \label{eq:const}
\end{align}
which further implies that for all distinct $x,a,b,c \in V$
\begin{align}
	(abcde) &\equiv [abc] + [bcd] + [cde] + [dea] + [eab],\nonumber\\
	(01234) &= (01234) + (0123) + (0234) + (023) = [401]+[103]+[302]+[204] = 0,\nonumber\\
	(abcde) &= [eab] + [bad] + [dac] + [cae] = 0 \label{eq:tinv}
\end{align}
where the last line is because our manipulation is invariant under label permutations.
\Cref{eq:tinv} implies that for all distinct $x,a,b,c \in V$
\begin{align}
	t(x;abc) \equiv [axb] + [bxc] + [cxa] = t(x;adc). \label{eq:tdef}
\end{align}
Since any two unordered triples can be ``connected'' via a sequence of unordered triples
where any neighboring pair in the sequence share two elements,
this means that \emph{$t(x) = t(x;abc)$ is independent of $abc$.}
Then, \eqref{eq:3c} implies that for all distinct $a,b,c,d \in V$
\begin{align}
	t(a) + t(b) + t(c) + t(d) &= t(a;bcd) + t(b;acd) + t(c;abd) + t(d;abc) = 0.
\end{align}
This means that $t(x) = t(x;012) = t(0) + t(1) + t(2) = t$ is a constant for all $x \neq 0,1,2$.
If we have only $5$ labels, then the left-hand side of a formula~$\sum_{a=0}^4 t(a) = t(0)$
is invariant under label permutations,
so we have $t = t(a)$ for all~$a \in V$.
If we have $6$ or more labels, then $t(a) = t(5) + t(4) + t(3) = t$ for $a = 0,1,2$.
Therefore, \emph{$t(a) = t$ is always a constant for all~$a \in V$ if~$|V| \ge 5$}.

\paragraph{$t=0$.}

We regard~$[b 0 c]$ as a variable on ``edge''~$b c$,
{\em i.e.}, the set $\{ [b 0 c] : b,0,c \text{ are distinct}\}$ 
is a $1$-cochain on the abstract simplex (a complete graph) formed by~$\{1,2,\ldots\} = V \setminus \{0 \}$.
The equation~$0 = t = [b 0 c] + [c 0 d] + [d 0 b]$ is a cocycle condition,
but the degree~$1$ cohomology of a simplex whose dimension is $\ge 3$ is trivial,
so we have~$[b 0 c] = [b 0] + [c 0]$ 
for some choice of~$\{ [b 0] \in \ZZ_2 : b \neq 0 \}$;
here, $[b 0]$ is a new symbol.
Defining $[b c] = [b c 0] + [c 0]$,
we see that $[b c] + [c b] = [b c 0] + [c 0] + [cb0] + [b0] = [bc0] + [cb0] + [b0c] = 0$ by \eqref{eq:3c}.
Then, \eqref{eq:4c} implies that 
\begin{align}
[bcd]&=[cd 0]+[d 0 b]+[0 bc] = [cd0]+[d0b]+([b0c]+[bc0]) \nonumber\\
 &=[cd0] + ([d0b]+[b0c])+[bc0] = [cd0]+[c0d]+[bc0] \\
 &= [cd0]+([d0]+[c0]) + [bc0] = [cd]+[bc].\nonumber
\end{align}
Conversely, if $[ab]=[ba] \in \ZZ_2$ are arbitrarily given 
and we set $[abc] = [ab] + [bc]$, which is equal to $[ba] + [cb] = [cba]$,
then $t=0$ and all \cref{eq:3c,eq:4c} are obviously satisfied.
Therefore, \emph{we conclude that $t=0$ if and only if
there is a choice of $\{[ab]=[ba] \in \ZZ_2~:~ a,b \in V, a \neq b\}$ such that $[abc] = [ab] + [bc]$}.

\paragraph{$t=1$.}

Since the difference of any two $t=1$ solutions is a $t=0$ solution,
it suffices to find one $t=1$ solution.
If $V = \{0,1,2,3,4\}$,
a solution is obtained by setting for any distinct $a,b,c \in \{1,2,3,4\}$,
\begin{align}
	[0ab] &= 0 \text{ except } [012] = [034] = 1.
\end{align}
This satisfies our sole nontrivial condition~\eqref{eq:const};
there are three cases to consider $(abcd) = (1234)$, $(1324)$, and $(1243)$.

For larger~$V$, we use induction.
Suppose we have a solution~$\{[abc]\}$ on~$V\supset \{0,1\}$ 
and we seek to find a solution on~$V \sqcup \{\star\}$.
Define
\begin{align}
[0\star a] &= 0 \text{ for all }a \in V,\quad \text{ and } \quad [0 1 \star] = 0.
\end{align}
Since $[0\star a]$ and $[01\star]$ are new variables
and
no linear relation involving two primitive variables may be generated by~\eqref{eq:const},
we see that this is an allowed choice.
It remains to find~$[0y\star]$ for~$y \in V\setminus\{0,1\}$.
\Cref{eq:tdef} implies that for any $y \in V \setminus\{0,1\}$
\begin{align}
	[\star 0 y] + [y 0 1] + [1 0 \star] &= t = 1,\nonumber\\
	([0 \star y] + [0 y \star]) + [y01] + ([01\star] + [0\star 1]) &= 1,\nonumber\\
	[0y \star] + [y01] &= 1.
\end{align}
The last line determines~$[0y\star]$ because $[y01]$ is given by the induction hypothesis.
This completes the induction for a $t=1$ solution.
(The same induction is applicable for $t=0$ once we have a solution with $|V| = 5$.)

\section{Local orientations of triangles}\label{app:agreement}

Let $V = \{0,1,2,\ldots\}$ be a set (of vertices).
We call an unordered pair of distinct vertices an edge,
and an unordered triple of distinct vertices a triangle.
We may regard $V$ as a simplicial complex of one simplex~$V$;
we will be interested in the $2$-skeleton of this simplicial complex.
\begin{definition}
An \emph{alignment}~$\chi$ is a symmetric $\ZZ_2$-valued function 
on all pairs of edge-sharing triangles.
An alignment~$\chi$ is \emph{regular} if $\chi$ satisfies the following two conditions.
\begin{itemize}
	\item (edge) for any triple~$abc,abd,abe$ of triangles that share a common edge~$ab$,
		we have $\chi(abc|abd) + \chi(abd|abe) + \chi(abe|abc) = 1$.
	\item (vertex) for any triple~$vab,vbc,vca$ of triangles that share a common vertex~$v$ 
	but without any edge common to all,
		we have $\chi(vab|vbc) + \chi(vbc|vca) + \chi(vca|vab) = 0$.
\end{itemize}
\end{definition}

An orientation for a triangle is by definition a preferred cyclic ordering of its vertices.
Let us define a regular alignment on~$V$ where each triangle is given an orientation.
On each edge, the orientation of a triangle that contains the edge
induces an ordering of vertices of the edge.
For two triangles~$abc,bcd$ that share an edge,
we define $\chi(abc|bcd)$ to be~$1$ if and only if the induced orderings on the shared edge agree,
{\em i.e.}, the value of the alignment is nonzero if and only if the local orientation
gives an orientation domain wall element on the shared edge.
Let us show that this alignment is regular.
For any triple of triangles that share a common edge,
we have three induced orderings on the common edge.
Either all three orderings agree or only two agree while the third does not.
In any case, the sum of alignment values must be odd.
This is the edge condition.
For any triple of triangles that share a common vertex but without any edge common to the three triangles,
either exactly two alignment values are~$1$ or none of them is~$1$.
In any case, the sum of $\chi$ values around the vertex must be even.
This is the vertex condition.
Therefore, $\chi$ is regular.
Let us say in this case that $\chi$ is \emph{compatible} with the local orientation of triangles.
\begin{lemma}\label{lem:alignment}
	For a regular alignment~$\chi$ on~$V$,
	there exists a local orientation for all triangles compatible with~$\chi$.
\end{lemma}
\begin{proof}
	If a triangle $abc$ is oriented, and $bcd$ is an edge-sharing triangle,
	then the orientation of~$bcd$ is determined by~$\chi(abc|bcd)$ in order for $\chi$ to be compatible:
	the induced orderings on $bc$ must be opposite if $\chi(abc|bcd) = 0$; otherwise they must be the same.
	We write $\xymatrix@=2ex{abc \ar@{-}[r]& bcd}$ to denote that 
	the two triangles $abc,bcd$ are given orientations that are compatible with~$\chi(abc|bcd)$.
	Then, the edge and vertex conditions translate to the following diagrammatic inference rules.
	\begin{align}
		\xymatrix@!0{
			& abd\ar@{-}[dl]\ar@{-}[dr] & \\
			abc& &abe
		} 
		&\quad\Longrightarrow\quad
		\xymatrix@!0{
			& abd\ar@{-}[dl]\ar@{-}[dr] & \\
			abc\ar@{-}[rr]^\prime& &abe
		}&\text{(edge condition)}\label{eq:edgecond}\\
		\xymatrix@!0{
			& vbc\ar@{-}[dl]\ar@{-}[dr] & \\
			vab& &vca
		} 
		&\quad\Longrightarrow\quad
		\xymatrix@!0{
			& vbc\ar@{-}[dl]\ar@{-}[dr] & \\
			vab\ar@{-}[rr]^\prime& &vca
		}&\text{(vertex condition)}\label{eq:vertexcond}
	\end{align}
	Here, distinct letters stand for distinct vertices,
	and primes denote inferred compatibility.
	
	To define a local orientation, we order $V = \{0,1,2,\ldots\}$ by the integer ordering.
	Orient the first triangle~$012$ once and for all.
	Then, we orient all triangles by
	\begin{align}
		\xymatrix{
			012 \ar@{-}[r] & 01c \ar@{-}[r]& 0bc \ar@{-}[r] & abc
		}\label{eq:orient}
	\end{align}
	for any $a,b,c$ such that $1 \le a < b < c$.
	
	It remains to show that the orientation defined by~\eqref{eq:orient}
	is compatible with~$\chi$ for any pair of edge-sharing triangles.
	To this end we start with an example chain of inference:
	\begin{align}
		&\xymatrix{
			& 012\ar@{-}[dl]\ar@{-}[rd] & \\
			01c\ar@{-}[dd]\ar@{-}[rr]|{(i)}\ar@{-}[dr]|{(iii)} & & 01z\ar@{-}[dl]\ar@{-}[dd]\\
			& 0cz\ar@{-}[dl]|{(iv)}\ar@{-}[dr]|{(ii)} &\\
			0bc\ar@{-}[rr]|{(v)} & &0bz
		}\qquad
		\xymatrix{
			&&\\
			0bc\ar@{-}[dd]\ar@{-}[rr]|{(ii)} \ar@{-}[dr]|{(vi)} & & 0yc\ar@{-}[dl]\ar@{-}[dd]\\
			& byc\ar@{-}[dl]|{(viii)}\ar@{-}[dr]|{(vii)} &\\
			abc\ar@{-}[rr]|{(ix)} & & ayc
		}
		\qquad
		\xymatrix{
		0ca\ar@{-}[r]|{(iv)}\ar@{-}[d] & 0bc & 0cz \ar@{-}[l]|{(iv)}\ar@{-}[d]\\
		bca \ar@{-}[ru]|{(x)} \ar@{-}[rr]|{(xi)} && bcz\ar@{-}[lu]|{(x)}\\
		 & 0bc \ar@{-}[dr]|{(x)}\ar@{-}[r] & xbc\ar@{-}[d]|{(xiii)}\\
		bac\ar@{-}[ru]|{(vi)}\ar@{-}[rr]|{(xii)} & & bcz
		} \label{eq:startingdiagram}
	\end{align}
	Here, any unmarked line exists due to~\eqref{eq:orient}.
	The letters denote arbitrary vertices that are distinct subject to the ordering implied by the notation.
	For example, in the lower right diagram it may be that~$b=1$.
	The lines marked by roman numeral is inferred in order 
	by either the edge or vertex condition.
	As each line is a mathematical statement,
	equivalent statements are labeled the same.
	
	The diagrams in~\eqref{eq:startingdiagram} actually prove the compatibility for all cases.
	To see this, we enumerate all edge-sharing pairs of triangles in a systematic fashion below,
	where a class of pairs is given a wiggly line.
	A wiggly line with a circle means that it is part of the definition~\eqref{eq:orient}.
	A wiggly line with a roman numeral means that it follows from
	the corresponding labeled edge in~\eqref{eq:startingdiagram}.
	
	For triangles sharing an edge with triangle~$012$, we have 
	\begin{align}
		\xymatrix{
		& 012 & \\ 
		01x\ar@{~}[ru]|\circ  & 02x\ar@{~}[u]|{(iii)} & 12x\ar@{~}[lu]|{(x)}
		}.
	\end{align}
	For triangles sharing an edge with triangle~$01c$ where $c > 2$, we have
	\begin{align}
		\xymatrix{
		&& 01c & & \\
		1cz \ar@{~}[rru]|{(x)} & 0cz \ar@{~}[ru]|{(iii)} & 0yc\ar@{~}[u]|\circ & 01w\ar@{~}[lu]|{(i)} & 1vc \ar@{~}[llu]|{(vi)}
		}.
	\end{align}
	For triangles sharing an edge with triangle~$0bc$ where $b,c > 1$, we have
	\begin{align}
		\xymatrix{
		0uc\ar@{~}[rr]|{(ii)}&& 0bc && 0cv \ar@{~}[ll]|{(iv)}\\
		0bz\ar@{~}[rru]|{(v)} & 0yb\ar@{~}[ru]|{(iv)} & xbc \ar@{~}[u]|{\circ} & bwc \ar@{~}[lu]|{(vi)} & bcv \ar@{~}[llu]|{(x)}
		}.
	\end{align}
	For triangles sharing an edge with triangle~$abc$ where $a,b,c > 0$, we have
	\begin{align}
		\xymatrix{
			& abc & \\
			bcz \ar@{~}[ru]|{(xiii)} & byc \ar@{~}[u]|{(viii)} & xbc \ar@{~}[lu]|{(vii)}
		}\qquad
		\xymatrix{
			bac & \\
			byc \ar@{~}[u]|{(ix)} & bcz \ar@{~}[lu]|{(xii)}
		}\qquad
		\xymatrix{
			bca \\
			bcz \ar@{~}[u]|{(xi)}
		}
	\end{align}
	where we do not have to consider cases where~$c < b$ by symmetry~$b \leftrightarrow c$.	
\end{proof}

\section{Proof that the loop self-statistics $\mu=\pm 1$}\label{app:mu}

\begin{figure}[b]
\centering
\includegraphics[width=0.5\textwidth]{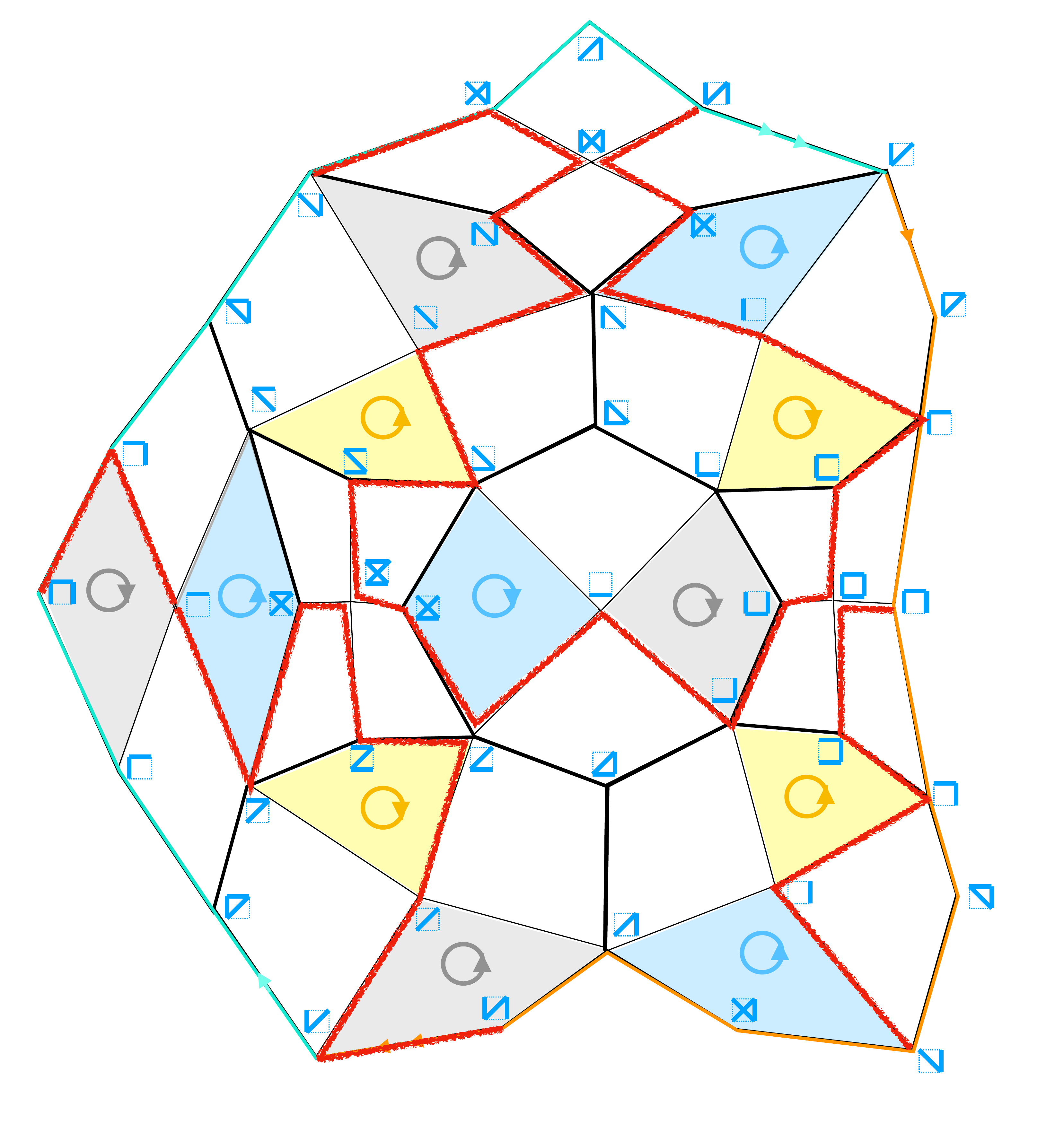}
\caption{
The complex defined in the text is explicitly seen to be the $2$ dimensional projective plane.  
Note that the outer orange edges are identified with the outer blue ones as indicated by the arrows.  
The figures at the vertices represent the corresponding gauge flux loop configuration, 
with the $0$ vertex, and the edges connecting to it, suppressed.  
The red path corresponds to the sequence of membrane operators and their inverses used to define $\mu$, 
see \cref{eq:defhI}.
As we can explicitly see from the figure, 
this path is nontrivial in the fundamental group of~$\mathbb{RP}^2$.  
The three sets of four quadrilaterals associated to the three group commutators 
of the form $M^{-1}_{jk} M^{-1}_{il} M_{jk} M_{il}$ are colored in yellow, grey, and blue.
}
\label{fig:complex_figure}
\end{figure}


\begin{figure}[tbh]
\centering
\includegraphics[width=0.5\textwidth, trim={80mm 210mm 230mm 30mm}, clip]{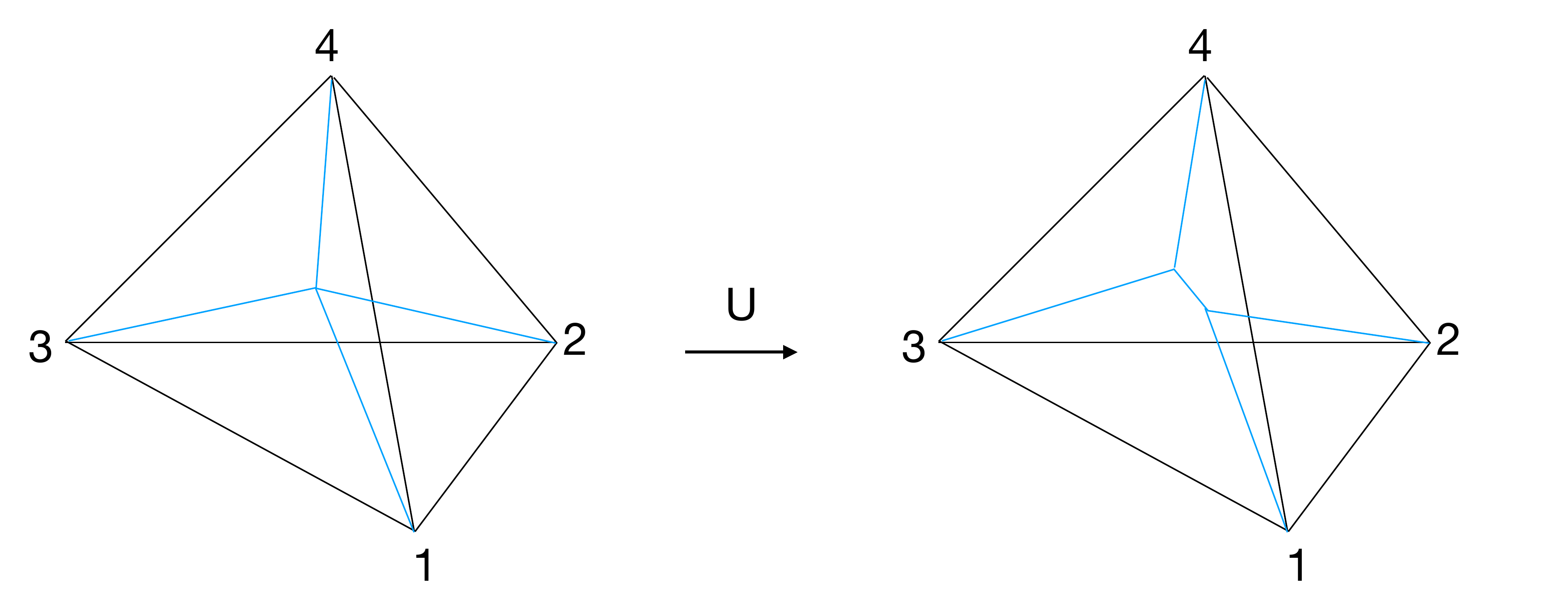}
\caption{There exists a shallow quantum circuit~$U$ that transforms any gauge flux configuration on the skeleton on the left to the corresponding configuration on the right.  The occupation of the extra middle link on the right is always uniquely determined.  Note that the Lieb--Robinson of $U$ may be longer than this extra middle link, which is assumed to be much shorter than the other links.  Now, for the configurations on the right, it is clear that one may choose membrane operators such that those associated to $(012)$ and $(034)$ commute.  Conjugating these by~$U$, one obtains membrane operators for the configurations on the left, again with the property that $M_{12}$ and $M_{34}$ commute. }
\label{fig:quantization_fig2}
\end{figure}

To prove that $\mu = \pm 1$,
it is useful to first construct a certain $2$-dimensional complex~$\mathcal L$, as follows.
The vertices of this complex are configurations of a gauge flux loop 
on a tetrahedron~$\mathfrak T$ with vertices~$1,2,3,4$ and a central $0$~vertex.
Each vertex of~$\mathcal L$ is denoted by a sequence of $\mathfrak T$-vertices,~{\it e.g.,}~$(1203)$,
up to cyclic and reverse reordering,~{\it e.g.,}~$(1203) = (2031) = (1302)$.
Specifically, they are the $33$ configurations that appear in the definition of~$\mu$, 
together with four additional configurations~$(123),(124),(134)$, and~$(234)$,
consisting of triangular loops on the surface of~$\mathfrak T$.
Two such configurations are connected by an edge 
if there exists a membrane operator~$M_{ij}$
that maps one to the other. 

The faces of~$\mathcal L$ are all quadrilaterals.
By dividing up each quadrilateral into two triangles, 
we will have a simplicial complex, 
though we will not need to do this.
The quadrilaterals are defined as follows.
First, for each edge~$(ij)$ of $\mathfrak T$ with $i,j \in \{1,2,3,4\}$
we define four quadrilaterals of~$\mathcal L$ as follows.
Given $i,j$, let $k,l \in \{1,2,3,4\}$ be the two vertices of $\mathfrak T$, different from $i,j$.
Then, the four quadrilaterals, 
defined by specifying the $\mathcal L$-vertices around their perimeters, are:
\begin{align}
&\{(ij0),(ijl0),(kijl0),(kij0)\}, & \{(ij0),(ijk0),(lijk0),(lij0)\}, \\
&\{(ij0),(ijk0),(ijk),(kij0)\}, & \{(ij0),(ijl0),(ijl),(lij0)\}.  \nonumber
\end{align}
Second, for each $(ijkl)$ of~$(1234),(1324)$, and~$(1243)$,
which are the three permutations of~$(1234)$ up to cyclic and reverse reordering,
we define four quadrilaterals 
\begin{align}
&\{(ijkl),(ijkl0),(ijk0),(ijk0l)\}, & \{(ijkl),(ijk0l),(ij0l),(ij0kl)\},\\
&\{(ijkl),(ij0kl),(i0kl),(i0jkl)\}, &\{(ijkl),(i0jkl),(0jkl),(0ijkl)\}.\nonumber
\end{align}
The resulting $2$-complex~$\mathcal L$ is in fact a $2$-manifold~$\mathbb{RP}^2$,
the real $2$-dimensional projective plane,
as illustrated explicitly in \cref{fig:complex_figure}.

Each quadrilateral in~$\mathcal L$ is associated 
with the eigenvalue of a certain group commutator of membrane operators acting on a certain configuration state.
Up to permutations of $\mathfrak T$-vertices,
there are three types of such eigenvalues:
\begin{align}
\begin{cases}
M^{-1}_{jk} M^{-1}_{il} M_{jk} M_{il} \text{ acting on } \ket{(ij0)},\\
M^{-1}_{jk} M^{-1}_{ik} M_{jk} M_{ik} \text{ acting on }\ket{(ij0)},\\
M^{-1}_{jk} M^{-1}_{ik} M_{jk} M_{ik} \text{ acting on } \ket{(ilj0)},  
\end{cases}
\end{align}
where $i,j,k,l \in \{1,2,3,4\}$ are distinct.
Here, the second and third lines have the same operators evaluated on different states.
We claim that, after a possible re-definition of membrane operators~$M_{**}$, 
these eigenvalues can all be set equal to~$+1$.

To prove this, we first consider the first type of group commutator, 
$M^{-1}_{jk} M^{-1}_{il} M_{jk} M_{il}$.  
Note that there are exactly three such group commutators.  
Each of these three acts on four possible states, 
yielding three sets of four quadrilaterals in \cref{fig:complex_figure}, 
each set of four colored in a different color.  
Also, the three group commutators of this form involve disjoint sets of membrane operators, 
so we can look at them in turn.  
Let us thus consider $M^{-1}_{34} M^{-1}_{12} M_{34} M_{12}$ 
(colored yellow in~\cref{fig:complex_figure}).
This group commutator acts on the four states $(130),(240),(140),(230)$, 
which are eigenstates of it with eigenvalues $\alpha,\beta,\gamma,\delta \in U(1) \subset \CC$.
Also, as an operator, it is localized near the center $0$ vertex.  To see that $\alpha,\beta,\gamma,\delta$ can all be set equal to one,
we simply note that, as illustrated in~\cref{fig:quantization_fig2}, there exists a choice of membrane operators such that $M_{12}$ and $M_{34}$ commute identically.  The same argument can be repeated for the other two group commutators of the form $M^{-1}_{jk} M^{-1}_{il} M_{jk} M_{il}$.


Now let us consider the second type of group commutator, $M^{-1}_{jk} M^{-1}_{ik} M_{jk} M_{ik}$.
Note that as a shallow quantum circuit, it is supported in the vicinity of the $(k0)$ edge, 
so its eigenvalue is the same on $\ket{(ij0)}$ and $\ket{(ilj0)}$.
Again, it is easy to see that one can modify the membrane operators near the $k$ vertex 
to set the eigenvalues of all group commutators of this form to~$1$.  Explicitly,
note that for each $k$, there are three commutators of this form 
--- this is just the number of ways of choosing two out of the three edges 
that connect the vertex~$k$ of~$\mathfrak T$ to noncentral vertices of~$\mathfrak T$.
For each such commutator
$M^{-1}_{jk} M^{-1}_{ik} M_{jk} M_{ik}$, 
we now modify the membrane operator~$M_{ik}$ by multiplying it by a phase 
that depends on the occupation number of the $(jk)$ edge, 
in such a way as to make the commutator equal to~$+1$.
Note that this modification does not affect the eigenvalues of the other two commutators 
associated to $k$, so by this argument they can all be set equal to~$+1$.  Furthermore, since all these modifications take place near the $1,2,3,4$ vertices, they are far away from the central~$0$ vertex, and hence do not change the value of $M^{-1}_{jk} M^{-1}_{il} M_{jk} M_{il}$ (which has already been set equal to $1$).

Thus, we can assume that, with an appropriate choice of membrane operators, 
all of the group commutator actions corresponding to the quadrilaterals in \cref{fig:complex_figure} 
have eigenvalues~$+1$.
Now consider computing~$\mu^2$, 
which corresponds to tracing out the red path in \cref{fig:complex_figure} twice.
This doubled path is contractible --- 
it can be deformed to a trivial path by sliding it over the various quadrilaterals
(in fact, over each quadrilateral exactly once).
Since the eigenvalue corresponding to each quadrilateral is~$+1$, 
we see that such path deformations do not change the eigenvalue associated to the path.
Since the eigenvalue associated to the trivial path is~$+1$, 
we thus see that $\mu^2=1$, so that $\mu=\pm 1$, as desired.

\section{Frame parity of fermion worldlines}\label{app:frm}

In this appendix, we collect results about~$\Frm$ and its extension~$\overline\Frm$, 
improving upon~\cite[Lem~II.4]{4dbeyond}.
The object~$\overline\Frm$ is a quadratic form on all $1$-cycles over~$\ZZ_2$.
This function~$\Frm$ resembles a spin structure 
in that $\Frm$ can be used to define a Hamiltonian with emergent fermions,
and the set of all possible $\overline\Frm$ is acted on transitively and freely by
the first cohomology group of the underlying space.
In essence, $\overline\Frm$ measures in $\ZZ_2$ how every $1$-cycle is twisted.
However, it is different from spin structures 
since there is no obstruction to define~$\overline\Frm$ given an underlying space ---
our results will be applicable to any $2$- or higher dimensional combinatorial manifold,
regardless of the Stiefel--Whitney classes.

We will begin by introducing a continuous map from $2$-skeleton down to $\RR^2$ 
which encodes how each null-homologous worldline is twisted, 
and define a $\ZZ_2$-valued function~$\Frm$.
We will note that $\Frm$ is essentially unique and independent of the continuous map down to~$\RR^2$,
and construct fermion string operators
by examining localizations of worldline fluctuation operators.
The results on fermion string operators are used to analyze 
a commuting Pauli Hamiltonian that is topologically ordered, 
is fixed under an entanglement renormalization group flow,
and has a unique nontrivial topological charge that is a fermion 
whenever the underlying space has dimension~$3$ or higher.
We will extend $\Frm$ to $\overline\Frm$ on all $1$-cycles
and consider the action on $\overline\Frm$ by first cohomology group of the underlying space.

We do not consider any primary chains or any primary qubits in the following discussion on~$\Frm$.
For simplicity we only consider simplicial complexes;
for a more general cell complex, one can use the prescription in~\cref{sec:fcbl},
especially~\cref{eq:lfc}.

\subsection{Projection and $\Frm$}

\begin{definition}\label{def:phi}
	For a $2$-dimensional simplicial complex~$\mathcal K$,%
	\footnote{a locally finite cell complex in which every cell is a simplex
	and the intersection of any two simplices is a simplex.}
	a \emph{projection} is a continuous map $\phi: \mathcal K \to \RR^2$
	such that all $1$-cells are mapped to transverse straight line segments
	and each $2$-cell, a triangle, is mapped injectively to the triangle in $\RR^2$ defined by its sides.
\end{definition}
\noindent
Unless states otherwise, 
$\mathcal K$ will always denote a simplicial $2$-complex equipped with a projection~$\phi$.

For each $2$-cell $f_2$,
we define $L_{Fc}(f_2)$ to be the product of Pauli~$X$ along the boundary of~$f_2$
and Pauli~$Z$ on every $1$-cell~$e_1$ such that
$\phi(e_1)$ intersects the interior of~$\phi(f_2)$.
This definition is the same as in~\cref{sec:fcbl} 
(except for $Z$ on the primary qubit on the Poincar\'e dual of~$f_2$).

By construction, the frame parity~$\Frm(a_1)$ of a null-homologous $1$-cycle~$a_1$ is computed by
any $2$-chain~$\tilde a_2$ bounded by~$a_1$:
\begin{align}
	L(\tilde a_2) &= \prod_{f_2 \in \tilde a_2}  L_{Fc}(f_2) = (-1)^{\Frm(a_1)} \prod X \prod Z, \label{eq:Lb2} \\
	(-1)^{\Frm(a_1)} &=  \bra{a_1} L(\tilde a_2) \ket 0.
\end{align}
where $\ket 0$ is the $+1$ eigenstate of all the relevant $Z$ operators.
\begin{proposition}\label{lem:Frm}
	The map $\Frm$ from any finite null-homologous cycle~$a_1$ to a sign,
	is a well-defined function that can be evaluated 
	by any $2$-chain~$\tilde a_2$ with~$\bd \tilde a_2 = a_1$.
\end{proposition}
\begin{proof}
Let $a_1 = \bd \tilde a_2$ be any null-homologous $1$-cycle of $\mathcal K$,
and let $\tilde a_2'$ be another $2$-chain such that $\bd \tilde a_2' = a_1$.
We have to show that
\begin{align}
	L(\tilde a_2) L(\tilde a_2') = L(\tilde a_2 + \tilde a_2') = + \prod Z,
\end{align}
which cannot have any $X$ factor because $\bd(\tilde a_2 + \tilde a_2') = 0$.
Consider a cone over the support~$N$ of $\tilde a_2 + \tilde a_2'$.%
\footnote{This argument is the same as in~\cite[\S II.C.b]{4dbeyond}.}
The subcomplex~$N$ is $\ZZ_2$-closed, so if $c$ is the chain of all $3$-cells of the cone,
the boundary~$\bd c$ is $\tilde a_2 + \tilde a_2'$.
Let us extend $\phi$ so that the extension~$\phi^e$ is defined on the $2$-skeleton of the cone.
This is straightforward:
place an additional generic point~$p$ on~$\RR^2$ for the apex of the cone,
and connect the images of the $0$-cells of~$N$ to $p$ by straight lines.
Every $2$-cell of the cone consisting of some $1$-cell of~$N$ and the apex,
has an obvious image under $\phi^e$.
Define $L(\bd e_3) = \prod_{f_2 \in \bd e_3} L_{Fc}(f_2)$ for any $3$-cell~$e_3$ of the cone.
Since $\bd c = \tilde a_2 + \tilde a_2'$,
we have 
\begin{align}
	\prod_{e_3 \in c} L(\bd e_3) = L(\tilde a_2 + \tilde a_2').
\end{align}
Each $e_3$ is a tetrahedron, 
and there are only two different projections of a tetrahedron, up to homotopy,
depending on whether a vertex is inside a triangle.
By direct calculation, we confirm that $\bra 0 L(\bd e_3)\ket 0 = +1$,
and therefore the claim is proved.
\end{proof}

\begin{remark}\label{rem:extphi}
If we extend $(\mathcal K,\phi)$ by adding more cells to obtain $(\mathcal K^e,\phi^e)$ where $\phi^e|_{\mathcal K} = \phi$,
the induced function $\Frm^e$ on null-homologous $1$-cycles of $\mathcal K^e$
agrees with $\Frm$ on those of $\mathcal K$.
This is simply because $\Frm^e(a_1)$ for a null-homologous $1$-cycle $a_1$ of $\mathcal K$ 
can be evaluated by some $2$-chain of $\mathcal K$,
over which $L_{Fc}$ are the same for $\mathcal K$ and $\mathcal K^e$.
Once we have an extension~$(\mathcal K^e,\phi^e)$,
the function $\Frm^e(a_1)$ where $a_1$ is a cycle of $\mathcal K$ may be evaluated by a $2$-chain of $\mathcal K^e$.
$\lozenge$\end{remark}

\begin{proposition}\label{lem:diffphi}
	Let $\phi'$~be another projection $\mathcal K \to \RR^2$,
	which defines a different set of operators~$L_{Fc}'$,
	and let $\Frm'$ be the resulting function by~\cref{eq:Lb2}.
	Assume that $\mathcal K$ is the $2$-skeleton of an $n$-dimensional simplicial manifold with~$n \ge 2$
	such that every $2$-cell of~$\mathcal K$ has at most~$m$ $1$-cells in its boundary
	and every $0$-cell has at most~$m$ $1$-cells in its coboundary.
	Then, $\Frm + \Frm'$ on the space of all null-homologous cycles is locally computable,
	{\em i.e.}, there exists a quantum circuit of depth at most a constant depending only on~$m$
	which implements $\ket a \mapsto (-1)^{\Frm(a) + \Frm'(a)}\ket a$ for all null-homologous cycles~$a$.
\end{proposition}

Thus, we may speak of \emph{the} frame parity~$\Frm$ 
without specifying a particular projection
for any ``uniform'' cellulation of a combinatorial manifold,
if we are interested in shallow-circuit-equivalence of many-body states.
Note that $\Frm$ is \emph{not} defined on $1$-cycles of nonzero homology class.
We will consider extensions of~$\Frm$ to all $1$-cycles in~\cref{prop:globalFrm}.

\begin{proof}
	Since $\mathcal K = \mathcal K_2$ is a simplicial $2$-complex,
	the projection is, by definition, determined by the image of the $0$-skeleton~$\mathcal K_0$.
	Suppose that $\phi'(\mathcal K_0)$ and~$\phi(\mathcal K_0)$ 
	differ only by the image of one point~$v_0 \in \mathcal K_0$.
	Since $\mathcal K$ is the $2$-skeleton of a manifold,
	any chain~$e$ that passes through~$v_0$ is homologous to one~$u(e)$ that does not.
	For any chain~$e$, we may choose $u(e)$ such that
	the change~$u(e) + e$ is a null-homologous cycle supported on 
	the union~$E(v_0)$ of the boundaries of $2$-cells that touch~$v_0$.
	Therefore, for any $1$-cycle~$a$, the change~$\Frm(a) + \Frm(u(a))$
	can be computed by a diagonal unitary~$F(v_0)$ supported on~$E(v_0)$,
 	and so can be $\Frm'(a)+\Frm'(u(a))$ by~$F'(v_0)$.
 	(The unitary~$F(v_0): \ket a \mapsto (-1)^{\Frm(a) + \Frm(u(a))} \ket a$ 
 	does not change the underlying cycle~$a$.)
 	But, we know $\Frm'(u(a)) = \Frm(u(a))$ due to~\cref{rem:extphi}.
 	Hence, $\Frm(a) + \Frm'(a)$ is computed 
 	by a local diagonal unitary~$G(v_0) = F(v_0)F'(v_0)$ supported on~$E(v_0)$.

	Now, consider general cases 
	where $\phi(\mathcal K_0)$ and $\phi'(\mathcal K_0)$ are arbitrarily different.
	Let $\phi^{(0)} = \phi$ and inductively define $\phi^{(t)}$ for $t \ge 1$ as follows.
	Choose a point~$v \in \mathcal K_0$ such that $\phi^{(t-1)}(v) \neq \phi'(v)$;
	if there is no such a point, then $\phi^{(t-1)} = \phi'$.
	Choose any maximal subset~$V(t)$ of~$\mathcal K_0$ such that 
	for any~$v,v' \in V(t)$ we have~$E(v) \cap E(v') = \emptyset$ whenever~$v \neq v'$,
	and define~$\phi^{(t)}(p) = \phi'(p)$ if~$p \in V(t)$ and~$\phi^{(t)}(p) = \phi^{(t-1)}(p)$ if~$p \notin V(t)$.
	That is, as $t$ increases, we gradually ``move'' projected points under~$\phi$ to match those under~$\phi'$.
	At any given~$t$, the gates~$G(v)$ for $v \in V(t)$ do not overlap, 
	so we may apply them in parallel.
	The assumption on neighbor numbers implies that for some~$t_\text{max}$ upper bounded only through~$m$,
	we have $\phi^{(t_\text{max})} = \phi'$,
	and hence the cumulative change $\Frm + \Frm'$ is computed 
	by a quantum circuit of depth~$t_\text{max}$ that is bounded only through~$m$.
\end{proof}

\begin{remark}
The function~$\Frm$ is not just a function but a $\ZZ_2$-valued quadratic form,
{\em i.e.}, $\Frm(a_1+a_1') - \Frm(a_1) - \Frm(a_1')$ is bilinear.
To see this, for any pair of secondary $2$- and $1$-chain $(b_2, a_1)$
we define $\BiFrm((b_2, a_1))$ to be the mod~$2$ number of times the \emph{thorns} of $L(b_2)$,
the $\prod Z$ part of~\cref{eq:Lb2}, intersect $a_1$.
Since $L_{Fc}$ are Pauli operators of which multiplication ignoring signs
may be regarded as $\ZZ_2$-vector addition,
we see that $\BiFrm$ is bilinear.
Since $\BiFrm((\tilde a_2, \bd \tilde a_2'))$ is
the sign upon rearranging Pauli factors in the product~$L(\tilde a_2)L(\tilde a_2')$,
we conclude that
\begin{align}
	\Frm(a_1 + a_1') = \Frm(a_1) + \Frm(a_1') + \BiFrm(( \tilde a_2, a_1' )) \label{eq:Frm}
\end{align}
where $\tilde a_2$ is any secondary $2$-chain such that $\bd \tilde a_2 = a_1$.
Since $\Frm$ is a function of null-homologous $1$-cycles,
the association 
\begin{align} \label{eq:association}
	(a_1,a_1') \mapsto \BiFrm(a_1,a_1') = \BiFrm((\tilde a_2, a_1')) = \BiFrm((\tilde a_2', a_1))
\end{align}
is a well-defined bilinear, symmetric function on the space of all null-homologous $1$-cycles.
We will see again, after~\cref{cor:localizing} below,
that $\BiFrm$ is independent of $2$-chains bounded by one of the argument $1$-chains.
$\lozenge$\end{remark}

\begin{figure}
\centering
\includegraphics[width=\textwidth, trim={0ex 20ex 55ex 0ex}, clip]{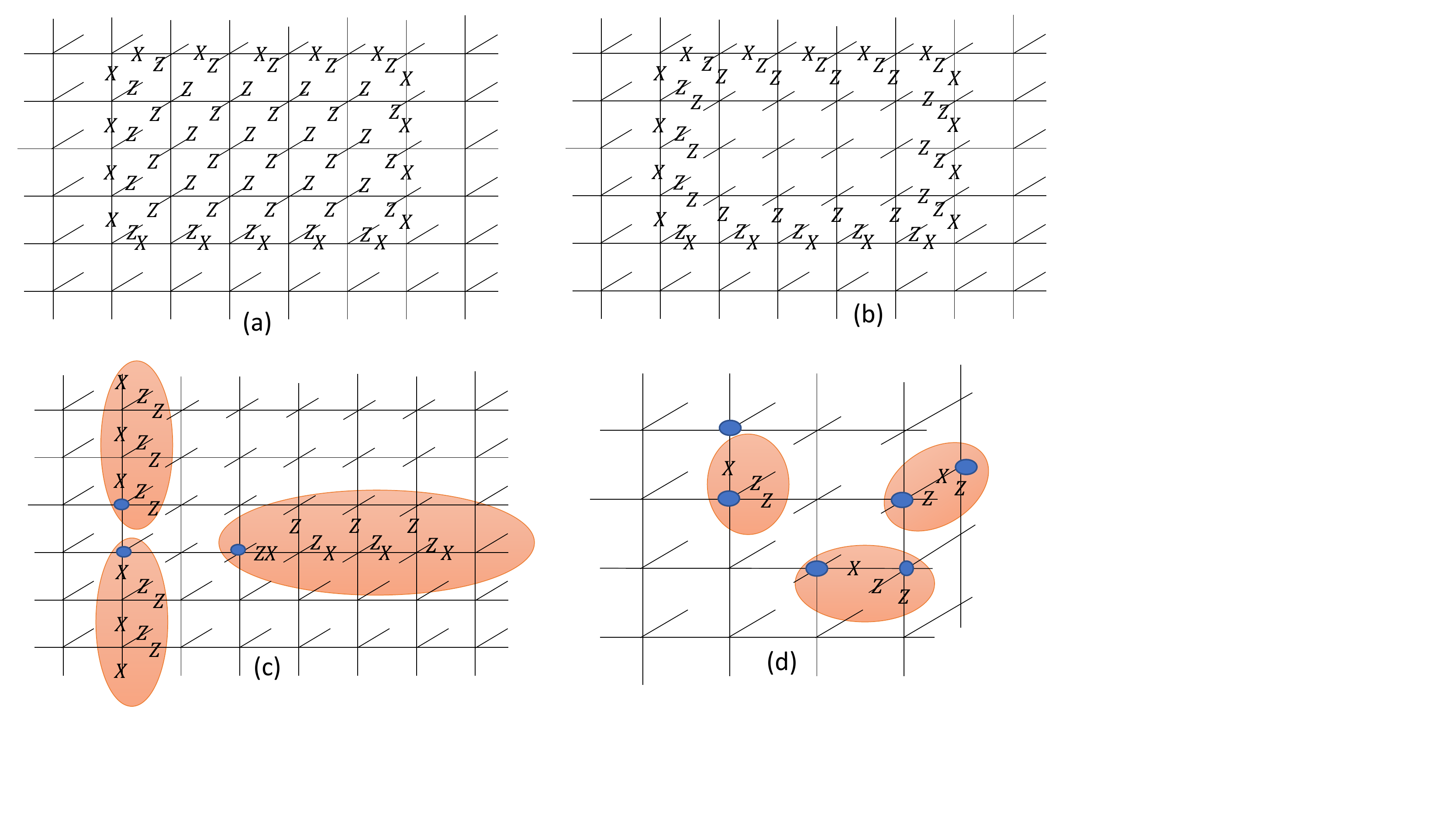}
\caption{
	The secondary cellulation of the system boundary is shown.
	The drawing specifies a projection~$\phi$.
	The fourth direction into the $4$-dimensional bulk is not shown.
	($\mathsf a$) A product of worldline fluctuation operators over a $2$-chain $b_2$.
	The $X$ factors are along~$\bd b_2$, and $Z$ factors are ``thorns.''
	($\mathsf b$) After multiplying by the cycle-enforcing 
	operators in the interior where the winding parity of $b_2$ is odd,
	the operator over~$b_2$ localizes along~$\bd b_2$.
	($\mathsf c$) If we break up the closed string operator,
	there will be a quasiparticle at the end of the string.
	The three semi-infinite operators have end configurations such that 
	when any pair of them are joined by overlaying the solid dot,
	the resulting operator commutes with every Hamiltonian term.
	Among the three, one pair is anticommuting, 
	showing that the quasiparticle is a fermion.
	With primary qubits and worldsheets that link with worldlines,
	the fermion is confined,
	but in a system with worldlines only, 
	the fermion is a deconfined topological excitation.
	($\mathsf d$) Any fermion string operator is built out of three types of segments.
}
\label{fig:stringop}
\end{figure}

\subsection{Fermion string operators}

Let us show that the product~$L(b_2) = \prod_{f_2 \in b_2} L_{Fc}(b_2)$ of $L_{Fc}$ 
over a $2$-chain~$b_2$ localizes around~$\bd b_2$ 
by multiplying~$L(b_2)$ by certain cycle-enforcing operators~$\pi(v_0) = \prod_{e_1 \in \db v_0} Z(e_1)$.
For any chain~$x$, let $\Supp(x)$ be (the topological closure of) the union of all cells in~$x$.
\begin{lemma}\label{lem:cancelthorns}
	For any $2$-chain~$b_2$, there exists a subset~$V_0$ of vertices in~$\Supp(b_2) \setminus \Supp(\bd b_2)$
	such that the operator $S = L(b_2)\prod_{v_0 \in V_0} \pi(v_0)$ is supported on
	the collection of $1$-cells that meet $\Supp(\bd b_2)$.
\end{lemma}
\begin{proof}
	We first define \emph{winding parity} of~$b_2$ at an \emph{interior} 
	vertex~$v_0 \in \Supp(b_2)\setminus\Supp(\bd b_2)$.
	Let~$N(v_0)$ be a chain of all $2$-cells of~$b_2$ which meet~$v_0$.
	The projected boundary~$\phi(\bd N(v_0))$ is a $1$-cycle in the plane~$\RR^2$, one which does not pass through~$\phi(v_0)$.
	Let $\gamma$ be an integral piecewise linear $1$-cycle of~$\RR^2$
	that reduces to~$\phi(\bd N(v_0))$ modulo~$2$. We may regard~$\gamma$ as a collection of (not necessarily disjoint) oriented immersed circles on~$\RR^2 = \CC$.
	The winding parity is defined to be the complex contour integral
	\begin{align}	
		\left(\frac{1}{2\pi i}\oint_{\gamma} \frac{\mathrm d z}{z - \phi(v_0)}\right) \bmod 2.
	\end{align}
	The lifting~$\gamma$ is arbitrary but any lifting differs from~$\gamma$
	by a cycle that is zero modulo~$2$.
	So, the difference in the contour integral is two times some contour integral over a cycle, 
	which vanishes modulo~$2$,
	and therefore the winding parity is well defined.
	
	A thorn of~$L(b_2)$ may be on a $1$-cell in~$\Supp(b_2)$ (internal) or not (external).
	We claim that external thorns of~$L(b_2)$ appear around an interior 
	$0$-cell~$v_0 \in \Supp(b_2)\setminus\Supp(\bd b_2)$
	if and only if the winding parity of~$b_2$ at~$v_0$ is odd.
	Indeed, an external thorn on~$e_1$ with~$\bd e_1 \ni v_0$ exists
	if and only if $\phi(e_1)$ intersects an odd number of~$\phi(f_2)$ for~$f_2 \in N(v_0)$.
	Consider a straight line segment~$\ell(x)$
	between~$\phi(v_0)$ and any point~$x$ on $\phi(\bd N(v_0))$.
	As $x$ moves along $\phi(\bd N(v_0))$, the line segment~$\ell(x)$ may overlap with~$\phi(e_1)$,
	in which case, tautologically, there is a $2$-cell in~$N(v_0)$ whose projection intersects~$\phi(e_1)$.
	We can choose an integral lifting~$\gamma$ of~$\phi(\bd N(v_0))$ 
	such that every nonzero coefficient is $\pm 1$. This is easy as follows.
	We follow a line segment of~$\phi(\bd N(v_0))$ until we encounter a junction
	and we make an arbitrary choice what segment to continue on,
	and stop if we return to the starting position.
	Repeat by choosing a segment we have never visited.
	Since every junction joins an even number of segments,
	this procedure takes us through every segment exactly once.
	With such a lifting~$\gamma$,	we let $x$ traverse~$\phi(\bd N(v_0))$, and
	the signed number of times that $\ell(x)$ overlaps with~$\phi(e_1)$,
	where $\phi(e_1)$ just fixes a reference axis, 
	is the value of the contour integral above.
	Hence the mod~$2$ reduction of this count determines the external thorn on~$e_1$.
	
	An internal thorn on~$s_1 \in \Supp(b_2)$ where both vertices of~$\bd s_1$ 
	are in the interior of~$\Supp(b_2)$,
	is determined by~$N(\bd s_1)$
	that is the chain of all $2$-cells of~$b_2$ which meet either of the boundary vertices~$t_0,u_0$ of~$s_1$.
	We should count the $2$-cells of $N(\bd s_1)$ that intersect $s_1$ with their interior.
	There are two classes of such $2$-cells:
	those in~$N(\bd s_1)\setminus N(t_0)$ or in~$N(\bd s_1) \setminus N(u_0)$.
	There cannot be any contribution from~$N(t_0) \cap N(u_0)$ 
	because the $2$-cells in this intersection has~$s_1$ as a boundary $1$-cell.
	But, even if we count the number of times that~$\phi(s_1)$ overlaps with~$\ell(x)$ 
	where~$x$ runs over all $\bd N(t_0)$ and all $\bd N(u_0)$,
	any contribution from~$N(t_0) \cap N(u_0)$ is counted twice.
	Hence, we conclude that an internal thorn on~$s_1 \in \Supp(b_2)$ exists 
	if and only if
	the sum of the winding parities of $b_2$ at the two end points of~$s_1$ is odd.
	
	Therefore, if we multiply~$L(b_2)$ by the cycle-enforcing operators $\pi(v_0)$ 
	for all interior $0$-cells~$v_0$ at which the winding parity is odd,
	then we obtain an operator without any thorns around any interior $0$-cells of~$\Supp(b_2)$.
	See~\cref{fig:stringop}(a,b).
\end{proof}

\begin{corollary}\label{cor:localizing}
	Given a null-homologous $1$-cycle~$a_1$, 
	the operator~$S(a_1)$ obtained by canceling interior thorns of~$L(b_2)$ where~$\bd b_2 = a_1$
	is independent of the choice of~$b_2$ up to cycle-enforcing operators~$\pi$.
	In addition, $S(a_1)$ commutes with any other $S(a_1')$ where $a_1'$ is any null-homologous $1$-cycle.
\end{corollary}
We will refer to~$S(a_1)$ as the fermion string operator along~$a_1$.
\begin{proof}
	If~$b_2'$ is another $2$-chain such that~$\bd b_2' = a_1$,
	then $L(b_2')L(b_2) = L(b_2' + b_2)$ is a product of~$Z$ only;
	the sign is $+1$ by~\cref{lem:Frm}.
	By~\cref{lem:cancelthorns}, $L(b_2'+b_2)$ must localize around~$\bd (b_2' + b_2) = 0$,
	implying that $L(b_2' + b_2)$ is in fact a product of cycle-enforcing operators~$\pi$.
	If $S$ and $S'$ are two localized operators from $L(b_2)$ and $L(b_2')$, respectively,
	then $S' S^{-1}$ is an operator supported on the collection of $1$-cells that meet~$\Supp(a_1)$,
	and we have shown that $S' S^{-1} = (L(b_2')\prod \pi)(L(b_2) \prod \pi)$ is a product of~$\pi$'s.
\end{proof}

\begin{proposition}\label{prop:longS}
Assume that for any two $1$-cycles~$x_1,x_1'$ 
there is a cycle~$x_1''$ homologous to~$x_1'$
such that no $L_{Fc}(f_2)$ with a $2$-cell~$f_2$ intersects both $S(x_1)$ and $S(x_1'')$.
Then, for arbitrary $1$-cycle~$a$ there is an operator $S(a)$, unique up to a sign,
such that $a \mapsto S(a)$ modulo signs and cycle-enforcing operators~$\pi$
is a $\ZZ_2$-linear extension of~$S$ in~\cref{cor:localizing}.
Two operators~$S(a)$ and~$S(a')$ for any two $1$-cycles~$a,a'$ commute.
\end{proposition}
\begin{proof}
Take~$x_1 = x_1' = a_1$ in the assumption.
Then, by \cref{cor:localizing} 
the string operator $S(a_1 + x_1'')$ is supported on two clearly separated ``tubes,''
one along~$a_1$ and the other along~$x_1''$.
We \emph{define} that the string operator~$S(a_1)$ is the part along~$a_1$:
\begin{align}
	S(a_1 + x_1'') = S(a_1) S(x_1''). \label{eq:fermionstring}
\end{align}
This definition leaves the sign of $S(a_1)$ \emph{undetermined} 
if $a_1$ is homologically nontrivial,
but the operator content (Pauli factors) is unique up to cycle-enforcing terms.
Clearly, this definition is a $\ZZ_2$-linear extension of~$S$ in~\cref{cor:localizing}.

Being a product of worldline fluctuation and cycle-enforcing terms,
the operator $S(a_1 + x_1'')$ commutes with every worldline fluctuation term~$L_{Fc}$.
But any~$L_{Fc}$ can meet only one of~$S(a_1)$ and~$S(a_1')$ by assumption,
so every~$L_{Fc}$ commutes with $S(a_1)$.
Therefore, for any~$a_1$, the fermion string operator~$S(a_1)$ commutes with every~$L_{Fc}$.

If $a_1'$ is another cycle,
we have a cycle~$y_1'$ homologous to~$a_1'$ and separated from~$a_1$.
The product of two string operators~$S(a_1' + y_1') = S(a_1')S(y_1')$,
being a product of~$L_{Fc}$ and cycle-enforcing terms,
commutes with $S(a_1)$.
Since $S(y_1')$ is separated from~$S(a_1)$,
we know $S(y_1')$ commutes with $S(a_1)$.
Therefore, $S(a_1)$ commutes with $(S(a_1')S(y_1')) \cdot S(y_1')^{-1} = S(a_1')$.
\end{proof}

\subsection{$\overline\Frm$ for all $1$-cycles}

Now, we consider how to extend $\Frm$ beyond the space of null-homologous $1$-cycles.
We begin with an observation that
\cref{cor:localizing} makes it possible to define
$\BiFrm$ in~\cref{eq:association} more directly.
Namely, for two null-homologous $1$-cycles $a_1$ and $a_1'$,
the bilinear form~$\BiFrm(a_1,a_1')$ equals the mod~$2$ number of times 
that the thorns (the $Z$ part) of the fermion string operator~$S(a_1)$ intersect $a_1'$.
The ambiguity in the string operator by cycle-enforcing terms
does not affect the value~$\BiFrm(a_1,a_1')$ 
since the cycle-enforcing terms always overlap with a cycle an even number of times.
Since two closed fermion string operators commute, $\BiFrm$ is symmetric.

In~\cref{prop:longS} we have defined a fermion string operator for any $a_1$, null-homologous or not.
Define an extension~$\overline\BiFrm(a_1,a_1')$ for \emph{any} $1$-cycles~$a_1,a_1'$
to be the mod~$2$ number of times that
the $Z$ thorns of $S(a_1)$ intersect~$a_1'$.
Clearly, $\overline\BiFrm$ is linear in the second argument.
Since $S(a_1)$ and $S(a_1')$ are products of Pauli $X$ and $Z$ and they commute,
we must have $\overline\BiFrm(a_1,a_1') = \overline\BiFrm(a_1',a_1)$.
We have shown that $\overline\BiFrm$ is a symmetric bilinear form
on the $\ZZ_2$ vector space of all $1$-cycles of~$\mathcal K$,
extending~$\BiFrm$.

\begin{remark}\label{rem:BiFrmByExt}
Analogously to~\cref{rem:extphi},
$\overline\BiFrm(a_1,a_1')$ for $1$-cycles~$a_1,a_1'$ of $\mathcal K$,
can be evaluated by any extension~$(\mathcal K^e, \phi^e)$ obtained by adding more cells,
as long as the assumption of~\cref{prop:longS} is satisfied.
Here, an extension means that $\mathcal K \subseteq \mathcal K^e$ and $\phi^e |_{\mathcal K} = \phi$.
This is because the string operator~$S(a_1)$ can still be obtained 
by some $2$-chain of~$\mathcal K$ whose boundary contains~$a_1$.
Additional $1$-cells from the extension 
may give extra thorns to $S(a_1)$,
but those extra thorns do not contribute to~$\overline\BiFrm$ since $a_1'$ is on~$\mathcal K$.
$\lozenge$\end{remark}

\begin{proposition}\label{prop:globalFrm}
	Under the assumption of~\cref{prop:longS},
	there exists a quadratic form~$\overline \Frm$, 
	an extension of~$\Frm$ to all $1$-cycles with the associated biliear form~$\overline\BiFrm$.
	Any other extension is given by 
	$\overline \Frm + h^1$ for some $1$-cocycle~$h^1$ of $\mathcal K$;
	The set of all possible extensions is acted on transitively and freely
	by the first cohomology group of~$\mathcal K$ over~$\ZZ_2$.
\end{proposition}
\begin{proof}
	Choose a $\ZZ_2$ basis $[h_1(1)],\ldots,[h_1(n)]$ of the first homology of $\mathcal K$.
	Define inductively for all $m = 1,2,\ldots,n$
	\begin{align}
		\overline\Frm\left(z_1 + \sum_{j=1}^m c(j) h_1(j)\right) 
		= 
		\overline\Frm\left(z_1 + \sum_{j = 1}^{m-1}c(j) h_1(j) \right) 
			+
		c(m) \overline\BiFrm\left(z_1 + \sum_{j = 1}^{m-1}c(j) h_1(j), h_1(m)\right) 
\end{align}
	where $z_1 \in [0]$ and $c(j) \in \ZZ_2$.
	This is the unique way to define a quadratic form given that $\overline\Frm(h_1(j)) = 0$ for all~$j$.
	This shows the existence of an extension.	
		
	Any two extensions that share an associated bilinear form differ by a linear form~$\Delta$; 
	over $\ZZ_2$, the bilinear form measures the deviation of a quadratic form from being a linear form.
	This difference~$\Delta$ must be coclosed
	because on an exact $1$-cycle, we already have $\overline\Frm = \Frm$.
	If $\Delta$ is coexact, then $\Delta$ vanishes on all $1$-cycles.
	Therefore, the first cohomology acts transitively and freely 
	on the set of all possible extensions	with~$\overline\BiFrm$ fixed.
\end{proof}

\subsection{Fermionic charge only (Fc)}

\begin{remark}\label{rem:Fc}
Consider a wavefunction 
(on a system given a projection~$\phi$ for secondary qubits but without any primary qubits)
\begin{align}
	\ket{\text{Fc};[h_1]} = \sum_{a_1 \in [h_1]} (-1)^{\overline\Frm(a_1)} \ket{a_1}
\end{align}
where $a_1$ range over a homology class~$[h_1]$ of~$1$-cycles.
There is a commuting Pauli Hamiltonian
\begin{align}
	H_{Fc} = - \sum_{f_2 : \text{ cells}} L_{Fc}(f_2) - \sum_{v_0} \underbrace{\prod_{e_1 \in \db v_0} Z(e_1)}_{\pi(v_0)}
\end{align}
whose ground state subspace is spanned by $\{ \ket{\text{Fc};[h_1]} : h_1 \text{ is any $1$-cycle.}\}$.
Here, $L_{Fc}$ lacks the factors of~$Z(f^2)$ on primary qubits.  Since~$[h]$ is any homology class, the basis of the ground subspace of this Hamiltonian 
can be identified with the first homology group of~$\mathcal K$.
$\lozenge$\end{remark}

\begin{remark}\label{rem:FcEntRG}
The state~$\ket{\text{Fc}}_{\mathcal K}$ on a combinatorial manifold 
embeds naturally into a system of a refined triangulation.
By refinement we mean an injective linear chain map~$\iota: \mathcal K \to \mathcal K'$
such that $0$-cells are mapped to $0$-cells and $\iota \circ \bd = \bd \circ \iota$.
When the refined triangulation is uniform in the sense of~\cref{lem:diffphi},
we can choose any projection~$\phi': \mathcal K' \to \RR^2$,
and a handy choice of~$\phi'$ is one that has the same images for $0$-cells of $\mathcal K$.
Consider two string operators~$S(\iota \bd b_2)$ on~$\mathcal K'$ and~$S(\bd b_2)$ on~$\mathcal K$,
acting on the span of~$\{ \iota \bd b_2 ~:~ b_2 \text{ is a $2$-chain of }\mathcal K \}$.
\Cref{lem:cancelthorns} says that the only difference
between these two string operators is in the thorns around~$\iota \bd b_2$. 
Indeed, some $1$-cell~$e_1$ of~$\mathcal K$ may have been subdivided in~$\mathcal K'$ 
so that $\iota(e_1)$ consists of two or more $1$-cells of~$\mathcal K'$,
in which case there can be some thorns attached to~$e_1\setminus \bd e_1$.
But there is no worldline segment from the embedded cycles on these thorns.
Hence, any null-homologous $1$-cycles of~$\mathcal K$ embeds into~$\mathcal K'$
with frame parity unchanged.

The embedding of~$\ket{\text{Fc}}$ is realized by a shallow quantum circuit.
To this end, we aim to disentangle a qubit on a $1$-cell~$e_1'$ of~$\mathcal K'$ 
that is not in the image of~$\iota$.
Choose one~$f_2'$ of the $2$-cells in the coboundary of~$e_1'$.
There is a local unitary that conjugates $L_{Fc}(f_2')$ to $X(e_1')$.
We choose a composition~$U$ of Clifford gates, control-$X$ and control-$Z$,
where the control is on~$e_1'$ and 
the ``target'' ranges over all other $1$-cells in the support of~$L_{Fc}(f_2')$.
This local unitary~$U$ commutes with any product~$L_{Fc}(f_2') L_{Fc}(p_2')$ 
for any~$p_2'$ that intersects~$f_2'$ along~$e_1'$.
Also, $U$ conjugates the cycle-enforcing terms at the end points of~$e_1'$
to strip off~$Z(e_1')$.
Hence, the unitary~$U$ disentangles one qubit at $e_1'$ and maps~$\ket{\text{Fc}}_{\mathcal K'}$
to a state of form~$\ket{\text{Fc}}$ on a nonsimplicial complex that lacks~$e_1'$
but admits the prescription of~$L_{Fc}$ in~\cref{sec:fcbl}.
After disentangling all qubits on $1$-cells outside the image of~$\iota$,
we recover the state~$\ket{\text{Fc}}_{\mathcal K}$.

Therefore, we can say that
if $\mathcal K$ is the $2$-skeleton of a combinatorial $n$-manifold where $n \ge 3$,
the Hamiltonian~$H_{Fc}$ is a $\ZZ_2$-gauge theory with \emph{fermionic} charges.
We only have to check this for some cellulation, {\em e.g.}, the hypercubic lattice,
which we do explicitly in~\cref{fig:stringop}.
This justifies our naming, fermion string operators.
Furthermore, we may speak of \emph{the fermion} string operator~$S(a_1)$ 
for any null-homological $1$-cycle~$a_1$.

One might wonder why we need~$n \ge 3$.
Even if $n=2$, the Hamiltonian~$H_{Fc}$ is well defined and 
all the lemmas and propositions above remain true;
however, there are ``too few'' thorns in the product~$L(b_2)$ of~$L_{Fc}$ 
over a $2$-chain~$b_2$,
and the string operators transports a boson.
In fact, $H_{Fc}$ on a $2$-dimensional manifold is shallow-circuit-equivalent to
the toric code Hamiltonian on that manifold.
$\lozenge$\end{remark}

\begin{remark}\label{rem:tologicalchargeFc}
On a combinatorial $3$- or higher dimensional manifold with sufficiently refined triangulation,
there is only one deconfined topological charge with respect to~$H_{Fc}$.  We can see this as follows.
Since $H_{Fc}$ consists of commuting Pauli terms,
we may identify an excitation with a set of flipped terms.
Consider any excited state where only $L_{Fc}$ are flipped.
We know from~\cref{lem:Frm} that the product of~$L_{Fc}$ over the boundary of a $3$-ball
is $+\prod Z$ which must assume~$+1$ on~$\ket{\text{Fc}}$.
Hence, the excitation consisting of flipped~$L_{Fc}$'s 
corresponds to a $2$-cocycle~$b$, {\em i.e.}, a flux.
Disallowing extensive energy for an excitation,
we see that $b$ has to be exact,
and the excitation is caused by some $Z$s that form a $1$-cochain
whose coboundary is~$b$.
In particular, if the excitation is localized around a region 
in which no nontrivial second cohomology class may be supported,
then the excitation can be annihilated by $Z$s near the region.
A small neighborhood of a $1$-chain qualifies as such a region.
Now, more general excitation contains flips of the cycle-enforcing terms~$\pi$.
An even number of violated $\pi$ terms can be canceled by a $1$-chain of $X$ operators
whose boundary $0$-cells correspond the flipped $\pi$'s.
Therefore, any excitation localized around a point can be annihilated locally
if and only if there are an even number of flips of $\pi$.
Since a truncated fermion string operator creates an excitation with one flipped $\pi$,
the fermion at the string end is the only deconfined topological excitation of~$H_{Fc}$.
$\lozenge$\end{remark}

\begin{remark}\label{rem:fermionstringsegments}
In the setting of~\cref{rem:tologicalchargeFc},
we can assign to each $1$-cell a specific fermion string segment
that has only one $X$ factor on the $1$-cell and some $Z$ factors.
Any long fermion string operator will be a product of the operator segments.
An example set of such segments is depicted in~\cref{fig:stringop}(d).
A complete set of segments can always be found 
by (i) assigning a specific excitation that is a fermion to each $0$-cell
and (ii) letting segments transport the excitation from an end to the other end of the segment.
The latter~(ii) is possible because we have the unique deconfined topological charge in~$H_{Fc}$.
The former~(i) is achieved by choosing a base point and some fermion string operator
on a closed path from the base point to a given $0$-cell~$v$ and back.
If the path bounds a $2$-chain~$b_2$ that has no interior points, 
{\em i.e.}, every $0$-cell of $\Supp(b_2)$ belongs to $\Supp(\bd b_2)$,
which may happen if the path almost self-intersects,
then it may be tricky to truncate the path so that at the end point there is a localized nontrivial excitation.
To avoid such degenerate paths,
we require that the path is sufficiently non-self-intersecting
that
\emph{if a plaquette term $L_{Fc}$ overlaps with the path,
then it may do so only along one $1$-cell.}
This requirement is fulfilled for a sufficiently refined triangulation.
For such a non-self-intersecting path, 
a truncation of string operator at~$v$ gives a fermion there.
The excitation at the base point is specified 
after all other $0$-cells are assigned with a fermion,
by letting one of other $0$-cells be a new base point.
$\lozenge$\end{remark}

\bibliography{afq-ref}
\end{document}